\documentclass[runningheads]{llncs}
\usepackage[T1]{fontenc}
\usepackage{amsmath,amssymb}
\usepackage{hyperref}
\usepackage{booktabs}
\usepackage{graphicx}
\usepackage{subcaption}
\usepackage{cite}
\usepackage{colortbl} 
\usepackage[table]{xcolor} 
\usepackage{pgf} 
\usepackage{graphicx}
\usepackage{multirow}
\newcommand{\qedsymbol}{\hfill \ensuremath{\square}}

\begin{document}
\title{Mining Power Destruction Attacks in the Presence of Petty-Compliant Mining Pools}
\titlerunning{Mining Power Destruction Attacks}
\author{Roozbeh Sarenche \and
Svetla Nikova\and
Bart Preneel}
\institute{COSIC, KU Leuven, Belgium
\email{\{roozbeh.sarenche,svetla.nikova,bart.preneel\}@esat.kuleuven.be}}
%
%
\maketitle              
\begin{abstract}
Bitcoin's security relies on its Proof-of-Work consensus, where miners solve puzzles to propose blocks. The puzzle's difficulty is set by the difficulty adjustment mechanism (DAM), based on the network's available mining power. Attacks that destroy some portion of mining power can exploit the DAM to lower difficulty, making such attacks profitable. In this paper, we analyze three types of mining power destruction attacks in the presence of petty-compliant mining pools: selfish mining, bribery, and mining power distraction attacks.
We analyze selfish mining while accounting for the distribution of mining power among pools, a factor often overlooked in the literature. Our findings indicate that selfish mining can be more destructive when the non-adversarial mining share is well distributed among pools. We also introduce a novel bribery attack, where the adversarial pool bribes petty-compliant pools to orphan others’ blocks. For small pools, we demonstrate that the bribery attack can dominate strategies like selfish mining or undercutting.
Lastly, we present the mining distraction attack, where the adversarial pool incentivizes petty-compliant pools to abandon Bitcoin’s puzzle and mine for a simpler puzzle, thus wasting some part of their mining power. Similar to the previous attacks, this attack can lower the mining difficulty, but with the difference that it does not generate any evidence of mining power destruction, such as orphan blocks. 

\end{abstract}
\section{Introduction}
Bitcoin~\cite{nakamoto2008bitcoin} was the starting point of the exciting journey into the blockchain and cryptocurrency era. Despite the introduction of many new cryptocurrencies, Bitcoin has maintained the highest market capitalization~\cite{Bitcoin_market} and remains the most famous cryptocurrency. Any serious attack that threatens Bitcoin's progress not only harms Bitcoin users but also impacts users of other cryptocurrencies and blockchain platforms. Although Bitcoin has operated smoothly since its emergence, the absence of a serious attack thus far does not guarantee that its security will not be compromised in the future. Since the introduction of Bitcoin, numerous research papers~\cite{garay2015bitcoin, garay2017bitcoin, badertscher2018but} have analyzed its security and explored potential attacks~\cite{eyal2018majority, sapirshtein2017optimal, bag2016bitcoin, liao2017incentivizing} that could target it. Identifying potential flaws in Bitcoin's underlying mechanism and studying possible solutions can not only benefit Bitcoin's progress but also serve as inspiration for new cryptocurrencies.

Several attacks have been introduced in the literature aimed at destroying the efforts of other miners in producing valid blocks, which we refer to as \emph{mining power destruction attacks}. Examples include selfish mining~\cite{eyal2018majority, sapirshtein2017optimal}, block withholding~\cite{rosenfeld2011analysis,eyal2015miner}, power adjusting withholding and bribery selfish mining~\cite{gao2019power}, block denial of service (BDoS)~\cite{mirkin2020bdos}, undercutting~\cite{carlsten2016instability}, eclipse attack~\cite{nayak2016stubborn}, Pitchforks~\cite{judmayer2018pitchforks}, and script puzzle distraction attack~\cite{teutsch2016cryptocurrencies}. While these attacks differ in their specific methods, they all share the common goal of destroying the mining power of other miners. The intuition behind the profitability of these attacks is the principle that mining rewards are distributed among the participating mining powers in proportion to their contributions to the chain extension. By destroying another miner's share, the adversary increases its own relative contributions, which consequently boosts its revenue.
No mining power destruction attacks can be more profitable than following the honest strategy during the initial difficulty epoch~\cite{grunspan2018profitability, sarenche2024time}. The adversary must wait for the subsequent difficulty adjustment mechanism to observe the impact of these attacks on its revenue. In contrast to attacks like double spending, where the adversary receives revenue immediately upon a successful attack, mining power destruction attacks must be continued for at least two weeks to become profitable.

In this paper, we analyze mining power destruction attacks in Bitcoin within the context of petty-compliant mining pools, referred to as a semi-rational setting. The analysis of these attacks in a semi-rational setting differs in two important ways from their analysis in the presence of altruistic mining pools. First, in semi-rational settings, we cannot naively assume that mining pools will always follow honest behavior, especially when they are the victims of an attack. As petty-compliant pools, they may choose to deviate from the honest strategy to defend against mining power destruction attacks and protect their mining efforts. Second, an adversarial mining pool in such settings may carry out incentive manipulation attacks~\cite{judmayer2021sok}, such as bribery, to convince non-victim petty-compliant mining pools to adopt its desired strategy, thereby increasing the success likelihood for its mining power destruction attacks. This potential for incentive manipulation allows adversarial mining pools even with a limited share of mining power to effectively attack the network. In this paper, we present a selfish mining analysis from a new perspective and introduce two novel mining power destruction attacks: the bribery and the mining power distraction attacks.

\textbf{Selfish mining attack (Section~\ref{sec: selfish_mining}):} 
Selfish mining, introduced by Eyal et al.~\cite{eyal2018majority}, is a well-studied attack in the context of Bitcoin. Numerous papers~\cite{eyal2018majority,sapirshtein2017optimal,grunspan2018profitability,zur2020efficient,bar2022werlman,sarenche2024deep} have analyzed the impact of selfish mining on the adversarial block ratio and profitability. However, most of these studies assume that, apart from the adversarial mining nodes, all the remaining mining nodes in the network follow the honest fork choice rule, meaning they always mine on the longest chain or the earliest of the longest forks. Carlsten et al.~\cite{carlsten2016instability} introduced the concept of petty-compliant mining nodes, highlighting that, in a fork race, these mining nodes choose the fork that offers the highest return.\footnote{\cite{carlsten2016instability} uses the concept of petty-compliant nodes in the undercutting attack analysis, overlooking them in the selfish mining analysis.} To the best of our knowledge, the first analysis of selfish mining in the presence of petty-compliant mining nodes was conducted by the authors of~\cite{bar2023deep}. The authors in~\cite{bar2023deep} showed that in a semi-rational setting, the mining share threshold for a profitable deviation from the honest strategy is reduced.

Despite considering petty-compliant mining nodes, the analysis in~\cite{bar2023deep} is based on an implicit assumption that each mining node controls only an infinitesimal share of the total mining power, thereby overlooking the presence of any petty-compliant \emph{mining pools}. A mining pool can be viewed as a group of infinitesimal mining nodes working together.
Assuming no petty-compliant mining pools exist, the analysis can simply presume that, during a fork race, if the adversary places a bribe on top of the adversarial fork, almost all the network's mining power will mine on top of the adversarial fork. This is because the mining nodes that do not adopt the adversarial fork are only those that own a block in the non-adversarial fork, with a total mining share that is infinitesimal.
However, when considering the existence of mining pools, the presence of a bribe offered by the adversary does not guarantee that all the network's mining power will mine on top of the adversarial fork. If a mining pool has already mined a block in the non-adversarial fork, it will continue to mine on top of that fork with its entire mining power, which can no longer be assumed to be infinitesimal.
Figure~\ref{fig:selfish_mining_setings} intuitively illustrates the mining share distributions between forks in a fork race, in the altruistic setting, in the presence of infinitesimally small petty-compliant mining nodes, and in the presence of petty-compliant mining pools.

In this paper, we analyze selfish mining in the presence of petty-compliant mining pools and examine the effect of mining power decentralization on selfish mining profitability. Our theoretical and Markov Decision Process (MDP)-based analysis shows that selfish mining is more destructive in settings where the non-adversarial mining pools are more decentralized. This suggests that the greater the gap in mining power share between the adversarial mining pool and the other pools, the more profitable selfish mining becomes.
\begin{figure*}[t] 
    \centering
    \begin{subfigure}[t]{0.26\textwidth} 
        \centering
        \includegraphics[width=\linewidth]{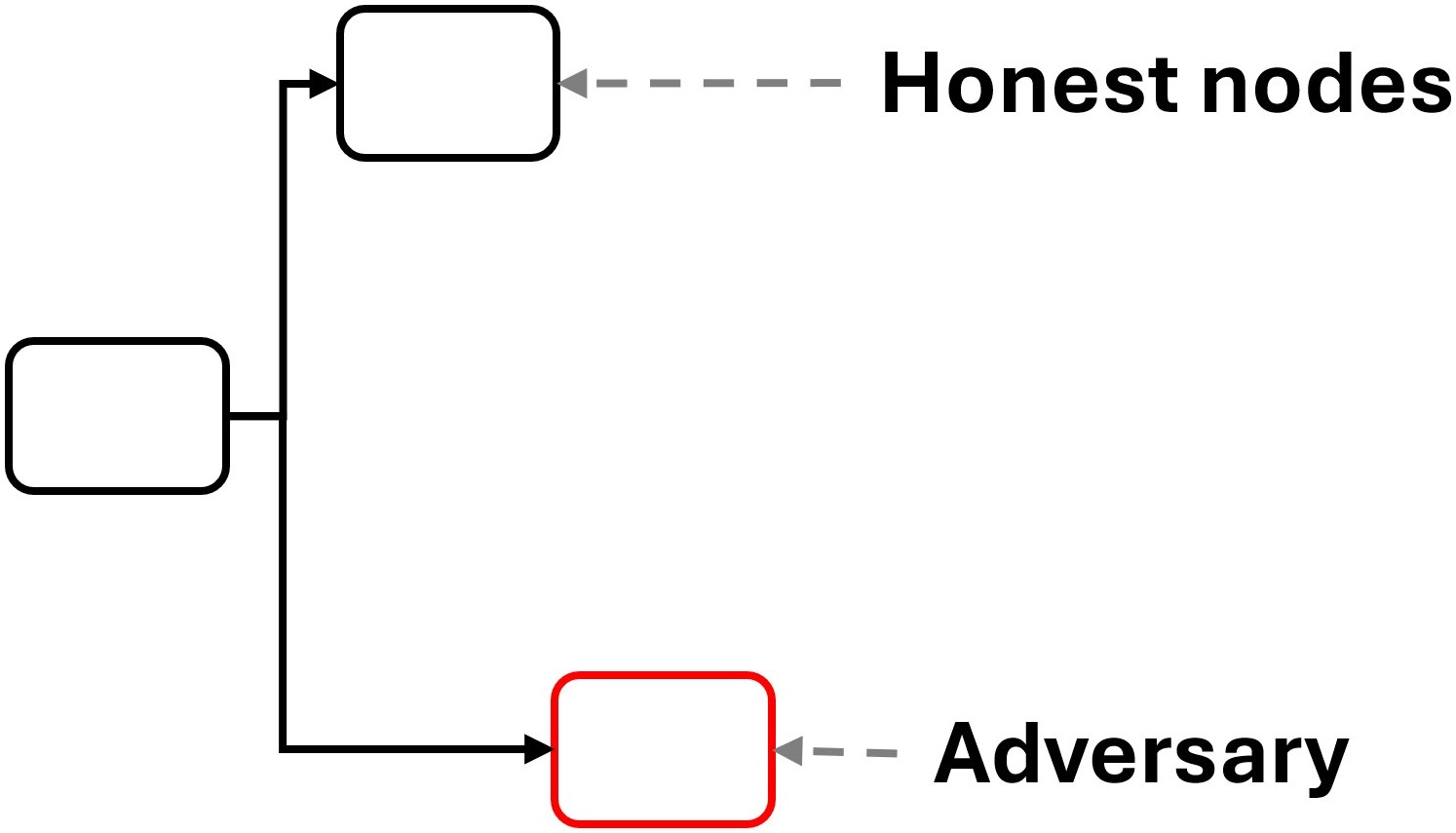} 
        \caption{Altruistic setting: $\alpha_\mathcal{A}^\texttt{fork} = \alpha_\mathcal{A}$.}
        \label{fig:largest_selfish}
    \end{subfigure}%
    \hfill 
    \begin{subfigure}[t]{0.37\textwidth} 
        \centering
        \includegraphics[width=\linewidth]{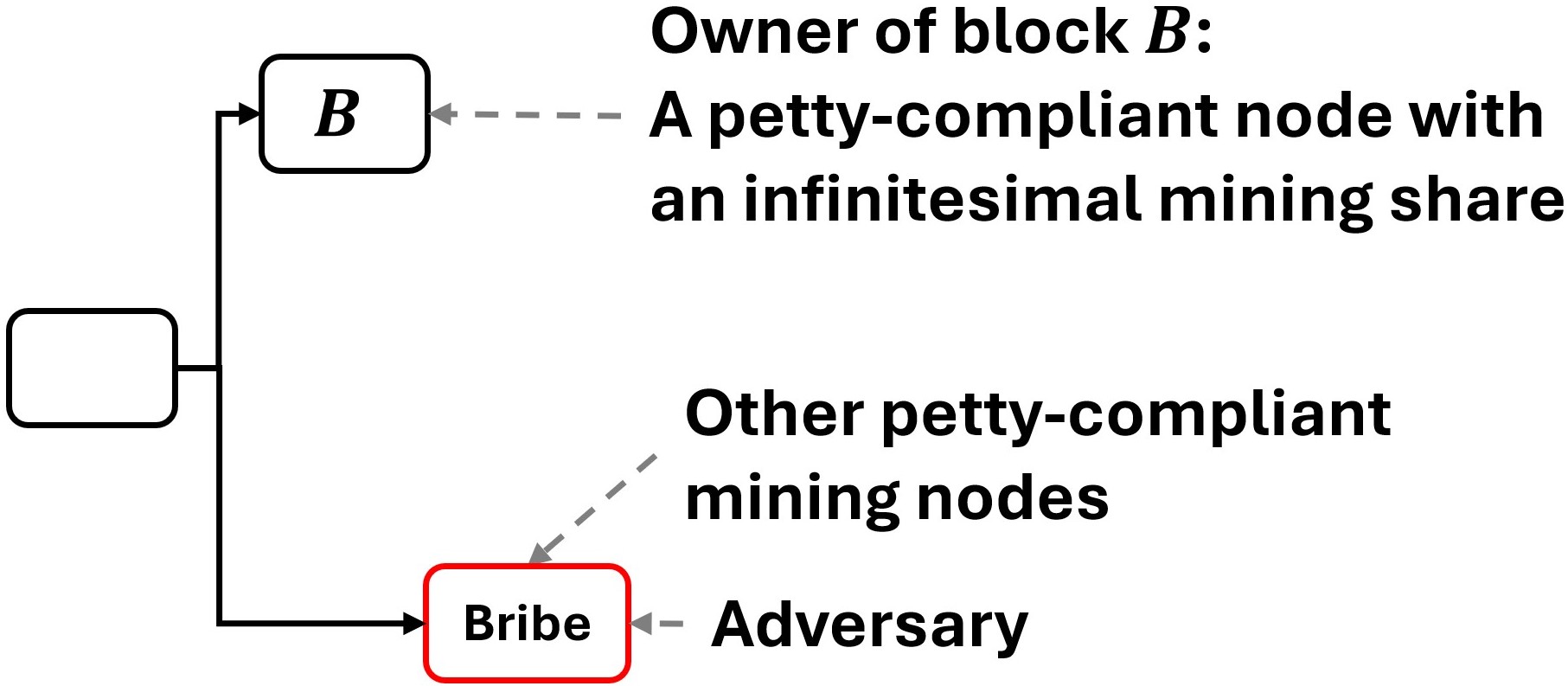} 
        \caption{Semi-rational setting with infinitesimal mining nodes: \\ $\alpha_\mathcal{A}^\texttt{fork} \approx 1$.}
        \label{fig:min_mining_share}
    \end{subfigure}%
    \hfill 
    \begin{subfigure}[t]{0.34\textwidth} 
        \centering
        \includegraphics[width=\linewidth]{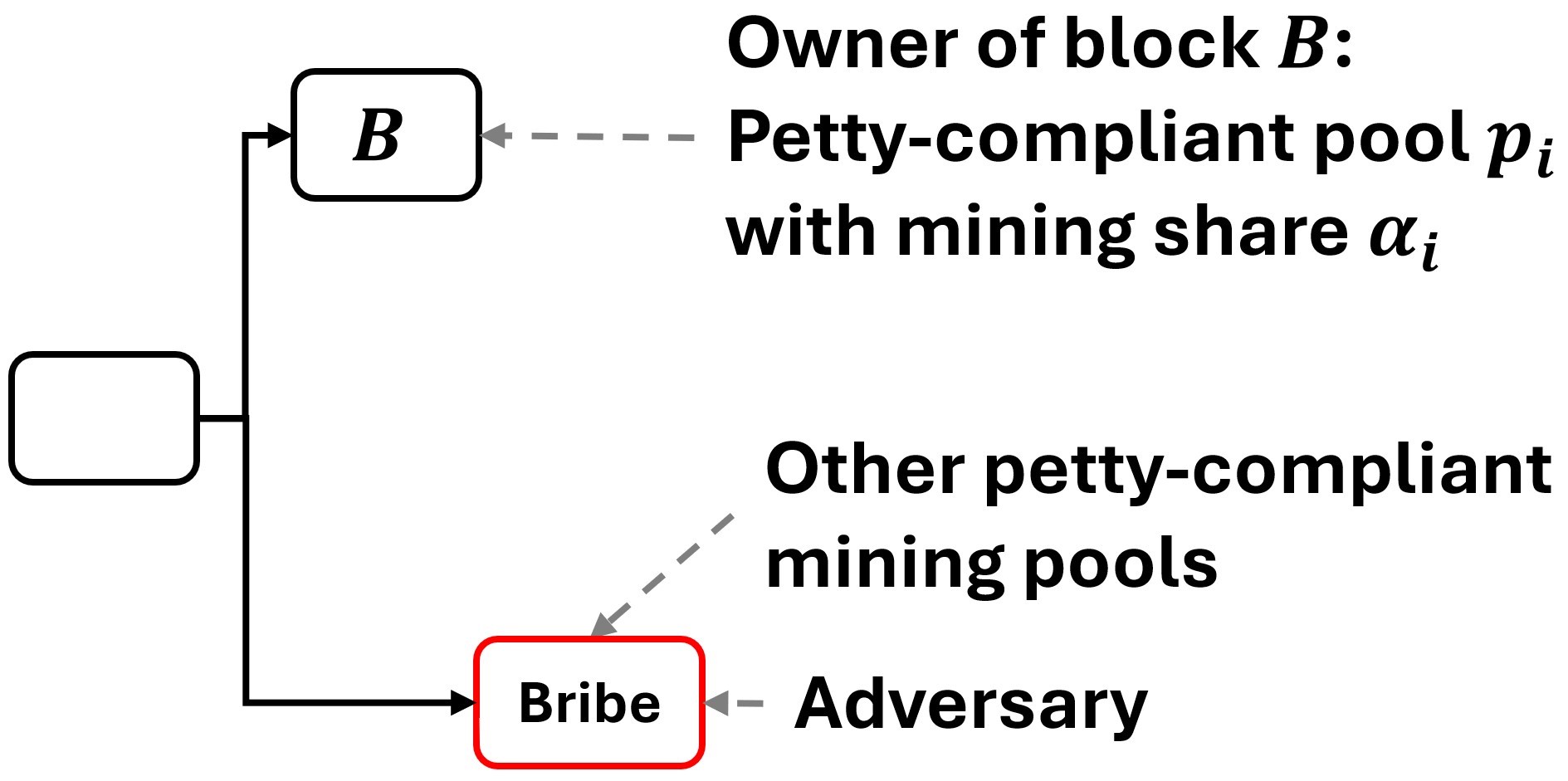} 
        \caption{Semi-rational setting with mining pools: $\alpha_\mathcal{A}^\texttt{fork} = 1 - \alpha_i$.}
        \label{fig:largest_selfish_2}
    \end{subfigure}%
    \caption{Mining share distribution between forks in a block race. The adversarial block, denoted in red, is published later than the rival block in the block race. We denote by $\alpha_\mathcal{A}$, $\alpha_i$, and $\alpha_\mathcal{A}^\texttt{fork}$ the mining shares of the adversarial mining pool, mining pool $p_i$, and all mining nodes extending the adversarial fork, respectively.}
    \label{fig:selfish_mining_setings}
    \vspace{-10 pt}
\end{figure*}

\textbf{Bribery attack (Section~\ref{sec: bribery}):}
In the blockchain literature, various bribery attacks have been introduced for different purposes. These attacks can range from attempts to censor a single transaction, as seen in~\cite{nadahalli2021timelocked, tsabary2021mad}, to efforts aimed at rewriting the history of blockchain blocks to facilitate a successful double-spending attack~\cite{judmayer2019pay, liao2017incentivizing, bonneau2016buy, mccorry2019smart}. The former requires a relatively small bribe to incentivize miners not to include a specific transaction, while the latter necessitates a substantial budget to persuade miners to abandon the longest chain and mine on top of a block that is deep within the chain.
In this paper, we introduce a novel bribery attack aimed at destroying mining hash power. In our attack, the adversarial mining pool bribes petty-compliant mining pools to orphan the blocks of other mining pools, with the sole intention of wasting a portion of their mining power. Orphaning these blocks can reduce mining difficulty, resulting in revenue that compensates for the bribe paid. Our bribery attack is similar to selfish mining, except that instead of gambling on its own block to orphan another, the adversary pays a bribe less than the block reward to orphan a block. The advantage of the bribery attack over attacks such as selfish mining and feather forking is that it is risk-free for the adversarial pool, meaning the attacker does not risk losing its block or bribing budget. Additionally, a mining pool with any arbitrary share can always exploit other pools with a lesser share, making the attack profitable even for small mining pools.

\textbf{Distraction attack (Appendix~\ref{sec:distraction}):}
As discussed in the literature~\cite{sarenche2024time, grunspan2018profitability}, modifying the Bitcoin Difficulty Adjustment Mechanism (DAM) to adjust difficulty based on the total active mining power—both wasted and effective—can mitigate mining power destruction attacks that generate valid proof of mining power destruction. For instance, selfish mining and the introduced bribery attack become unprofitable under such a modified DAM, as the orphan blocks generated during these attacks serve as evidence of mining power destruction.
However, there are other types of mining strategies and attacks discussed in the literature, such as smart mining~\cite{goren2019mind}, coin hopping~\cite{kwon2019bitcoin, ilie2021unstable, grunspan2023profit}, and distraction attacks~\cite{teutsch2016cryptocurrencies}, which can destroy mining power without leaving evidence of destruction. In smart mining, a miner alternates between being idle and mining honestly, while in coin hopping, the miner switches between mining on different networks. Both strategies aim to manipulate and reduce the mining difficulty level.
In distraction attacks, the adversary incentivizes miners to stop mining the Bitcoin puzzle and instead mine an alternative one. An example of a distraction attack is the script puzzle~\cite{teutsch2016cryptocurrencies}, where the adversary bribes miners through a smart contract-based puzzle to divert them from the Bitcoin chain. The goal of the script puzzle attack is to facilitate a double-spending attack or gain majority control of the network, which requires both a significant share of mining power and a substantial bribe (equivalent to six Bitcoin block rewards) for successful execution. 
These strategies and attacks can target Bitcoin without leaving any evidence, making them difficult to mitigate without relying on external trusted platforms.

In this paper, we introduce a novel distraction attack that aims to destroy a portion of mining power without leaving any evidence of power destruction. In our attack, the adversarial mining pool publishes a Proof-of-Work (PoW) puzzle with lower difficulty on another platform and incentivizes petty-compliant mining pools to mine the lower-difficulty puzzle. This attack can be carried out with minimal bribes, less than the value of a block reward.

\section{Preliminaries and System Model} \label{sec:system model}
In our system model, time is divided into smaller units referred to as rounds. We use $\lambda$ to denote the block mining rate of the network, representing the number of blocks mined per unit of time. 
We denote by $\texttt{Rev}_i(t; \pi)$ and $\texttt{Cost}_i(t; \pi)$ the revenue and the cost of mining pool $p_i$ in round $t$ under strategy $\pi$, respectively. If all the mining pools follow the honest strategy, the average per-round revenue of the mining pool $p_i$ with mining share $\alpha_i$ is equal to $\alpha_i \cdot \lambda R$, where $R$ denotes the block reward. If mining pool $p_i$ mines with its whole mining power, its average mining cost per round is equal to $\alpha_i \cdot c_i$, where $c_i$ denotes the average normalized mining cost of pool $p_i$ per round. 
\begin{definition}[Time-averaged profit] \label{def_profit}
    The time-averaged profit (per-round profit) of pool $p_i$ following strategy $\pi$ is defined as follows:
    \begin{equation}
        \texttt{Profit}_i^{t}(\pi) = {\frac{\sum_{t'=0}^{t-1} {\Big(\texttt{Rev}_i(t'; \pi)-\texttt{Cost}_i(t'; \pi)}\Big)}{t}},
        \;
        \texttt{Profit}_{i}(\pi) =\lim_{t\to\infty} {\texttt{Profit}_{i}^{t}}(\pi) \enspace.
    \end{equation}
\end{definition}

We define honest, petty-compliant, and adversarial mining pools as follows.
\begin{definition}[Honest mining pool]
    An honest mining pool is defined as a mining pool that i) always chooses the longest chain available in its view (in case of a tie, it chooses the block that was seen first) as its canonical chain to mine on top of, and ii) once it mines a new block, it immediately publishes the block to all other mining pools.
\end{definition}
To define petty-compliant mining pools, we first define the chain expected return and semi-rational fork choice rule.
\begin{definition}[Chain expected return]
    Let $\pi_C$ and $r_C$ denote the strategy of mining on top of a given chain $C$, and its expected return, respectively. The expected return $r_C$ is defined as follows:
    \begin{equation}
        r_C = \sum_{t=0}^{\infty} \gamma^{t} \texttt{Rev}_i(t; \pi_C) \enspace ,
    \end{equation}
    where $\gamma$ is a decaying factor set by each mining pool.\footnote{The presence of $\gamma$ in the expected return definition makes it subjective to each mining pool's view of profitability. This reflects reality, as some pools prefer immediate rewards (low $\gamma$), while others prioritize long-term profit (with $\gamma$ close to 1). In our paper analysis, we assume $\gamma$ is close to 1 and that petty-compliant pools do not account for the impact of difficulty adjustment on the chain expected return.}
\end{definition}
For a given chain $C$, we denote by $|C|$ the length of the chain.
\begin{definition}[Semi-rational fork choice rule]
    Let $C^\texttt{long}$ denote the longest chain available in the view of petty-compliant mining pool $p_i$ (in case there are multiple chains of equal length, consider the one that was seen first by the mining pool). Also. let $\texttt{CH}^D$ denote the set of all the chains $C$ in the view of petty-compliant mining pool $p_i$ that satisfy $|C^\texttt{long}| - |C| \le D$. The $(\epsilon, D)$-semi-rational fork choice rule to select the canonical chain is defined as follows:
    If there is a chain $C^* \in \texttt{CH}^D$ that satisfies the following conditions:
    \begin{enumerate}
        \item $r_{C^*} \ge r_{C}$ for all $C \in \texttt{CH}^D$, and
        \item $r_{C^*} - r_{C^\texttt{long}} > \epsilon \alpha_i R$,
    \end{enumerate}
    then $C^*$ is the canonical chain. Otherwise, $C^\texttt{long}$ is the canonical chain.
\end{definition}
We refer to $\epsilon$ as the incentivizing factor, representing the threshold of normalized loss a mining pool tolerates before deviating from the honest strategy. A discussion of the parameters $\epsilon$ and $D$, as well as their role in the semi-rational fork choice definition is provided in Appendix~\ref{appendix_incentivization_factor}.

\begin{definition}[Petty-compliant mining pool]
    An $(\epsilon, D)$-petty-compliant mining pool is defined as a mining pool that i) follows the $(\epsilon,D)$-semi-rational fork choice rule to select the canonical chain to mine on top of, and ii) once it mines a new block, it immediately publishes the block to all other mining pools unless it is incentivized not to do so.
\end{definition}
In our paper, petty-compliant pools are restricted to specific strategies defined for each attack and cannot act arbitrarily.
\begin{definition}[Adversarial mining pool]
   The adversarial mining pool may arbitrarily deviate from the honest strategy (for example, by delaying the publication of its blocks) or execute an incentive manipulation attack to induce petty-compliant mining pools to deviate from the honest strategy.
\end{definition}

\begin{definition}[Semi-rational environment]
    An $(\epsilon, D)$-semi-rational environment is an environment in which any non-adversarial mining pool is an $(\epsilon', D)$-petty-compliant mining pool, where $\epsilon' \leq \epsilon$. If $D = \infty$, we simply denote the environment as the $\epsilon$-semi-rational environment.
\end{definition}

\noindent\textbf{System model} We assume our model operates within an $(\epsilon,D)$-semi-rational environment. The system comprises a set of petty-compliant mining pools, denoted by $p_i$ for $i$ in ${1,2, \cdots, N}$, alongside an adversarial mining pool denoted by $p_\mathcal{A}$. We denote by $\alpha_i$ and $\alpha_\mathcal{A}$ the mining power share of the $i^\text{th}$ petty-compliant mining pool and the adversarial mining pool, respectively, where $\alpha_\mathcal{A} + \sum_{i=1}^{N}{\alpha_i} = 1$. Our model assumes a fixed block reward of $R$; however, the adversarial mining pool may also place a bribe on top of any block. A discussion of the limitations of our system model is provided in Appendix~\ref{appendix:model_limitations}.

\begin{definition}[Normalized bribe] \label{Def: norm_bribe}
    The normalized bribe is defined as the amount of bribe divided by the block reward $R$.
\end{definition}
According to Definition~\ref{Def: norm_bribe}, a normalized bribe of $\texttt{br}$ implies that the total amount of the bribe is equal to $\texttt{br}R$.

\section{Selfish Mining Attack} \label{sec: selfish_mining}
In this section, we analyze selfish mining in the presence of petty-compliant mining pools. 
Selfish mining results in a fork race between the adversarial and non-adversarial forks, with only one eventually being included in the canonical chain. 
The presence of petty-compliant mining pools offers some benefits to an adversarial pool. During a fork race, the adversary can bribe these pools to mine on the adversarial fork, boosting its chances of winning. 
The adversary can use in-band methods like whale transactions~\cite{liao2017incentivizing} (discussed in Appendix~\ref{appendix: whale}) or out-of-band methods like smart contracts~\cite{velner2017smart, mccorry2019smart} to bribe miners. It may also leave transaction fees in mempool as bribes for others~\cite{bar2023deep, carlsten2016instability}. 

Analyzing selfish mining in the presence of petty-compliant mining pools requires revisiting the selfish mining analysis typically applied to altruistic settings or semi-rational settings with infinitely many infinitesimal mining nodes. Under the altruistic assumption, the adversarial mining pool is guaranteed to win the fork race if its fork is longer than the competing fork. Additionally, if the adversarial pool propagates its fork faster, it increases its chances of winning a same-height fork race. In semi-rational settings with infinitely many infinitesimal mining nodes, the analysis may assume that the adversary can incentivize \textbf{all} petty-compliant nodes to abandon the non-adversarial fork with a bribe, as the mining share of infinitesimal miners in the non-adversarial fork is negligible.

However, the scenarios that can actually occur when selfish mining takes place in practice, particularly in the presence of petty-compliant mining pools, may differ significantly from the analyses conducted under the simplified assumptions of altruistic and infinitesimal miners. In the presence of petty-compliant mining pools, the adversary is not necessarily guaranteed to win the fork race solely based on the length of its fork, the speed at which it propagates its fork, or even the minimal bribe placed on its fork. Petty-compliant mining pools select the fork to mine based on its expected return, where the longest, fastest-propagated, or bribed chain may not always offer the highest expected return. For instance, consider a fork race between the non-adversarial fork with length $l_{\overline{\mathcal{A}}}$ and the adversarial fork with length $l_\mathcal{A} > l_{\overline{\mathcal{A}}}$. If the petty-compliant mining pool $p_i$ has $n>0$ blocks in the non-adversarial fork, although the non-adversarial fork is shorter than the adversarial fork, $p_i$ may still be incentivized to continue mining on top of the non-adversarial fork to revive its $n$ blocks included in it. This example suggests that in the presence of petty-compliant mining pools, the strategy employed by the selfish pool and its profitability differ from those in altruistic settings or in environments with infinitely many infinitesimal nodes.

\subsection{Theoretical Analysis}
To analyze the selfish mining attack in the presence of petty-compliant mining pools, we first introduce metrics to assess the distribution of mining power among these pools.
\begin{definition}[Centralization factor]
    Let $\mathcal{P} = \{p_1, p_2, \cdots, p_N\}$ denote the set of all mining pools available in the environment. Also, let $\alpha_{p_i}$ denote the corresponding mining power share of mining pool $p_i \in \mathcal{P}$. The centralization factor, which is denoted by $\beta$, is defined as follows:
    \begin{equation}
        \beta = \sum_{p_i \in \mathcal{P}} {\alpha_{p_i}}^2 \enspace .
    \end{equation}
\end{definition}
The centralization factor $\beta$ can take on values in the range of $(0,1)$, with higher $\beta$ values indicating a more centralized network. In a fully decentralized network, $\beta$ approaches $0$. Within a network where the maximum mining share of its mining pools is denoted by $\alpha$, the centralization factor $\beta$ is less than or equal to $\alpha$, with the upper bound case of $\beta = \alpha$ occurring when all mining pools possess a mining share exactly equal to $\alpha$.

\begin{definition}[Residual centralization factor, Pool advantage] \label{def:residual_CF}
    Let $\mathcal{P}_{\overline{i}}= \{p_1, p_2, \cdots, p_N\} \setminus \{p_i\}$ denote the set of all mining pools excluding mining pool $p_i$ whose mining power share is denoted by $\alpha_{p_i}$. Also, let $\alpha_{p_j}$ denote the corresponding mining power share of mining pool $p_j \in \mathcal{P}_{\overline{i}}$. The residual centralization factor w.r.t. mining pool $p_i$, which is denoted by $\beta_i$, is defined as follows:
    \begin{equation}
        \beta_i = \frac{\sum_{p_j \in \mathcal{P}_{\overline{i}}} {\alpha_{p_j}}^2}{1-\alpha_{p_i}} \enspace .
    \end{equation}
    The mining advantage of pool $p_i$ is defined to be $1-\beta_i$.
\end{definition}
In the following, we demonstrate that a mining pool with a lower residual centralization factor (i.e., higher mining advantage) has a higher chance of conducting a successful selfish mining attack. To get the intuition of why pool advantage is defined as above, we first review the following lemma.
\begin{lemma}\label{lemma: selfish_mining_advantage}
    Consider a fork race within an $\epsilon$-semi-rational environment, where the length of both semi-rational and adversarial forks is equal to 1, and a normalized bribe\footnote{The normalized bribe is defined as the amount of bribe divided by the block reward.} of $\texttt{br}=\epsilon$ is available on top of the adversarial fork. The probability of the event that the next block is mined on top of the adversarial fork is equal to the mining advantage of the adversarial mining pool.
\end{lemma}
The proof of Lemma~\ref{lemma: selfish_mining_advantage} is presented in Appendix~\ref{appendix: selfish_mining_advantage}.
In Definition~\ref{def:strategy} presented in Appendix~\ref{appendix_slefish_mining_strategy}, we introduce a simple selfish mining strategy $\pi^\texttt{selfish}$ suitable for an $(\epsilon, D=1)$-semi-rational environment. As already discussed, strategies designed for an altruistic environment cannot be easily applied in a semi-rational setting. The outcome of each action that the adversarial mining pool takes in the semi-rational environment must be evaluated from the perspective of all petty-compliant mining pools. The following theorem examines the effect of the adversarial centralization factor on the profitability of strategy $\pi^\texttt{selfish}$.

\begin{theorem} \label{theorem_D=inf_selfish}
    Assume an $(\epsilon, D=1)$-semi-rational environment\footnote{$D=1$ implies that the best chain in a mining pool view is a chain that is at most one block shorter than the longest chain.} in which the mining share of each pool, excluding the adversarial pool, is less than $0.4302$. In this environment, the selfish mining strategy $\pi^\texttt{selfish}$ for an adversarial mining pool with mining share $\alpha_\mathcal{A}$ and residual centralization factor $\beta_\mathcal{A}$ can dominate the honest mining if the following inequality holds:
    \begin{equation} \label{eq:selfish_profitability}
        \beta_\mathcal{A} < \frac{\alpha_\mathcal{A}- \epsilon (1-\alpha_\mathcal{A})^2}{(1-\alpha_\mathcal{A})(1-\epsilon)} \enspace .
    \end{equation}
\end{theorem}
The proof of Theorem~\ref{theorem_D=inf_selfish} is presented in Appendix~\ref{proof:theorem_D=inf_selfish}. The assumption that the mining share of each non-adversarial mining pool is less than $0.4302$ ensures that if an $(\epsilon, D=1)$-petty-compliant mining pool has a single block in a fork that lags behind the longest fork by one block, it is incentivized to abandon its fork and adopt the longest fork (Lemma~\ref{lemma_0.43}).

Note that strategy $\pi^\texttt{selfish}$ is not the optimal selfish mining strategy that an adversarial mining pool can follow in an $(\epsilon, D=1)$-semi-rational environment. 
The main goal of Theorem~\ref{theorem_D=inf_selfish} is to show that a mining pool with a lower residual centralization factor has a higher chance of successfully executing a selfish mining attack. This implies that the profitability of selfish mining in a semi-rational environment depends not only on the mining share of the adversarial pool but also on the distribution of mining power among the remaining mining pools.
\begin{corollary} \label{corollary:p_1andp_2}
    Let $P_1$ and $P_2$ denote the first and second mining pools with the highest mining share in the network, whose corresponding mining share are denoted by $\alpha_1$ and $\alpha_2$, respectively. Then, we can obtain the following statements for an $\epsilon$-semi-rational environment:
    \begin{itemize}
        \item For $\epsilon>0$, selfish mining dominate honest mining for mining pool $P_1$ if $\frac{\alpha_1}{1-\alpha_1}>\alpha_2 + \epsilon (1-\alpha_1-\alpha_2)$.
        \item For $\epsilon=0$, selfish mining always dominate honest mining for mining pool $P_1$.
    \end{itemize}
\end{corollary}
The proof of Corollary~\ref{corollary:p_1andp_2} is presented in Appendix~\ref{appendix: p_1andp_2}.
According to Corollary~\ref{corollary:p_1andp_2}, selfish mining is the dominant strategy for the largest mining pool if the network incentivizing threshold $\epsilon$ is less than the difference between the mining shares of the largest and the second largest mining pools. An important consideration in the literature on selfish mining is the minimum threshold of mining power required for a profitable attack. Perhaps the most well-known example is that, in an altruistic setting, a miner with a normal communication capability needs at least $25\%$ of the network's mining power to successfully execute selfish mining~\cite{eyal2018majority}. These thresholds might suggest that if no mining pool holds a share above the threshold, selfish mining cannot threaten the network. However, Theorem~\ref{theorem_D=inf_selfish} and Corollary~\ref{corollary:p_1andp_2} show that in practice, where multiple mining pools aim to maximize their payoffs, mining share alone is not the only factor determining the threat of selfish mining. The distribution of mining power among pools and the gap between their shares also play a critical role. Even if all mining pools hold less than these well-known thresholds, a pool with a sufficiently large gap between its share and that of others can still find selfish mining profitable.
\begin{figure*}[t!] 
    \centering
    \begin{subfigure}[t]{0.49\textwidth} 
        \centering
        \includegraphics[width=\linewidth]{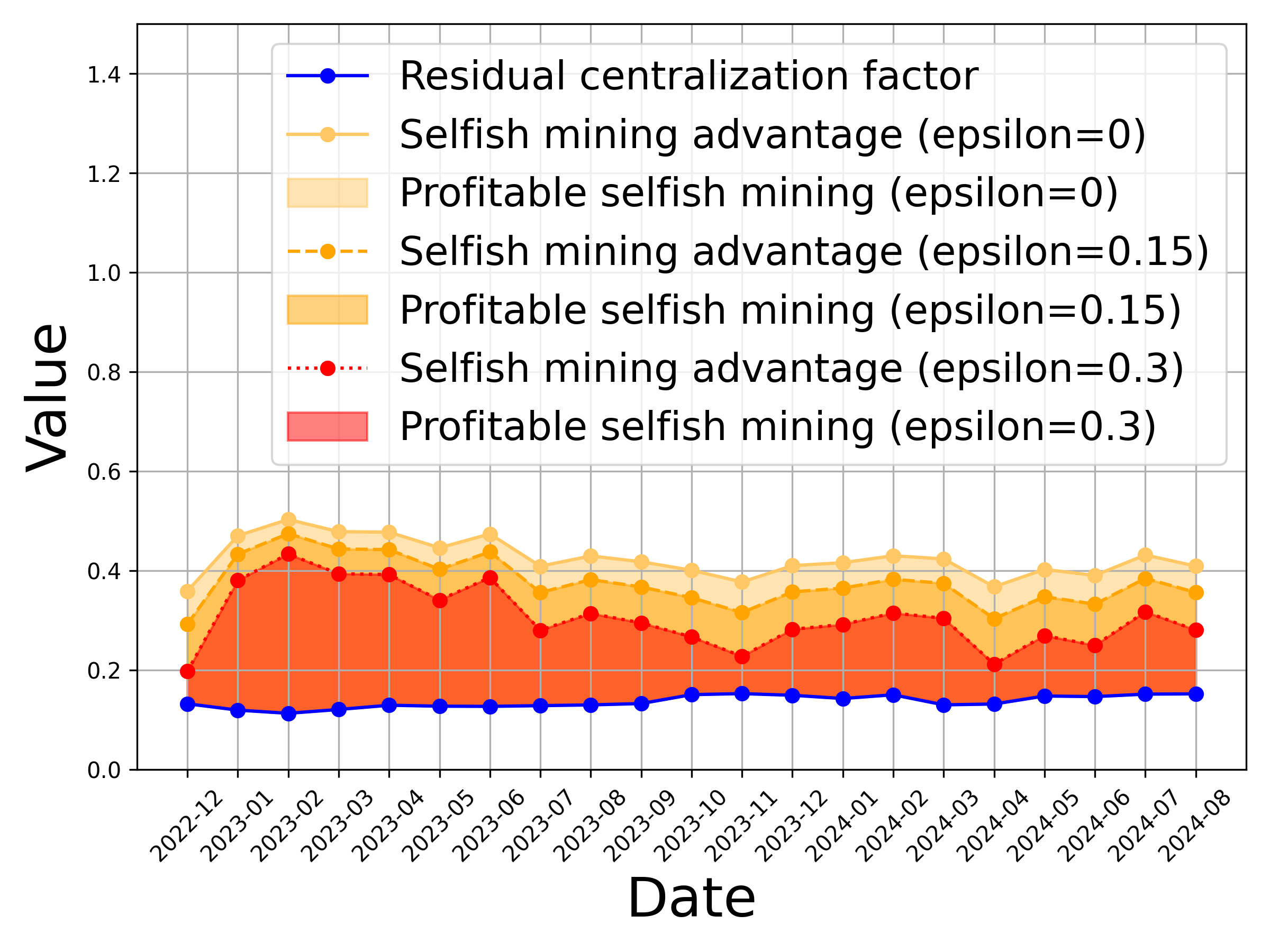}
        \caption{Largest mining pool}
        \label{fig:largest_pool_selfish_mining_profitable_period}
    \end{subfigure}%
    \hfill 
    \begin{subfigure}[t]{0.49\textwidth}
        \centering
        \includegraphics[width=\linewidth]{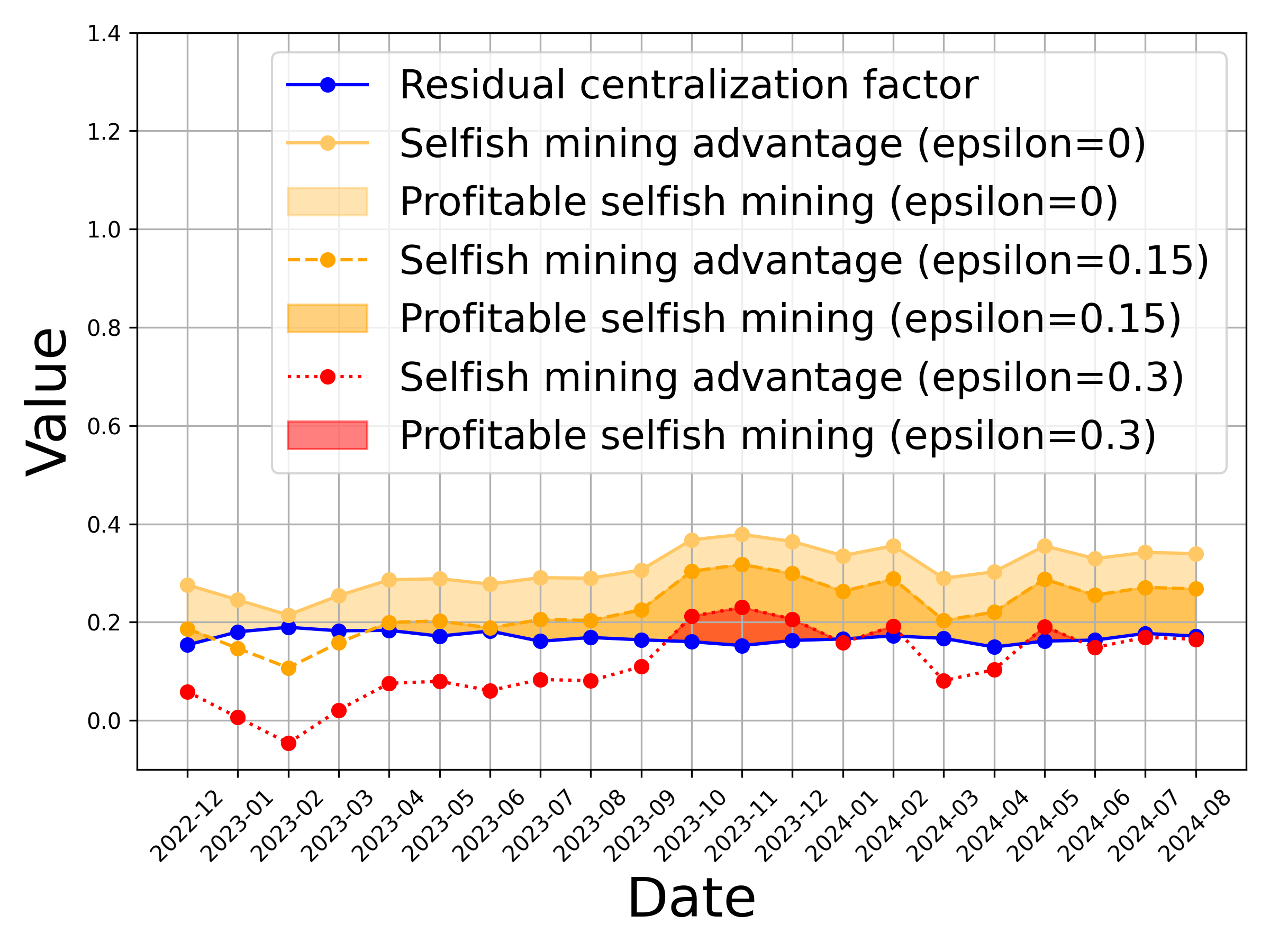} 
        \caption{Second largest mining pool}
        \label{fig:second_largest_pool_selfish_mining_profitable_period}
    \end{subfigure}%
    \caption{Periods of profitable selfish mining. According to Theorem~\ref{theorem_D=inf_selfish}, selfish mining becomes more profitable than honest mining when the selfish mining advantage of a mining pool (defined as the right-hand side of inequality~\ref{eq:selfish_profitability}) surpasses its residual centralization factor (defined in Definition~\ref{def:residual_CF}).}
    \label{fig:selfish_mining_profitable_period}
\end{figure*}

We define the $\epsilon$-selfish mining advantage of the adversary as the right-hand side of inequality~\ref{eq:selfish_profitability}. According to Theorem~\ref{theorem_D=inf_selfish}, selfish mining dominates honest mining if the adversary's $\epsilon$-selfish mining advantage exceeds its residual centralization factor. The mining power distribution for the first 8 months of 2024 is presented in Table~\ref{tab:MiningPowerDistribution} in Appendix~\ref{appendix_mining_pools}~\cite{Pools}. Figure~\ref{fig:selfish_mining_profitable_period} depicts both the centralization factor and the $\epsilon$-selfish mining advantage for the largest and second-largest mining pools over the months from December 2022 to August 2024. 
To generate Figure~\ref{fig:selfish_mining_profitable_period}, we used data on the number of blocks mined by the respective pools during these months to calculate their corresponding mining shares for each month. This figure illustrates in which months selfish mining could be considered profitable based on different values of the network incentivizing factor. 
As shown in Figure~\ref{fig:largest_pool_selfish_mining_profitable_period}, selfish mining is always the dominant strategy for the largest pool, even with a high incentivizing factor of $\epsilon=0.3$. If we assume $\epsilon=0$, as can be seen in Figure~\ref{fig:second_largest_pool_selfish_mining_profitable_period}, selfish mining is also the dominant strategy for the second-largest pool. However, as the incentivizing factor increases, some months are excluded from the profitable selfish mining period.

In Appendix~\ref{sec:MDP}, we present our MDP-based implementation, available at~\cite{Our_MDP}, to derive the optimal selfish mining strategy in the semi-rational setting. In Appendix~\ref{appendix:selfish_mining_time_cost}, we discuss the duration of the initial financial loss associated with selfish mining and the cost incurred by the adversarial mining pool during this period.

\section{Bribery Attack} \label{sec: bribery}
Based on the selfish mining results from our MDP implementation presented in Table~\ref{tab:selfish_mining_MDP_result_real}, the profitability of selfish mining drastically decreases as a pool's mining share decreases. The main reason for the marginal profitability of selfish mining for smaller pools is that, in selfish mining, an adversary is gambling on its own blocks, with no guarantee of winning the fork race. In other words, when an adversary withholds a block to initiate a fork race, it either collects the block reward $R$ or loses it entirely. For larger mining pools, the probability of winning the fork race is high, making selfish mining profitable. However, smaller pools face a significantly higher risk of losing the block reward in the fork race. As a result, small mining pools prefer not to initiate a fork race unless they know the competing mining pool in the fork race is also small.

Knowing that selfish mining is a marginally-profitable mining power destruction attack for smaller pools, in this section, we introduce a bribery attack that similar to selfish mining aims to waste a portion of the network's hash power. However, unlike selfish mining, this bribery attack is risk-free for the adversary and is not limited to times when the adversary has mined a block, thereby allowing the adversary to increase its profitability. In this attack, an adversarial mining pool pays a bribe to any other mining pool that can orphan a specific block targeted by the adversarial pool. In other words, the adversarial pool incentivizes petty-compliant mining pools to orphan a block mined by another mining pool. 
The profitability of the bribery attack is based on the same principle as selfish mining. In both of these attacks, the adversary aims to waste part of the network's hash power to exploit the difficulty adjustment mechanism, thereby reducing mining difficulty. The key difference is that in selfish mining, the adversarial pool uses its own blocks to carry out the attack, whereas in the bribery attack, it uses its budget to incentivize other mining pools to conduct the attack. The bribery attack offers two key advantages over selfish mining. First, it is risk-free for the adversary, as the bribe is only paid to collaborating pools upon successfully orphaning a block. Second, it can be executed by pools with small mining shares, enabling them to exploit weaker pools.


\begin{definition}[Target and rival blocks]
    Let $B_1$ represent the head of the canonical chain that is mined on top of block $B_0$. Additionally, suppose an adversarial mining pool offers a bribe for publishing a block $B'_1$ on top of $B_0$, where $B'_1$ differs from $B_1$. In this case, $B_1$ and $B'_1$ are referred to as the target block and the rival block of a bribery attack, respectively.
\end{definition}
When an adversarial mining pool targets a block $B_1$ for a bribery attack, it is not guaranteed that a rival block will always be mined for the target block $B_1$. If no rival block is mined, the attack fails. However, if a rival block $B'_1$ is mined for the target block $B_1$, the attack succeeds. In a successful bribery attack, both the target and rival blocks undergo a fork race. Note that the success of this bribery attack is independent of the outcome of the fork race, as neither of the blocks involved in the fork race is adversarial. Since only one block between the target and rival blocks can be included in the canonical chain, the orphaning of one block is definite, implying that the occurrence of the fork race is enough to consider the attack successful.
The following theorem determines the maximum bribe the adversary can spend on the bribery attack while keeping it profitable.
\begin{theorem} \label{theorem_bribery attack_long_term}
    Consider an adversarial mining pool with a mining share $\alpha_\mathcal{A}$ that allocates a normalized bribe $\texttt{br}$ for each target block $B$, which is payable only upon the successful mining of a rival block for block $B$. Assume \(k\) is the average number of target blocks per epoch for which a rival block is mined under the bribery attack. The time-averaged profit of the adversarial pool under the bribery attack exceeds the honest mining profit as long as \(\texttt{br} < \alpha_\mathcal{A}\), for any value of $k$.
\end{theorem}
The proof of Theorem~\ref{theorem_bribery attack_long_term} is presented in Appendix~\ref{appendix:bribery_theorem_proof}. 
The idea behind the proof of Theorem~\ref{theorem_bribery attack_long_term} is to demonstrate that for each non-adversarial block that gets orphaned, the adversarial mining pool receives an additional revenue of $\alpha_\mathcal{A}R$ on average. Therefore, if the bribe spent on orphaning each non-adversarial block is less than $\alpha_\mathcal{A}R$, the adversarial mining pool can still earn a profit.

The question that arises is how the adversary can effectively use this bribe budget to incentivize other petty-compliant pools to mine a rival block for a target block. In the following section, we introduce a bribery attack and analyze it in the setting where all mining pools are aware of which pool has mined the target block $B$. Once a mining pool mines a block, it can reveal its identity in the coinbase transaction. According to the statistical information presented in~\cite{Pools}, the ratio of unknown blocks in the first 8 months of 2024 is less than $8\%$, indicating that the majority of mining pools reveal their identity in their blocks. In Appendix~\ref{appendix:bribery_unknown}, we analyze the bribery attack in the setting where the miner of the target block is unknown.
\subsection{Bribery Attack Under the Assumption of Known Miners} \label{section: bribe_attack_known_miners}
In this section, we present a smart-contract-based bribery attack under the assumption that the miners of target blocks are known.
Appendix~\ref{appendix:whale_based_known_miners} discusses a bribery attack using whale transactions.
\vspace{-10 pt}
\subsubsection{Description of Bribery Attack Using Smart Contracts}
Let $\alpha_\mathcal{A}$ represent the mining share of the adversarial mining pool $p_\mathcal{A}$. Assume block $B_1$ denotes the head of the canonical chain and is mined by mining pool $p_i$ with mining share $\alpha_i$. From the perspective of the adversarial mining pool, block $B_1$ can be considered a valid target block for the smart contract-based bribery attack if the following conditions hold: 1) $B_1$ is a non-adversarial block, and 2) $\alpha_\mathcal{A} > \alpha_i + 2\epsilon$.
If these conditions are not satisfied, the attack is not applicable, and the adversarial mining pool waits for the next block. If block $B_1$ is a valid target block, the adversarial mining pool proceeds with the bribery attack.
Let $B_0$ denote the parent of the target block $B_1$. The adversarial mining pool deploys a smart contract and publishes it for all the mining pools. The smart contract stores three parameters: the hash of block $B_0$ denoted by $\texttt{Hash}(B_0)$, the hash of block $B_1$ denoted by $\texttt{Hash}(B_1)$, and a difficulty target denoted by $\texttt{Target}$ that is equal to the difficulty target of the epoch to which block $B_1$ belongs. The adversary deposits 2 normalized bribes $\texttt{br}_1 = \alpha_i + \epsilon$ and $\texttt{br}_2 = \epsilon$ in the smart contract. Anyone who submits a witness that proves the mining of a rival block $B_\texttt{rival}$ for block $B_1$ can withdraw $\texttt{br}_1$. A valid witness includes a Bitcoin $\texttt{nonce}$, a Merkle root $\texttt{MR}$, a Bitcoin coinbase transaction $tx$, a Merkle inclusion $\texttt{proof}$, and a payout address $\texttt{add}$ in the host blockchain, and satisfies the conditions:
\begin{enumerate}
    \item Demonstrates valid evidence of power destruction\footnote{Verifying the transactions in the rival block is unnecessary. Any user who solves a new PoW puzzle (not necessarily a valid Bitcoin block) deserves the bribe.}:
    \begin{itemize}
        \item A valid rival block is mined on top of block $B_0$: \\$\texttt{Hash}(B_\texttt{rival}) = \texttt{Hash}\big( \texttt{Hash}(B_0), \texttt{MR}, \texttt{nonce}\big) \leq \texttt{Target}$.
        \item The rival block differs from the target block: $\texttt{Hash}(B_\texttt{rival}) \neq \texttt{Hash}(B_1)$.
    \end{itemize}
    
    \item Includes a payout address that belongs to the compliant miner:
    \begin{itemize}
        \item The coinbase transaction is included in the list of block transactions, i.e., a valid Merkle inclusion $\texttt{proof}$ is available that proves coinbase transaction $tx$ is a leaf of a Merkle tree with Merkle root $\texttt{MR}$.
        \item The payout address $\texttt{add}$ is included in the coinbase transaction.
    \end{itemize}
\end{enumerate}
If the witness transaction satisfies the conditions above, bribe $\texttt{br}_1$ will be transferred from the smart contract deposit to the payout address $\texttt{add}$, and the rival block hash $\texttt{Hash}(B_{\text{rival}})$ gets stored in the contract as the hash of the rival block. Anyone who submits a witness that proves the mining of a valid supporting block on top of the rival block $B_{\text{rival}}$ can withdraw $\texttt{br}_2$. The witness verification is similar to the previous case with the difference that the hash of the parent block should be equal to $\texttt{Hash}(B_{\text{rival}})$. Note that this smart contract should have a time lock, after the expiration of which the adversarial mining pool is able to withdraw the deposits. This is to prevent any loss in the case of the attack's failure.
During this attack, the adversarial mining pool follows the following strategy: if no rival block is published, it mines on top of the target block. Once a rival block is published, it switches to mining on top of the rival block.
\begin{theorem} \label{theorem_bribery_smart}
    In an $\epsilon$-semi-rational environment, an adversarial mining pool with mining share $\alpha_\mathcal{A}$ is incentivized to conduct a smart contract-based bribery attack on any target block $B$ if block $B$ is mined by a mining pool with mining share $\alpha_i$, where $\alpha_\mathcal{A} > \alpha_i + 2\epsilon$.
\end{theorem}
The proof of Theorem~\ref{theorem_bribery_smart} is presented in Appendix~\ref{appendix:prrof_theorem_bribery_smart}. The proof demonstrates that the adversarial mining pool must pay a normalized bribe of $\alpha_i + 2\epsilon$ to incentivize petty-compliant pools to mine a rival block for a target block mined by a pool with share $\alpha_i$. Each bribery attack orphans one non-adversarial block, granting the adversary additional normalized revenue of $\alpha_\mathcal{A}$. Thus, if this additional revenue $\alpha_\mathcal{A}$ exceeds the bribe $\alpha_i + 2\epsilon$, the attack becomes profitable.

According to Theorem~\ref{theorem_bribery_smart}, assuming $\epsilon = 0$, an adversarial mining pool can perform a bribery attack on blocks mined by any mining pool whose mining power is less than the adversarial pool's mining share. 
The bribery attack is similar to the undercutting attack introduced in~\cite{carlsten2016instability}. The key difference is that in undercutting, the adversarial mining pool uses its own block to undercut a target block. However, in the bribery attack, the adversarial mining pool incentivizes other mining pools to undercut a target block. Note that the undercutting attack is a subset of the possible actions considered in our selfish mining MDP analyzer introduced in Section~\ref{sec:MDP}. The Markov chain analysis of both the introduced bribery attack and the undercutting attack in the presence of petty-compliant pools is presented in Appendix~\ref{appendix:markov_chain_bribery} and Appendix~\ref{appendix:markov_chain_undercut}, respectively.

Figure~\ref{fig:bribery_selfish_undercut_pecentage_increase}, illustrates the percentage increase in reward share for the smart contract-based bribery attack, the selfish mining attack, and the undercutting attack, calculated for the 8 largest mining pools, based on the mining distribution in Table~\ref{tab:MiningPowerDistribution} and incentivizing factor $\epsilon = 0$. Figure~\ref{fig:bribery_selfish_undercut_infinitesimal} depicts the reward share that an adversarial mining pool can achieve in the presence of infinitesimal petty-compliant mining nodes.
As can be seen in both Figures~\ref{fig:bribery_selfish_undercut_pecentage_increase} and~\ref{fig:bribery_selfish_undercut_infinitesimal}, the bribery attack can dominate both selfish mining and undercutting in an $\epsilon$-semi-rational environment when the adversarial mining pool's share is small.
The higher profitability of the bribery attack can be attributed to two main factors. First, in the bribery attack, the adversarial mining pool does not risk losing its blocks. Second, the bribery attack is not limited to the times when the adversarial mining pool has mined a block. In the undercutting attack, once the tip of the blockchain is a valid target block, the adversarial pool must mine the next block to successfully undercut the target block. This probability is relatively low, especially if the adversarial mining pool is small. In contrast, in the bribery attack, the adversarial mining pool can incentivize all mining pools, except the target block miner, to undercut the target block.
\begin{figure*}[t!] 
    \centering
    \begin{subfigure}[t]{0.49\textwidth} 
        \centering
        \includegraphics[width=\linewidth]{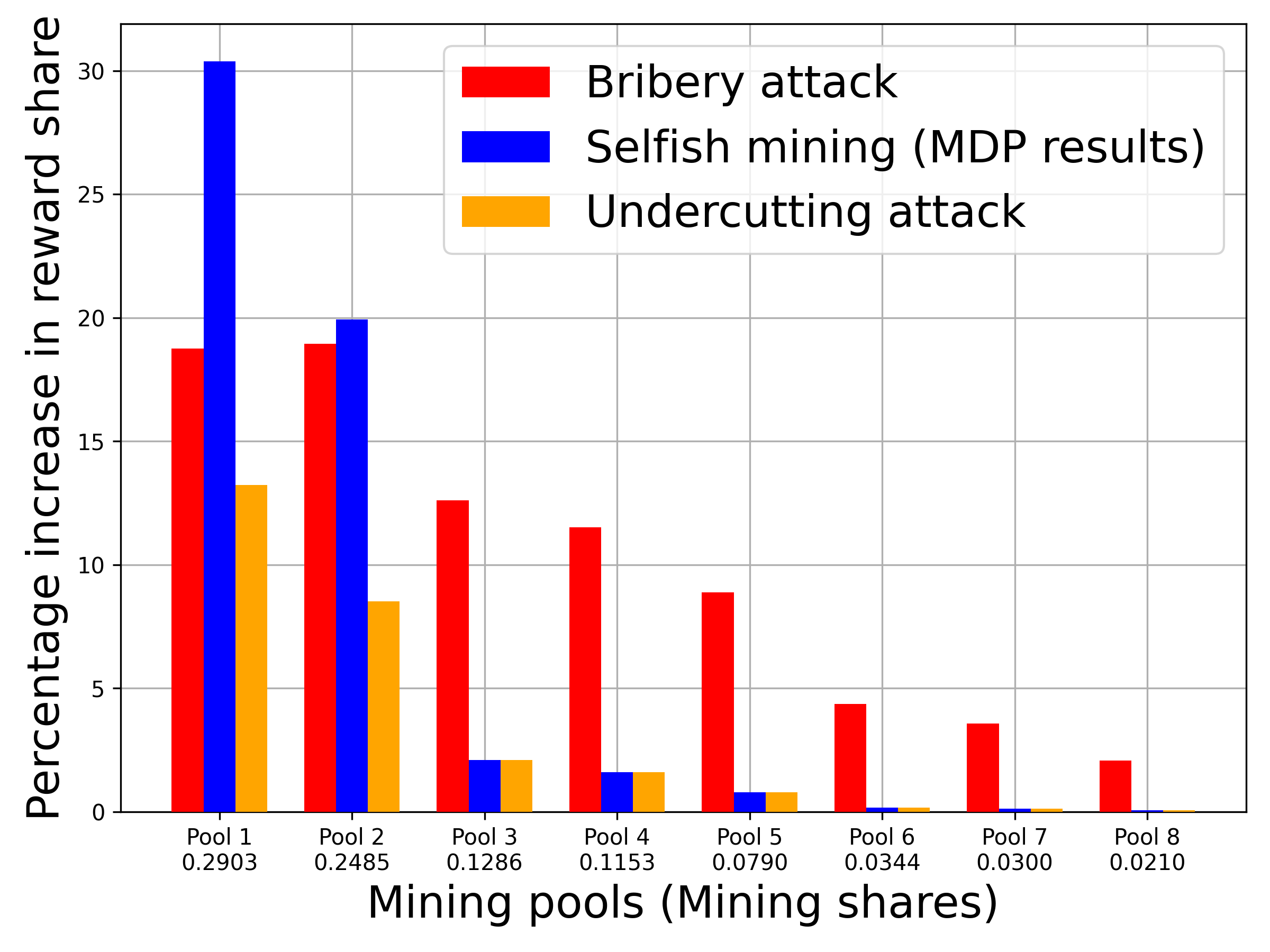}
        \caption{Petty-compliant mining pools (real-world mining power distribution).}
        \label{fig:bribery_selfish_undercut_pecentage_increase}
    \end{subfigure}%
    \hfill 
    \begin{subfigure}[t]{0.49\textwidth}
        \centering
        \includegraphics[width=\linewidth]{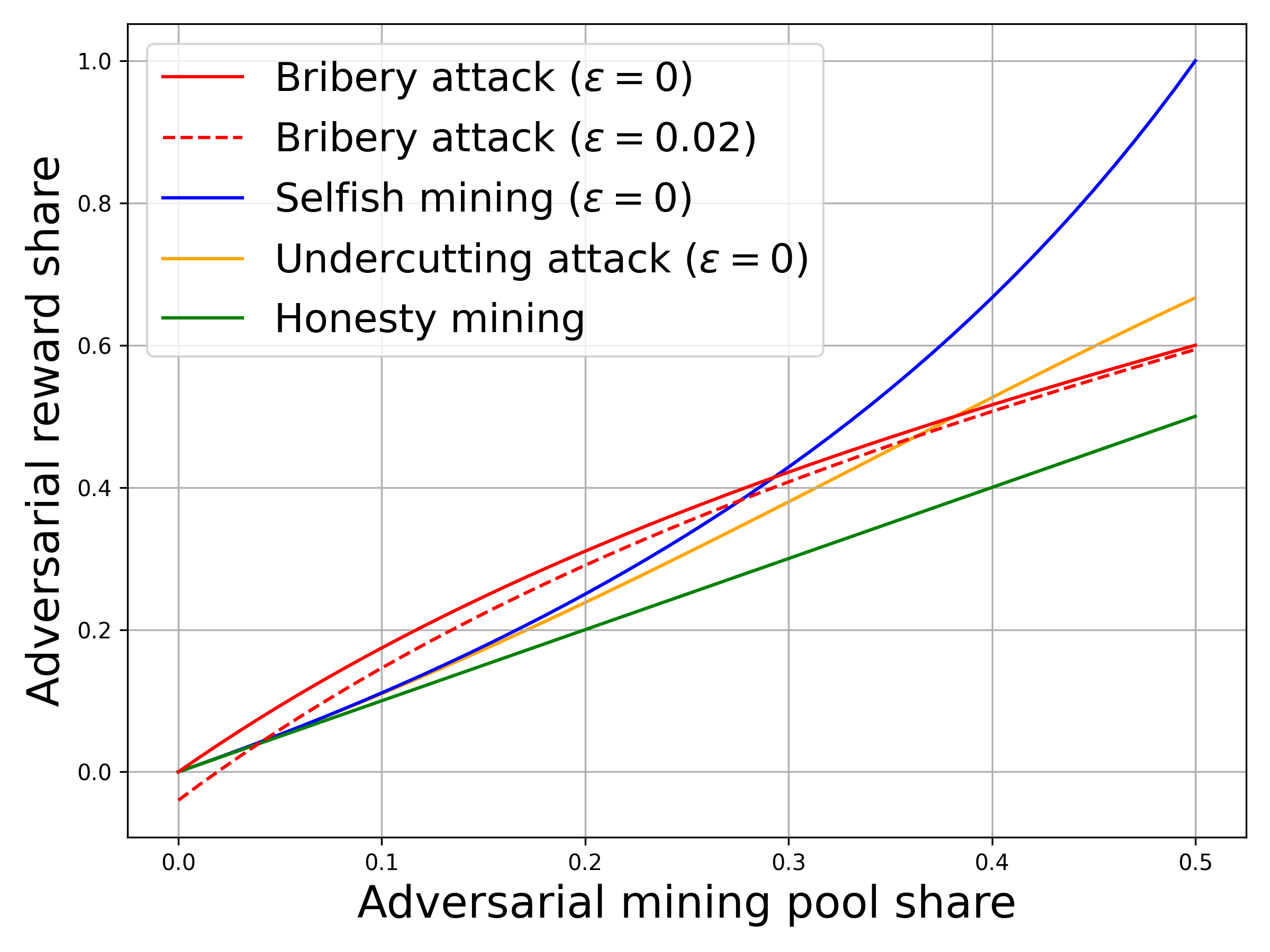} 
        \caption{A single adversarial pool and infinitesimal petty-compliant nodes.}
        \label{fig:bribery_selfish_undercut_infinitesimal}
    \end{subfigure}%
    \caption{The adversarial reward share obtained from different attacks.}
    \vspace{ -10 pt}
\end{figure*}
Compared to other bribery attacks in the literature~\cite{judmayer2019pay, liao2017incentivizing, bonneau2016buy, mccorry2019smart}, where double-spent transactions compensate the adversary, in the introduced bribery attack, the protocol itself compensates the adversary, eliminating the need for performing deep block reorganizations. Additionally, the budget for each instance of the introduced bribery attack is only a fraction of a block reward, significantly less than the double-spending bribes, which require several block rewards.

\textbf{Duration and cost of the initial loss period:} Mining destruction attacks, including the bribery attack, incur an initial period of financial loss before the mining difficulty adjusts. During the first epoch of the bribery attack, the adversary must pay bribes without any immediate increase in revenue per unit of time. This implies that the adversarial mining pool must allocate a budget to initiate the attack. In Fig.~\ref{fig:profit_lag_bribery}, we depict the normalized revenue advantage of the bribery attack for different real-world mining pools as a function of time. The revenue advantage is defined as the difference between the attack revenue and honest mining revenue since the launch of the attack~\cite{grunspan2023profit}. As shown, during the first epoch of the attack---lasting longer than the standard 2-week epoch duration---the revenue advantage decreases, indicating a financial loss for the adversary. However, after the difficulty adjustment at the end of the first epoch, the revenue advantage of the bribery attack begins to increase and eventually becomes positive, indicating profitability.

\begin{figure}[t!]
    \centering
    \includegraphics[height=2in]{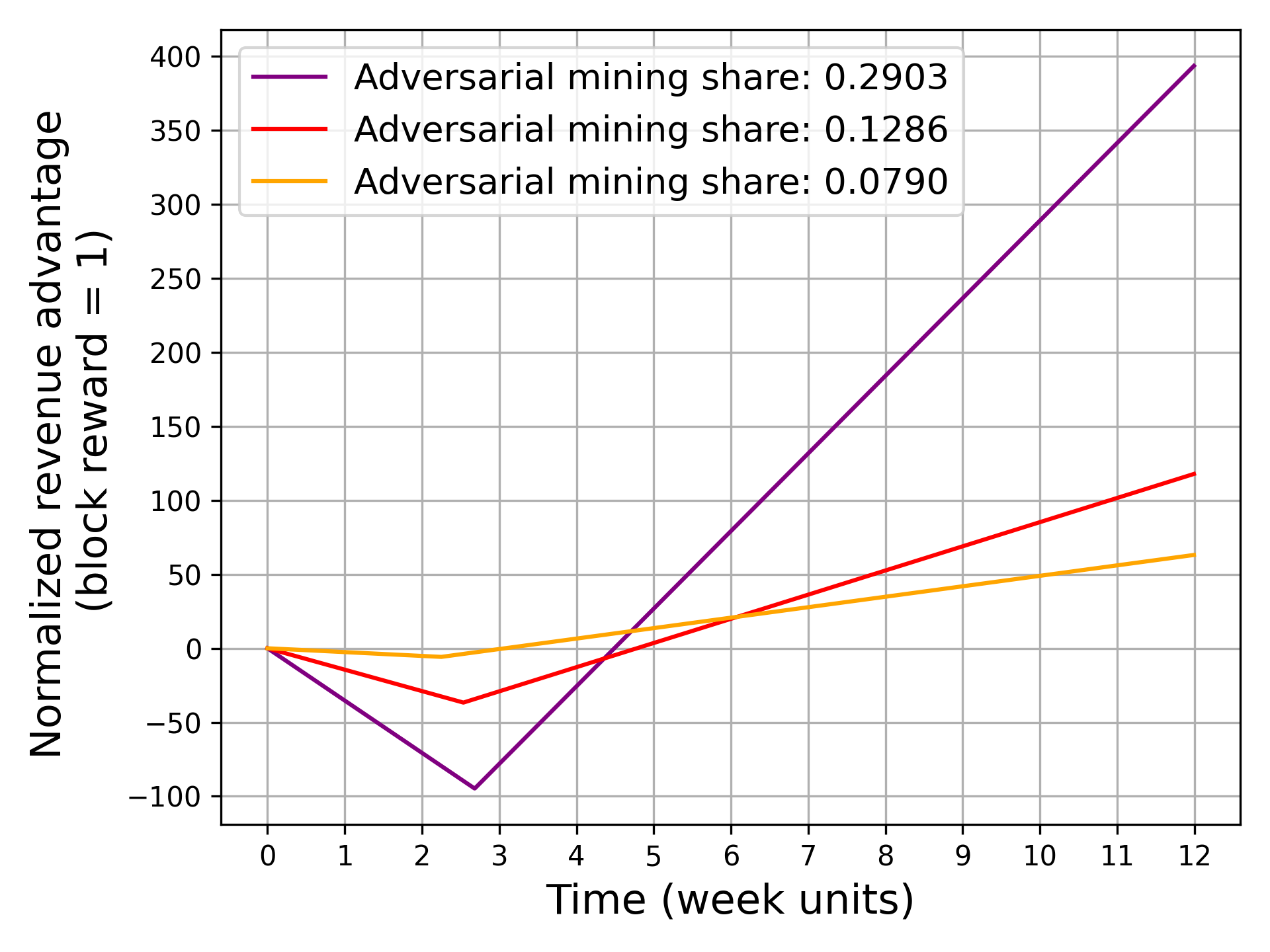}
    \caption{The normalized revenue advantage of the bribery attack over time.}
    \label{fig:profit_lag_bribery}
\end{figure}

\section{Discussion}
Thus far, Bitcoin has not experienced a serious mining destruction attack. This can be attributed to various factors, including the high required mining share and the long initial period of loss associated with these attacks. While these factors may limit mining power destruction attacks, they do not provide a 100\% guarantee that Bitcoin will remain immune to such attacks in the future. For example, while a typical selfish mining attack requires a relatively high hash power ($\geq 0.25$)~\cite{sapirshtein2017optimal}, conducting these attacks more strategically, such as through the bribery attack introduced in this paper, can make them feasible even for smaller mining pools. Additionally, as Bitcoin undergoes more halving events and transaction fees become an increasingly important source of rewards, mining destruction attacks can become more destructive. 
This shift toward a transaction-fee era can lower the mining share threshold for a profitable deviation from the honest strategy, as an adversarial mining pool may be incentivized to risk its low-fee-valuable blocks to orphan others' high-fee-valuable blocks. Orphaning such blocks returns their fee-valuable transactions to the mempool, giving the adversary an opportunity to claim them by mining the subsequent block.
Eliminating the protocol reward can also remove the initial loss period for the adversary, as stolen fee-valuable transactions can offset the reduced block generation rate during the first epoch of the attack~\cite{carlsten2016instability, sarenche2024bitcoin_volatile}. This discussion highlights the importance of analyzing these threats and developing solutions to help Bitcoin remain secure against such attacks in the future. 

Multiple solutions have been suggested in the literature, each with its own limitations. Some solutions rely on assessing the timeliness of block publication to detect honest and adversarial behavior~\cite{heilman2014one, lee2018preventing, solat2017brief, reno2022preventing, zhang2017publish}. However, beyond the limitations of these schemes in partially-synchronous and asynchronous networks, they cannot effectively mitigate attacks in the presence of rational mining pools. Note that rational miners choose the fork to mine on based on its return, regardless of whether the fork is honest or not. There are other solutions that suggest modifying the current Bitcoin difficulty adjustment mechanism (DAM) to prevent difficulty reduction during mining destruction attacks~\cite{zhou2022effective, azimy2022preventing, grunspan2018profitability, sarenche2024time}. These solutions rely on the existence of an honest minority that reports evidence of destruction attacks, such as orphaned blocks. This evidence helps the DAM distinguish between an attack occurrence and a part of the network being offline, thereby preventing difficulty reduction during the attack. A comprehensive discussion on the impact of DAM on the profitability of selfish mining and bribery attacks, as well as modifications to the current Bitcoin DAM to mitigate these attacks, is presented in Appendix~\ref{appendix:mitigation_DAM_selfish_bribery}. However, it is important to note that there exist certain mining power destruction attacks, such as the distraction attack introduced in this paper (Appendix~\ref{sec:distraction}), that destroy mining power without generating any evidence. Furthermore, these solutions may limit adversarial profit at the cost of reducing the block generation rate, thereby penalizing honest miners. For more information on strategies to counter selfish mining, readers are referred to~\cite{sarenche2024time, madhushanie2024selfish, zhang2019lay, nicolas2019comprehensive}.

\section{Conclusion}
In this paper, we discussed how an adversarial mining pool can use various methods to destroy a portion of network's mining power. To execute a mining power destruction attack, the adversarial mining pool must invest some resources, which can be viewed as an initial cost. These resources may include withholding a block in selfish mining or offering a bribe in a bribery attack. However, this initial cost can be compensated, as the destruction attack lowers mining difficulty, resulting in an increased block production rate for the adversarial mining pool.
Victims of destruction attacks can take countermeasures to mitigate them. For example, hiding their identity as block miners is a simple yet effective strategy. Additionally, mining pools can collaborate on mutual agreements, like dividing larger pools into smaller ones or adopting a more secure difficulty adjustment mechanism. However, for every wise solution, there may be an equally clever counterattack.
%
%
%
\bibliographystyle{splncs04}
\bibliography{1_mybib}

\appendix
\section{Parameters Involved in Semi-Rational Fork Choice}\label{appendix_incentivization_factor}
For $\epsilon = \infty$, the $(\epsilon,D)$-semi-rational fork choice rule is the same as the honest fork choice rule. The $\epsilon$ parameter in the definition of the semi-rational fork choice rule captures the idea that a mining pool may decide to stick to the honest strategy up to facing a normalized loss of $\epsilon$. However, once the normalized loss it incurs by following the honest strategy exceeds $\epsilon$, it may deviate from the honest strategy. The parameter $D$ in the definition of the semi-rational fork choice rule represents the maximum length by which the selected chain can be shorter than the longest chain. We use $D$ to present our theorems and lemmas throughout the text. The larger the parameter $D$, the closer we get to modeling a real-world rational mining pool. However, as $D$ increases, modeling the system becomes more complex, since the number of actions a petty-compliant mining pool might consider also increases. 

\section{Our System Model Limitations} \label{appendix:model_limitations}
In our system model, we assume that all block rewards are fixed, and the adversary can place bribes on some blocks to increase the expected return from mining on those blocks. However, this model does not accurately reflect the real world, where blocks can collect varying levels of transaction fees. In Bitcoin, block rewards come from two sources: the protocol reward and the transaction fee reward. In the early years of Bitcoin's introduction, the protocol reward was much higher than the transaction fees, making the fixed block reward assumption reasonable. However, as Bitcoin undergoes more halving events, the protocol reward decreases, and transaction fees become the primary source of mining rewards. In the new transaction-fee era, different blocks can have varying rewards, even without bribes being placed on top of them. This can significantly impact mining power destruction attacks. For instance, in a block race, petty-compliant mining pools will typically adopt the block that leaves the higher transaction fee in the mempool, which is usually the block mined earlier. Therefore, a selfish miner who mines the first block and keeps it secret has a higher chance of winning the block race compared to the miner of the other block in the public fork, suggesting that these attacks may be more threatening in the transaction-fee era. Several papers in the literature~\cite{bar2022werlman, bar2023deep, carlsten2016instability, sarenche2024bitcoin_volatile} have analyzed Bitcoin mining destruction attacks in the presence of altruistic or infinitesimally petty-compliant nodes. However, studying mining destruction attacks in the presence of petty-compliant mining pools in the transaction-fee era remains an interesting direction for future research.

Another assumption in our model is that there is a single adversarial mining pool in the system, while the remaining pools are petty-compliant. This limits the strategies of the remaining petty-compliant pools to a predefined set of actions, implying they cannot respond arbitrarily to the attack. However, in practice, multiple mining pools may behave adversarially to maximize their payoffs. For example, once an adversarial mining pool offers a bribe to launch an attack, another pool may respond by offering a higher bribe to defend against the attack. Several papers in the literature have analyzed multi-agent selfish mining attacks~\cite{bai2023blockchain, zhang2020analysing}, although they do not consider bribes and transaction fees. Analyzing multi-agent destruction attacks in the presence of bribes and within the transaction-fee era could be another interesting direction for future research.

\section{Selfish Mining Attack: Supplementary}
\subsection{Whale Transactions As a Bribe} \label{appendix: whale}
Once the adversarial mining pool mines its first block in its private fork, it publishes transaction $tx_1$ in the public transaction pool. If the transaction fee of $tx_1$ is chosen properly, $tx_1$ would be included in the first block of the public non-adversarial fork. The adversarial mining pool refrains from including $tx_1$ in its private fork. Once the non-adversarial fork length becomes equal to the adversarial fork, the adversarial mining pool publishes its private fork along with transaction $tx_2$ that is only valid if $tx_1$ is not already included in the chain (both $tx_1$ and $tx_2$ pays out form the same address). The transaction fee of $tx_2$ can be set to $\epsilon R$. In this case, the indifferent petty-compliant mining pools get incentivized to select the adversarial fork since they can include $tx_2$ only in the block mined on top of the adversarial fork, not the non-adversarial fork.
\subsection{Proof of Lemma~\ref{lemma: selfish_mining_advantage}} \label{appendix: selfish_mining_advantage}
\begin{proof}
    Let $\mathcal{P}_{\overline{\mathcal{A}}}$ denote the set of all mining pools excluding the adversarial mining pool. Assume that the single block available in the non-adversarial fork is mined by petty-compliant mining pool $p_i \in \mathcal{P}_{\overline{\mathcal{A}}}$, whose mining share is equal to $\alpha_i$. The probability of this event is equal to $\alpha_i$. Since there is a bribe of $\epsilon R$ available on top of the adversarial fork, all the mining pools excluding $p_i$ select the adversarial fork to mine on top of. This implies the only mining pool that is incentivized to mine on top of the non-adversarial fork is the mining pool $p_i$ itself. Therefore, given that the first block of the non-adversarial fork is mined by the petty-compliant mining pool $p_i$, the probability that the next block is mined on top of the non-adversarial fork is equal to $\alpha_i$. Consequently, the total probability of the event that the next block is mined on top of the non-adversarial fork (denoted by event $E$) can be obtained as follows:
    \begin{equation}
        \texttt{Pr}(E) = \frac{\sum_{p_i \in \mathcal{P}_{\overline{\mathcal{A}}}} {{\alpha_i}^2}}{\sum_{p_i \in \mathcal{P}_{\overline{\mathcal{A}}}} {\alpha_i}} = \frac{\sum_{p_i \in \mathcal{P}_{\overline{\mathcal{A}}}} {{\alpha_i}^2}}{1-\alpha_\mathcal{A}} = \beta_\mathcal{A} ,
    \end{equation}
    where $\alpha_\mathcal{A}$ and $\beta_\mathcal{A}$ represent the mining share and the residual centralization factor of the adversarial mining pool, respectively. The event mentioned in the lemma is complementary to event $E$, and thus, its probability is equal to $1 - \beta_\mathcal{A}$.
\null\hfill\qedsymbol\end{proof}

\subsection{Selfish Mining Strategy} \label{appendix_slefish_mining_strategy}
In this section, we introduce the selfish mining strategy $\pi^\texttt{selfish}$, which is suitable for an $(\epsilon, D=1)$-semi-rational environment, where the mining share of all mining pools excluding the adversarial mining pool is assumed to be less than $0.4302$. This assumption is important as it helps us analyze the strategy of a petty-compliant pool that is under the selfish mining $\pi^\texttt{selfish}$ attack.
\begin{lemma} \label{lemma_0.43}
     Assume mining pool $p_i$ to be an $(\epsilon, D=1)$-petty-compliant mining pool. Consider a fork race between forks $f_1$ and $f_2$, each of length $l_1 = 1$ and $l_2 = 2$, respectively, where the single block on fork $f_1$ is mined by mining pool $p_i$. The mining pool $p_i$ is incentivized to abandon mining on fork $f_1$ to mine on top of $f_2$ if $\alpha_i < 0.4302$.
\end{lemma}
The proof of Lemma~\ref{lemma_0.43} is presented in Appendix~\ref{appendix:proof_lemma_0.43}. Note that in the altruistic setting, we can always assume that once a mining pool finds its fork shorter than the adversarial fork, it will adopt the adversarial fork. However, this behavior is not necessarily followed by a petty-compliant mining pool under the attack.
Lemma~\ref{lemma_0.43} shows that, if under a selfish mining attack, an $(\epsilon, D=1)$-petty-compliant mining pool with mining share less than $0.4302$ finds its one-block fork shorter than the adversarial fork, it is also incentivized to abandon its fork and mine on top of the adversarial fork.

\begin{definition} [Strategy $\pi^\texttt{selfish}$] \label{def:strategy}
Let $l_\mathcal{A}$ and $l_\mathcal{R}$ denote the length of the adversarial and semi-rational forks, respectively. Strategy $\pi^\texttt{selfish}$ is defined as follows:
\begin{itemize}
    \item Assume the initial fork state is $(l_\mathcal{A}, l_\mathcal{R})  = (0,0)$. If a petty-compliant mining pool mines a new block, the adversarial mining pool adds the semi-rational block to the head of its chain and continues mining on top of that. Therefore, the fork state remains in state $(l_\mathcal{A}, l_\mathcal{R})  = (0,0)$.
    \item Assume the initial fork state is $(l_\mathcal{A}, l_\mathcal{R})  = (0,0)$. If the adversarial mining pool mines a new block, it keeps the block hidden and does not publish it to the petty-compliant pools. Therefore, the fork state transitions to state $(l_\mathcal{A}, l_\mathcal{R})  = (1,0)$.
    \item Assume the initial fork state is $(l_\mathcal{A}, l_\mathcal{R})  = (1,0)$. If a petty-compliant mining pool $p_i$ mines a new block, a fork race occurs between the adversarial and the semi-rational blocks, and the fork state transitions to state $(l_\mathcal{A}, l_\mathcal{R})  = (1,1)$. In this case, the adversarial mining pool publishes its block along with a normalized bribe of $\texttt{br}= \epsilon$ on top of it. As a result of this bribe, all petty-compliant pools, except $p_i$, mine on top of the adversarial fork, while mining pool $p_i$ continues to mine on top of its own block.
    \item Assume the initial fork state is $(l_\mathcal{A}, l_\mathcal{R})  = (1,0)$. If the adversarial mining pool mines a new block, the fork state transitions to state $(l_\mathcal{A}, l_\mathcal{R})  = (2,0)$.
    \item Assume the initial fork state is $(l_\mathcal{A}, l_\mathcal{R})  = (1,1)$. If the next block is mined on top of the semi-rational fork, the adversarial mining pool gives up on its block and admits the semi-rational fork as the canonical fork. Therefore, the fork state transitions to state $(l_\mathcal{A}, l_\mathcal{R})  = (0,0)$. In this case, the adversarial mining pool retracts its bribe and avoids incurring the bribe cost.
    \item Assume the initial fork state is $(l_\mathcal{A}, l_\mathcal{R})  = (1,1)$. If the next block is mined on top of the adversarial fork, the length of the adversarial fork exceeds that of the semi-rational fork by one block. In this case, we argue that all mining pools adopt the longer adversarial fork and mine on top of it. For a petty-compliant pool $p_j$ that is not the miner of the single block in the semi-rational fork ($p_j \ne p_i$), it is evident that it adopts the longer adversarial fork, as there is no incentive to mine on top of the shorter semi-rational fork owned by another pool. According to Lemma~\ref{lemma_0.43} and the assumption $\alpha_i < 0.4302$, even the petty-compliant mining pool $p_i$, which mined the single semi-rational block, is incentivized to abandon its block and adopt the adversarial fork as the canonical fork. Therefore, the fork state transitions to state $(l_\mathcal{A}, l_\mathcal{R})  = (0,0)$. Note that, in this case, the bribe becomes payable only if the adversarial fork is extended by a non-adversarial block.
    \item Assume the initial fork state is $(l_\mathcal{A}, l_\mathcal{R})  = (2,0)$. If a petty-compliant mining pool $p_i$ mines a new block, the adversarial mining pool publishes its $2$ hidden blocks. Similar to the previous case, according to Lemma~\ref{lemma_0.43} and under the assumption of $\alpha_i < 0.4302$, the petty-compliant mining pool $\pi$ gives up on its block and admits the adversarial fork as the canonical fork. Therefore, the fork state transitions to state $(l_\mathcal{A}, l_\mathcal{R})  = (0,0)$.
    \item Assume the initial fork state is $(l_\mathcal{A}, l_\mathcal{R})  = (2,0)$. If the adversarial mining pool mines a new block, the adversarial mining pool publishes its first hidden block and remains in the state $(l_\mathcal{A}, l_\mathcal{A})  = (2,0)$. \footnote{This action simplifies the relative revenue calculation and is not the optimal action that the adversarial mining pool can take in state $(l_\mathcal{A}, l_\mathcal{R})  = (2,0)$. In an optimal strategy, the adversarial mining pool may hide all three blocks to orphan a longer non-adversarial fork.}
\end{itemize}
    
\end{definition}

\subsection{Proof of Lemma~\ref{lemma_0.43}} \label{appendix:proof_lemma_0.43}
Before proving Lemma~\ref{lemma_0.43}, we first review some useful lemmas. To calculate the expected return of different forks in the view of a given petty-compliant mining pool $p$, we use random walking on a $(x,y)$-grid. The two-dimensional $(x,y)$-grid represents the fork race between the fork of mining pool $p$ and the fork of remaining mining pools, referred to as the public fork. Each fork race scenario can be represented by a path on this grid, which is referred to as the \emph{mining path}. Whenever the petty-compliant mining pool $p$ mines a new block, the mining path moves one step to the right, and whenever a block is added on top of the public fork, the mining path moves one step up.
\begin{lemma}\label{lemma_never_reach_y=x_r}
    Consider the $(x,y)$-grid representation of the fork race between the petty-compliant mining pool $p$ with mining share $\alpha$ and the remaining mining pools. Let $\overline{P_{r}}$ for $r\ge1$ denote the probability of the event that the mining path starting at $(0,0)$ never reaches the line $y=x-r$. We have $\overline{P_{r}} = 1-\Big(\frac{\alpha}{1-\alpha}\Big)^r$.
\end{lemma}
\begin{proof}
    The proof of Lemma~\ref{lemma_never_reach_y=x_r} is presented in Appendix D.2 of~\cite{sarenche2024deep}. Here, we briefly review the idea behind the proof. 
    
    Let $P_r$ denote the probability of the event that mining path starting at $(0,0)$ reaches the line $y=x-r$ at least once. Starting at point $(0,0)$, the mining path moves to point $(1,0)$ with probability $\alpha$ and moves to point $(0, 1)$ with probability $1-\alpha$. The probability of the event that mining path starting at $(1,0)$ reaches the line $y=x-r$ is equal to $P_{r-1}$. And, the probability of the event that mining path starting at $(0,1)$ reaches the line $y=x-r$ is equal to $P_{r+1}$. Therefore, we have: $P_r = \alpha P_{r-1} + (1-\alpha) P_{r+1}$. This recursive equation can be evaluated for different values of $r \ge 1$ by haviing $P_1$ and $P_2$.

    We first find the probability of \( P_1 \), i.e., the probability of reaching the line \( y = x - 1 \) starting at point $(0,0)$. Suppose the mining path reaches the line \( y = x - 1 \) at the point \( (s + 1, s) \) for the first time. For this to occur, the mining path must reach the point \( (s, s) \) without crossing below the line \( y = x \), and then take a step to the right. The number of valid paths from \( (0, 0) \) to \( (s, s) \), while not passing below the line \( y = x \), corresponds to the \( s \)-th Catalan number, which we refer to as \( c_s \). Therefore, the probability that the mining path reaches the line \( y = x - 1 \) for the first time at \( (s + 1, s) \) is given by \( c_s \alpha^{s+1} (1 - \alpha)^s \). Hence, the overall probability of reaching the line \( y = x - 1 \) can be expressed as:

    \begin{equation}
        P_1 = \alpha \sum_{s=0}^{\infty} c_s \alpha^s (1 - \alpha)^s = \frac{\alpha}{1 - \alpha} \enspace.
    \end{equation}
    
    The method for solving this series is provided in~\cite{chen2016interesting}. We next find the probability of \( P_2 \), i.e., the probability of reaching the line \( y = x - 1 \) starting at point $(0,0)$. Suppose the mining path first reaches the line \( y = x - 2 \) at the point \( (s + 2, s) \). For this to happen, the mining path must first reach the point \( (s + 1, s) \) without falling below the line \( y = x - 1 \), and then take one step to the right. The number of valid paths from \( (0, 0) \) to \( (s + 1, s) \), without falling below the line \( y = x - 1 \), corresponds to the \( (s + 1) \)-th Catalan number, \( c_{s+1} \). Thus, the probability of reaching the line \( y = x - 2 \) for the first time at the point \( (s + 2, s) \) is given by \( c_{s+1} \alpha^{s+2} (1 - \alpha)^s \). Therefore, the probability of reaching the line \( y = x - 2 \) is:

    \begin{equation}
        P_2 = \alpha^2 \sum_{s=0}^{\infty} c_{s+1} \alpha^s (1 - \alpha)^s = \frac{\alpha}{1 - \alpha} \sum_{s=1}^{\infty} c_s \alpha^s (1 - \alpha)^s = \left( \frac{\alpha}{1 - \alpha} \right)^2 \enspace.
    \end{equation}

    Knowing $P_1 = \frac{\alpha}{1-\alpha}$, $P_2 = \big(\frac{\alpha}{1-\alpha}\big)^2$, and the recursive equation $P_r = \alpha P_{r-1} + (1-\alpha) P_{r+1}$, the probability $P_r$ for $r\ge1$ can be obtained as $P_r = \big(\frac{\alpha}{1-\alpha}\big)^r$. Therefore, the probability of the event that the mining path starting at $(0,0)$ never reaches the line $y=x-r$ canbe obtained as $\overline{P_{r}} = 1-\Big(\frac{\alpha}{1-\alpha}\Big)^r$.
    
\null\hfill\qedsymbol\end{proof}

\begin{lemma}\label{lemma_number_path_y=x+d,y>x-2,y<x+d}
     Let $G_s^d$ for $d \ge 1$ denote the number of paths in a $(x,y)$-grid from start point $(0,0)$ to the point $(s,s+d)$ without reaching the lines $y=x-2$ and $y=x+d$ in advance. We have:
    \begin{equation}
        \begin{split}
            & \sum_{s=0}^{\infty} {G_s^d \big(\alpha(1-\alpha)\big)^s}= \frac{1-2\alpha}{(1-\alpha)^{d+2}-\alpha^{d+2}} \enspace, \\
            & \sum_{s=0}^{\infty} {s G_s^d \big(\alpha(1-\alpha)\big)^s} = \frac{(d+2)\alpha(1-\alpha)\big((1-\alpha)^{d+1}+\alpha^{d+1}\big)}{\big((1-\alpha)^{d+2}-\alpha^{d+2}\big)^2} \\
            & - \frac{2\alpha(1-\alpha)}{(1-2\alpha)\big((1-\alpha)^{d+2}-\alpha^{d+2}\big)} \enspace .
        \end{split}
    \end{equation}
\end{lemma}
\begin{proof}
    The proof follows the same logic as the proof of Lemma 4 presented in Appendix D.2 of~\cite{sarenche2024deep}.
    Assume $\texttt{Event}_1(s)$ is defined as follows: (i) the block path starting at $(0,0)$ reaches the point $(s,s+d-1)$, and (ii) before reaching the point $(s,s+d-1)$, the block path never passes the line $y=x+d-1$ and never reaches the line $y=x-2$. According to the definition of $G_s^d$, the probability of $\texttt{Event}_1(s)$ is equal to $G_s^d \alpha^s (1-\alpha)^{s+d-1}$. Assume $\texttt{Event}_2(s)$ is defined as follows: (i) the block path starting at $(0,0)$ passes the line $y=x+d-1$ for the first time at point $(s,s+d-1)$, and (ii) before reaching the point $(s,s+d-1)$, the block path never reaches the line $y=x-2$. $\texttt{Event}_2(s)$ happens if, after the occurrence of $\texttt{Event}_1(s)$ and reaching the point $(s,s+d-1)$, the block path immediately moves one step up to pass the line $y=x+d-1$ and reach the point $(s,s+d)$. Therefore, the probability of $\texttt{Event}_2(s)$ is equal to $G_s^d \alpha^s (1-\alpha)^{s+d}$. Assume $\texttt{Event}_3(s)$ is defined as follows: (i) the block path starting at $(0,0)$ reaches the line $y=x+d$ for the first time at point $(s,s+d)$, and (ii) the block path never reaches the line $y=x-2$ both before and after reaching the point $(s,s+d)$. $\texttt{Event}_3(s)$ happens if, after the occurrence of $\texttt{Event}_2(s)$ and reaching the point $(s,s+d)$, the block path never reaches the line $y=x-r$. The event that the block path starting at $(s,s+d)$ never reaches the line $y=x-2$ is equivalent to the event that the block path starting at $(0,0)$ never reaches the line $y=x-(d+2)$, the probability of which is equal to $1-(\frac{\alpha}{1-\alpha})^{d+2}$ according to Lemma~\ref{lemma_never_reach_y=x_r}. Therefore, the probability of $\texttt{Event}_3(s)$ is equal to $G_s^d \alpha^s (1-\alpha)^{s+d} \big(1-(\frac{\alpha}{1-\alpha})^{d+2}\big)$. Note that since $1-\alpha>\alpha$, a mining path will eventually pass the line $y=x+d$. Therefore, the sum of $\texttt{Event}_3(s)$ probabilities over all values of $s$ is equal to the probability that the block path never reaches the line $y=x-2$, the probability of which is equal to $1-(\frac{\alpha}{1-\alpha})^{2}$ according to Lemma~\ref{lemma_never_reach_y=x_r}. As a result,
    \begin{equation}
        \begin{split}
            &\sum_{s=0}^{\infty} G_s^d \alpha^s (1-\alpha)^{s+d} \big(1-(\frac{\alpha}{1-\alpha})^{d+2}\big) = 1-(\frac{\alpha}{1-\alpha})^{2}\\
            & \Rightarrow \sum_{s=0}^{\infty} G_s^d \big(\alpha(1-\alpha)\big)^s = \frac{1-2\alpha}{(1-\alpha)^{d+2}-\alpha^{d+2}} \enspace.
        \end{split}
   \end{equation}
    To prove the second equality, we use the variable substitution $\alpha(1-\alpha)=x$ in the equality above.
    By taking the derivative from both sides, then multiplying both sides to $x$, and finally substituting $x=\alpha(1-\alpha)$, the second equality can be obtained.   
\null\hfill\qedsymbol\end{proof}

\begin{lemma}\label{lemma_number_path_y=x-2,y>x-2,y<x+d}
     Let $F_s^d$ for $d \ge 1$ denote the number of paths in a $(x,y)$-grid from start point $(0,0)$ to the point $(s+2,s)$ without reaching the lines $y=x-2$ and $y=x+d$ in advance. We have:
    \begin{equation}
        \begin{split}
            & \sum_{s=0}^{\infty} {F_s^d \big(\alpha(1-\alpha)\big)^s}= \frac{1}{\alpha^2}\Big(1-\frac{1-2\alpha}{(1-\alpha)^2\big(1-(\frac{\alpha}{1-\alpha})^{d+2}\big)}\Big) \enspace.
        \end{split}
    \end{equation}
\end{lemma}

\begin{proof}
    According to the proof of Lemma~\ref{lemma_number_path_y=x+d,y>x-2,y<x+d}, the probability of the event that the mining path starting at point $(0,0)$ reaches the line $y=x+d$ before reaching the line $y=x-2$ is equal to $\frac{1-2\alpha}{(1-\alpha)^2\big(1-(\frac{\alpha}{1-\alpha})^{d+2}\big)}$. Therefore, the probability of the event that the mining path starting at point $(0,0)$ reaches the line $y=x-2$ before reaching the line $y=x+d$ is equal to $1-\frac{1-2\alpha}{(1-\alpha)^2\big(1-(\frac{\alpha}{1-\alpha})^{d+2}\big)}$. According o the definition of $F_s^d$, the latter probability is equal to $F_s^d \alpha^2 \big(\alpha(1-\alpha)\big)^s$. This completes our proof.
\null\hfill\qedsymbol\end{proof}

\begin{proof}[Proof of Lemma~\ref{lemma_0.43}]
    We need to compare the expected return of mining on top of 2 forks $f_1$ and $f_2$ to find out which one is more profitable. We assume:
    \begin{itemize}
        \item If fork $f_1$ lags behind fork $f_2$ by $d$ blocks, where $d \ge 2$, the mining pool $p_i$ switches to mine on top of $f_2$. This is because mining pool $p_i$ is assumed to be a $D=1$-petty-compliant mining pool.
        \item If fork $f_1$ becomes even one block longer than fork $f_2$, all the remaining mining pools switch to mine on top of $f_1$. This assumption is in favor of the mining pool $p_i$.
    \end{itemize}
    We calculate the expected return using the $(x,y)$-grid. $l_1=1$ and $l_2=2$ implies that the mining path is at point $(1,2)$. If the mining path reaches the line $y=x-1$ before reaching the line $y=x+d$, the mining pool $p_i$ wins the fork race. If the mining path reaches the line $y=x+d$ before reaching the line $y=x-1$, the mining pool $p_i$ loses the fork race. Assume the mining path starting at point $(1,2)$ reaches the line $y=x-1$ for the first time at point $(s+3,s+2)$ without reaching the line $y=x+d$ in advance. The probability of this event is equal to $F_s^{d-1} \alpha_i^2 (\alpha_i(1-\alpha_i))^s$. The mining pool $p_i$ receives $s+3$ block rewards under this event. Therefore, the expected return of mining on top of fork $f_1$ for mining pool $p_i$ up to the point that the fork race is resolved can be obtained as follows:
    \begin{equation}
        r_1 = \sum_{s=0}^{\infty} {(s+3) F_s^{d-1} \alpha_i^2 (\alpha_i(1-\alpha_i))^s} \enspace .
    \end{equation}
    The fork race ends once the mining path reaches either $y=x-1$ or $y=x+d$. To obtain the expected return of switching the fork, we calculate the reward that the mining pool $p_i$ could have received if it was mining on top of the fork $f_2$. Assume the mining path starting at point $(1,2)$ reaches the line $y=x+d$ for the first time at point $(s+1,s+d+1)$ without reaching the line $y=x-1$ in advance. The probability of this event is equal to $G_s^{d-1} (1-\alpha_i)^{d-1} (\alpha_i(1-\alpha_i))^s$. The mining pool $p_i$ receives $s$ block rewards under this event. Now, assume the mining path starting at point $(1,2)$ reaches the line $y=x-1$ for the first time at point $(s+3,s+2)$ without reaching the line $y=x+d$ in advance. The probability of this event is equal to $F_s^{d-1} \alpha_i^2 (\alpha_i(1-\alpha_i))^s$. The mining pool $p_i$ receives $s+2$ block rewards under this event. Therefore, the expected return of mining on top of fork $f_2$ for mining pool $p_i$ up to the point that the mining path reaches either $y=x-1$ or $y=x+d$ can be obtained as follows:
     \begin{equation}
        r_2 = \sum_{s=0}^{\infty} {(s+2) F_s^{d-1} \alpha_i^2 (\alpha_i(1-\alpha_i))^s + s G_s^{d-1} (1-\alpha_i)^{d-1} (\alpha_i(1-\alpha_i))^s} \enspace .
    \end{equation}
    Using Lemmas~\ref{lemma_number_path_y=x-2,y>x-2,y<x+d} and~\ref{lemma_number_path_y=x+d,y>x-2,y<x+d}, $r_1$ and $r_2$ can be calculated. The petty-compliant mining pool $p_i$ is incentivized to leave its fork at state $l_1=1$ and $l_2=2$ if $r_1 - r_2 < 0$. For $\alpha_i \in (0,0.5)$, the function $r_1 - r_2$ is decreasing on $d$. By assigning $d=2$, the inequality $r_1 - r_2 < 0$ holds if $\alpha_i < 0.4302$.
\null\hfill\qedsymbol\end{proof}

\subsection{Proof of Theorem~\ref{theorem_D=inf_selfish}} \label{proof:theorem_D=inf_selfish}
\begin{proof}[Proof of Theorem~\ref{theorem_D=inf_selfish}]
    Let $\texttt{Pr}_{l_\mathcal{A}, l_\mathcal{R}}$ denote the probability of being in state $(l_\mathcal{A}, l_\mathcal{R})$. We can obtain the following equations:
    \begin{equation}
        \begin{split}
            & \texttt{Pr}_{0,0} = \frac{1-\alpha_\mathcal{A}}{1+(1-\alpha_\mathcal{A})^2 \alpha_\mathcal{A}} \enspace , \\
            & \texttt{Pr}_{1,0} = \alpha_\mathcal{A} \texttt{Pr}_{0,0} \enspace , \\
            & \texttt{Pr}_{2,0} = \frac{\alpha_\mathcal{A}}{1-\alpha_\mathcal{A}} \texttt{Pr}_{1,0} \enspace , \\
            & \texttt{Pr}_{1,1} = (1-\alpha_\mathcal{A}) \texttt{Pr}_{1,0} \enspace .
        \end{split}
    \end{equation}
    To find the adversarial average revenue, we should consider the average number of adversarial blocks that get added to the canonical chain at each state. The average revenue that the adversarial mining pool receives can be obtained as follows:
    \begin{equation}
        \texttt{Rev} = \Big(\alpha_\mathcal{A}\texttt{Pr}_{2,0}+2(1-\alpha_\mathcal{A})\texttt{Pr}_{2,0} + 2\alpha_\mathcal{A}\texttt{Pr}_{1,1} + (1-\alpha_\mathcal{A}-\beta_\mathcal{A}) \texttt{Pr}_{1,1}\Big) \lambda' R \enspace, 
    \end{equation}
    where $\lambda'$ is the block mining rate under the selfish mining attack.

    To find the adversary's cost, we ignore the mining cost\footnote{The mining cost during both honest mining and selfish mining is the same, and thus does not affect the profit comparison between these two strategies.} and only consider the cost incurred by the adversarial mining pool due to paying bribes. Note that the normalized bribe $\texttt{br}= \epsilon$ is payable if, under a fork race, the adversarial fork is extended by a non-adversarial block. In this case, the non-adversarial mining pool that extends the adversarial block is eligible to collect the bribe.
    The average cost of the adversarial mining pool can be obtained as follows:
    \begin{equation}
       \texttt{Cost} = \sum_{p_i \in \mathcal{P}_{\overline{\mathcal{A}}}}{\epsilon(1-\alpha_\mathcal{A}-\alpha_i) \alpha_i\texttt{Pr}_{1,0} \lambda' R} =\epsilon(1-\alpha_\mathcal{A}-\beta_\mathcal{A}) \texttt{Pr}_{1,1} \lambda' R \enspace,
    \end{equation}
    where $\mathcal{P}_{\overline{\mathcal{A}}}$ is the set of non-adversarial mining pools.
    
    The mining rate under the selfish mining attack can be obtained as $\lambda' = \frac{\lambda}{\texttt{Diff}}$, where $\texttt{Diff}$ is the normalized number of blocks added to the canonical chain and can be obtained as follows:
    \begin{equation}
        \texttt{Diff} = (1-\alpha_\mathcal{A})\texttt{Pr}_{0,0}+\alpha_\mathcal{A}\texttt{Pr}_{2,0}+2(1-\alpha_\mathcal{A})\texttt{Pr}_{2,0} + 2 \texttt{Pr}_{1,1} \enspace .
    \end{equation}
    Therefore, after mining difficulty adjustment, the time-averaged profit can be calculated as below:
    \begin{equation}
        \texttt{Profit}\big(\pi^\texttt{selfish}\big) = \frac{2\alpha_\mathcal{A}^4-5\alpha_\mathcal{A}^3+4\alpha_\mathcal{A}^2+\alpha_\mathcal{A}(1-\alpha_\mathcal{A})(1-\alpha_\mathcal{A}-\beta_\mathcal{A})(1-\epsilon)}{\alpha_\mathcal{A}^3-\alpha_\mathcal{A}^2+1} \lambda R \enspace .
    \end{equation}
    The selfish mining strategy $\pi^\texttt{selfish}$ strongly dominates honest mining if we have:
    \begin{equation}
        \texttt{Profit}\big(\pi^\texttt{selfish}\big)>\alpha_\mathcal{A}\lambda R \Longleftrightarrow \beta_\mathcal{A} < \frac{\alpha_\mathcal{A}- \epsilon (1-\alpha_\mathcal{A})^2}{(1-\alpha_\mathcal{A})(1-\epsilon)} \enspace .
    \end{equation}
\null\hfill\qedsymbol\end{proof}

\subsection{Proof of Corollary~\ref{corollary:p_1andp_2}}\label{appendix: p_1andp_2}
\begin{proof}
According to Theorem~\ref{theorem_D=inf_selfish}, selfish mining dominates honest mining for $P_1$ if the inequality $\beta_1 < \frac{\alpha_1- \epsilon (1-\alpha_1)^2}{(1-\alpha_1)(1-\epsilon)} \enspace$ holds. The maximum value that $P_1$'s residual centralization factor  $\beta_1$ can take is equal to the mining share of the largest mining pool excluding $P_1$, which in our case is $\alpha_2$. This maximum value occurs when all mining pools, excluding $P_1$, have the same mining share $\alpha_i$. Therefore, to satisfy the inequality $\beta_1 < \frac{\alpha_1- \epsilon (1-\alpha_1)^2}{(1-\alpha_1)(1-\epsilon)} \enspace$, one must ensure that $\alpha_2 < \frac{\alpha_1- \epsilon (1-\alpha_1)^2}{(1-\alpha_1)(1-\epsilon)} \enspace$ holds. The latter inequality holds if $\frac{\alpha_1}{1-\alpha_1}>\alpha_2 + \epsilon (1-\alpha_1-\alpha_2)$ holds. This proves the first statement of the corollary.

In the case of $\epsilon = 0$, selfish mining is the dominant strategy for $P_1$ if $\frac{\alpha_1}{1-\alpha_1}>\alpha_2$ holds, which is always the case as $\alpha_1 \ge \alpha_2$.
\null\hfill\qedsymbol\end{proof}

\subsection{Mining pools}\label{appendix_mining_pools}
The distribution of mining power share for the first $8$ months of 2024 is presented in Table~\ref{tab:MiningPowerDistribution}~\cite{Pools}. We consider all the unknown miners as a single mining pool.
\begin{table}[b]
    \caption{Mining Power Distribution of Bitcoin Mining Pools}
    \label{tab:MiningPowerDistribution}
    \centering
    \begin{tabular}{|c|c|c|c|c|c|}
        \hline
        \textbf{Foundry USA} & \textbf{AntPool} & \textbf{ViaBTC} & \textbf{F2Pool} & \textbf{Unknown} & \textbf{Mara Pool}  \\
        \hline
        29.03\% & 24.85\% & 12.86\% & 11.53\% & 7.90\% & 3.44\% \\
        \hline
        \textbf{Binance Pool} & \textbf{SBI Crypto} & \textbf{Braiins Pool} & \textbf{BTC.com} & \textbf{BTC M4} & \textbf{Poolin}  \\
        \hline
        3.00\% & 2.10\% & 1.78\% & 1.42\% & 0.83\% & 0.80\% \\
        \hline
        \textbf{Ultimus} & \textbf{1THash} & \textbf{Solo CKPool} & \textbf{KanoPool} & & \\
        \hline
        0.35\% & 0.00054\% & 0.034\% & 0.011\% & & \\
        \hline
    \end{tabular}
\end{table}

\subsection{MDP-Based Analysis} \label{sec:MDP}
In this section, we present an MDP-based implementation to analyze selfish mining in the presence of petty-compliant mining pools. The implementation is available at~\cite{Our_MDP}. Our MDP implementation considers the following two assumptions: 
1) The environment is an $(\epsilon, D=0)$-semi-rational environment, where $D=0$ implies that petty-compliant mining pools always select the longest chain. However, in the case of a fork with chains of the same height, they opt for the chain with the highest return.
2) The bribe placed on the adversarial chain is valid until the next non-adversarial block is mined. The next non-adversarial block either collects the bribe or leaves it unspent, allowing the adversary to collect it back. This assumption can be implemented using a smart contract, as described in Appendix~\ref{appendix:bribe_selfish_mining_smart_contract}.
Our selfish mining analysis differs from the implementation in~\cite{bar2023deep}, which assumes volatile block rewards. Our approach assumes uniform block rewards, while allowing the adversary to place bribes on top of the adversarial fork.

MDP can be used to obtain the optimal strategy in settings where the number of states is limited to less than $10^7$~\cite{zhang2019lay}. Due to the limitation of the state space, our MDP implementation only accounts for sufficiently large mining pools, while treating the remaining mining power as a single, aggregated mining pool. Let $N$ denote the total number of mining pools available in the network. Our implementation can consider up to almost $N=10$ mining pools. Assume that mining pool $p_1$ is the single honest mining pool.\footnote{All honest mining power is unified into a single honest mining pool.} Assume further that the set $\{p_2, p_3, \ldots, p_{N-1}\}$ denotes the $N-2$ petty-compliant mining pools, and $p_\mathcal{A}$ denotes the adversarial mining pool. Each state in our implementation represents a fork race between the adversarial fork of length $l_\mathcal{A}$ and the non-adversarial fork of length $l_{\overline{\mathcal{A}}}$.
We first discuss the set of possible actions. The adversarial mining pool can choose from four major actions: \texttt{Override}, \texttt{Wait}, \texttt{Adopt}, and \texttt{Match}, where the action \texttt{Match} consists of multiple subactions. 

\noindent\texttt{Override}: This action represents the adversarial mining pool publishing a sub-fork of its secret fork that is one block longer than the honest fork. As a result of the \texttt{Override} action, non-adversarial mining pools abandon their current fork and start mining on top of the longer adversarial fork. This is because, according to our assumption, the petty-compliant pools are $(\epsilon, D=0)$-semi-rational, meaning that they always mine on top of the longest chain. This action is feasible if $l_\mathcal{A} > l_{\overline{\mathcal{A}}}$.

\noindent\texttt{Wait}: This action represents the adversarial mining pool continuing to mine new blocks on top of its secret fork. This action is always feasible.

\noindent\texttt{Adopt}: The \texttt{Adopt} action represents the adversarial mining pool abandoning its adversarial fork and accepting the non-adversarial fork. This action is always feasible.

\noindent$\texttt{Match}_i$: The action $\texttt{Match}_i$ for $i \in \{0,1, \cdots, \texttt{max\_bribe}\}$ represents the adversarial mining pool publishing a sub-fork of its secret fork that is equal in length to the non-adversarial fork and placing a normalized bribe of value $i+\epsilon$ on top of the adversarial fork. \texttt{max\_bribe} is a parameter denoting the maximum bribe the adversary can pay. Action $\texttt{Match}_i$ results in a race between two forks of the same height. Petty-compliant pools, whose number of blocks on the non-adversarial fork is less than or equal to $i$, are incentivized to mine on the adversarial fork. This action is feasible if $l_\mathcal{A} \ge l_{\overline{\mathcal{A}}}$. 

The most important difference between the set of possible actions in the setting of petty-compliant mining pools and the altruistic setting is the action $\texttt{Match}_i$ for $i \in [\texttt{max\_bribe}]$. Note that the undercutting attack is a subset of the stated possible actions. The undercutting attack can be described as follows: starting in state $(l_\mathcal{A} = 0, l_{\overline{\mathcal{A}}} = 1)$, the adversarial mining pool first takes the \texttt{Wait} action, attempts to mine a block to transition to state $(l_\mathcal{A} = 1, l_{\overline{\mathcal{A}}} = 1)$, and then proceeds with the $\texttt{Match}_0$ action.

Each state in our implementation has the following form: $$\big(l_1, l_2, \ldots, l_{N-1}, l_\mathcal{A}, \texttt{latest}, \texttt{match}, \texttt{bribe}\big) \enspace.$$ Here, $l_i$ for $i \in [N-1]$ denotes the number of blocks mined by pool $p_i$ in the non-adversarial fork, where we have $l_{\overline{\mathcal{A}}} = \sum_{i=1}^{N-1} l_i$. The variable \texttt{latest} represents information on the latest block mined in the system, \texttt{match} is a boolean element indicating whether the action $\texttt{Match}$ is active, and \texttt{bribe} denotes the amount of bribe placed on top of the adversarial fork. Table~\ref{tab:selfish_mining_MDP_result_40} presents the reward shares an adversarial mining pool with a mining share of $0.4$ can achieve under different distributions of mining power among petty-compliant mining pools. Table~\ref{tab:selfish_mining_MDP_result_30} presents the selfish mining reward shares of an adversarial mining pool with a mining share of $0.3$. Table~\ref{tab:selfish_mining_MDP_result_real} presents the selfish mining reward shares achieved by mining pools, with mining shares assigned according to the real-world distribution shown in Table~\ref{tab:MiningPowerDistribution}.
As can be observed, in scenarios where the residual centralization factor with respect to the adversary is lower, the reward share the adversary can achieve is higher. 
\begin{table}[t]
    \centering
    \caption{Selfish mining in the presence of petty-compliant pools ($\alpha_\mathcal{A}=0.4, \enspace \epsilon=0.1$).} \label{tab:selfish_mining_MDP_result_40}
    \resizebox{\textwidth}{!}{ 
    \begin{tabular}{|c|c|}
    \hline
    \textbf{Adversarial Share} & \textbf{Petty-Compliant Mining Pool Shares} \\
    \hline
    0.40 & [0.3, 0.3] \\
    \hline
    \textbf{Residual Centralization Factor}: 0.3 & \cellcolor{red!20}\textbf{Reward Share:} 0.5448\\
    \hline\hline
    \textbf{Adversarial Share} & \textbf{Petty-Compliant Mining Pool Shares} \\
    \hline
    0.40 & [0.2, 0.2, 0.2] \\
    \hline
    \textbf{Residual Centralization Factor}: 0.2 & \cellcolor{red!40}\textbf{Reward Share:} 0.5714\\
    \hline\hline
    \textbf{Adversarial Share} & \textbf{Petty-Compliant Mining Pool Shares} \\
    \hline
    0.40 & [0.1, 0.1, 0.1, 0.1, 0.05, 0.05, 0.05, 0.05] \\
    \hline
    \textbf{Residual Centralization Factor}: 0.0766 & \cellcolor{red!60}\textbf{Reward Share:} 0.5955\\
    \hline\hline
    \textbf{Adversarial Share} & \textbf{Petty-Compliant Mining Pool Shares} \\
    \hline
    0.40 & [0.075, 0.075, 0.075, 0.075, 0.075, 0.075, 0.075, 0.075] \\
    \hline
    \textbf{Residual Centralization Factor}: 0.075 & \cellcolor{red!80}\textbf{Reward Share:} 0.5967\\
    \hline
    \end{tabular}
    } 
\end{table}

\begin{table}[b]
    \centering
    \caption{Selfish mining in the presence of petty-compliant pools ($\alpha_\mathcal{A}=0.3, \enspace \epsilon=0$).} \label{tab:selfish_mining_MDP_result_30}
    \resizebox{\textwidth}{!}{ 
    \begin{tabular}{|c|c|}
    \hline
    \textbf{Adversarial Share} & \textbf{Petty-Compliant Mining Pool Shares} \\
    \hline
    0.30 & [0.4, 0.2, 0.1] \\
    \hline
    \textbf{Residual Centralization Factor}: 0.3 & \cellcolor{red!20}\textbf{Reward Share:} 0.3534\\
    \hline\hline
    \textbf{Adversarial Share} & \textbf{Petty-Compliant Mining Pool Shares} \\
    \hline
    0.30 & [0.2, 0.2, 0.2, 0.1] \\
    \hline
    \textbf{Residual Centralization Factor}: 0.1857 & \cellcolor{red!40}\textbf{Reward Share:} 0.3877\\
    \hline\hline
    \textbf{Adversarial Share} & \textbf{Petty-Compliant Mining Pool Shares} \\
    \hline
    0.30 & [0.2, 0.2, 0.05, 0.05, 0.05, 0.05, 0.05, 0.05] \\
    \hline
    \textbf{Residual Centralization Factor}: 0.1357 & \cellcolor{red!60}\textbf{Reward Share:} 0.4006\\
    \hline\hline
    \textbf{Adversarial Share} & \textbf{Petty-Compliant Mining Pool Shares} \\
    \hline
    0.30 & [0.0875, 0.0875, 0.0875, 0.0875, 0.0875, 0.0875, 0.0875, 0.0875] \\
    \hline
    \textbf{Residual Centralization Factor}: 0.0875 & \cellcolor{red!80}\textbf{Reward Share:} 0.4112\\
    \hline
    \end{tabular}
    } 
\end{table}

\begin{table}[t]
    \centering
    \caption{Selfish mining results under real-world mining power distribution ($\epsilon=0$).} \label{tab:selfish_mining_MDP_result_real}
    \resizebox{\textwidth}{!}{ 
    \begin{tabular}{|c|c|}
    \hline
    \textbf{Adversarial Share} & \textbf{Petty-Compliant Mining Pool Shares} \\
    \hline
    0.07902 & [0.29033, 0.24852, 0.12858, 0.11527, 0.03438, 0.03000, 0.02101, 0.05289] \\
    \hline
    \textbf{Residual Centralization Factor}: 0.1967 & \cellcolor{red!20}\textbf{Reward Share:} 0.0794\\
    \hline\hline
    \textbf{Adversarial Share} & \textbf{Petty-Compliant Mining Pool Shares} \\
    \hline
    0.11527 & [0.29033, 0.24852, 0.12858, 0.07902, 0.03438, 0.03000, 0.02101, 0.05289] \\
    \hline
    \textbf{Residual Centralization Factor}: 0.1968 & \cellcolor{red!40}\textbf{Reward Share:} 0.1166\\
    \hline\hline
    \textbf{Adversarial Share} & \textbf{Petty-Compliant Mining Pool Shares} \\
    \hline
    0.12858 & [0.29033, 0.24852, 0.11527, 0.07902, 0.03438, 0.03000, 0.02101, 0.05289] \\
    \hline
    \textbf{Residual Centralization Factor}: 0.1961 & \cellcolor{red!60}\textbf{Reward Share:} 0.1306\\
    \hline\hline
    \textbf{Adversarial Share} & \textbf{Petty-Compliant Mining Pool Shares} \\
    \hline
    0.24852 & [0.29033, 0.12858, 0.11527, 0.07902, 0.03438, 0.03000, 0.02101, 0.05289] \\
    \hline
    \textbf{Residual Centralization Factor}: 0.1672 & \cellcolor{red!80}\textbf{Reward Share:} 0.2980\\
    \hline\hline
    \textbf{Adversarial Share} & \textbf{Petty-Compliant Mining Pool Shares} \\
    \hline
    0.29033 & [0.24852, 0.12858, 0.11527, 0.07902, 0.03438, 0.03000, 0.02101, 0.05289] \\
    \hline
    \textbf{Residual Centralization Factor}: 0.1453 & \cellcolor{red!100}\textbf{Reward Share:} 0.3785\\
    \hline
    \end{tabular}
    } 
\end{table}

\subsection{Smart Contract-Based Selfish Mining}\label{appendix:bribe_selfish_mining_smart_contract}
The goal of the adversarial mining pool is to design a bribing smart contract such that the bribe placed on the adversarial fork remains valid until the next non-adversarial block is mined. If a non-adversarial block is mined on top of the adversarial fork, the miner of the non-adversarial block should be able to receive the bribe; otherwise, the adversary should be able to recollect it.

Let $B_\mathcal{A}$ denote the tip of the adversarial fork and $B_{\overline{\mathcal{A}}}$ denote the tip of the non-adversarial fork. The adversarial mining pool deploys a smart contract and publishes it for all mining pools. The smart contract stores four parameters: the hash of block $B_\mathcal{A}$, denoted by $\texttt{Hash}(B_\mathcal{A})$; the hash of block $B_{\overline{\mathcal{A}}}$, denoted by $\texttt{Hash}(B_{\overline{\mathcal{A}}})$; a difficulty target denoted by $\texttt{Target}$, which is equal to the difficulty target of the epoch to which blocks $B_\mathcal{A}$ and $B_{\overline{\mathcal{A}}}$ belong; and a payout address $\texttt{add}_\mathcal{A}$ for the adversarial mining pool. The adversarial mining pool deposits a normalized bribe $\texttt{br}$ into the smart contract. Anyone who submits a witness proving the mining of a block on top of the adversarial block $B_\mathcal{A}$ can withdraw $\texttt{br}$. A valid witness includes a Bitcoin $\texttt{nonce}$, a Merkle root $\texttt{MR}$, a Bitcoin coinbase transaction $tx$, a Merkle inclusion $\texttt{proof}$, and a payout address $\texttt{add}$ in the host blockchain, satisfying the following conditions:

\begin{itemize}
    \item A valid block is mined on top of adversarial block $B_\mathcal{A}$: 
    \begin{equation} \label{eq:hash}
        \texttt{Hash}\big( \texttt{Hash}(B_\mathcal{A}), \texttt{MR}, \texttt{nonce}\big) \leq \texttt{Target} \enspace .
    \end{equation}
    \item The coinbase transaction is included in the list of block transactions, meaning a valid Merkle inclusion $\texttt{proof}$ is available to demonstrate that the coinbase transaction $tx$ is a leaf of a Merkle tree with Merkle root $\texttt{MR}$.
    \item The payout address $\texttt{add}$ is included in the coinbase transaction.
\end{itemize}

If the witness transaction satisfies the conditions above, the bribe $\texttt{br}$ will be transferred from the smart contract deposit to the payout address $\texttt{add}$. However, if the adversary submits a witness proving the mining of a block on top of the non-adversarial block $B_{\overline{\mathcal{A}}}$, it can recollect $\texttt{br}$. Upon submission of a valid witness of mining on top of the non-adversarial block $B_{\overline{\mathcal{A}}}$, the bribe will be transferred to the adversarial payout address $\texttt{add}_\mathcal{A}$ stored in the contract.

\subsection{Duration and Cost of the Initial Loss Period in Selfish Mining}\label{appendix:selfish_mining_time_cost}
As discussed in the literature~\cite{grunspan2018profitability, grunspan2023profit, sarenche2024time}, all mining destruction attacks, including selfish mining, face an initial period of financial loss before the mining difficulty adjusts. During the first epoch of selfish mining, some adversarial blocks may become orphaned, while the adversarial block generation rate remains unchanged. Therefore, due to the loss of some adversarial blocks, the adversarial mining pool incurs a profit loss before the difficulty adjustment. However, after the difficulty adjustment at the end of the attack's first epoch, the mining difficulty reduces, leading to an increase in the adversarial block generation rate. If this increase in the block generation rate can compensate for the loss of orphaned adversarial blocks, the selfish mining attack can become profitable.

In Fig.~\ref{fig:profit_lag_selfish}, we depict the normalized revenue advantage of the selfish mining attack for a real-world mining pool with mining share $0.29033$ as a function of time for different values of incentivizing factor $\epsilon$. The revenue advantage is defined as the difference between the attack revenue and honest mining revenue since the launch of the attack~\cite{grunspan2023profit}. As shown in Fig.~\ref{fig:profit_lag_selfish}, during the first epoch of the attack, which lasts longer than the standard 2-week epoch duration, the revenue advantage decreases, indicating a financial loss for the adversarial mining pool. However, after the difficulty adjustment at the end of the first epoch, the revenue advantage of the selfish mining attack begins to increase. A couple of weeks after the attack begins, the revenue advantage finally turns positive, meaning that selfish mining can be considered profitable from that point onward. The initial duration in which the revenue advantage is negative, referred to as the profit lag in~\cite{grunspan2023profit}, can be considered the initial loss duration of a selfish mining attack. As can be seen in Fig.~\ref{fig:profit_lag_selfish}, the initial cost of the attack increases with a higher incentivizing factor. This is because, in an environment with a higher incentivizing factor, the petty-compliant pools are more likely to behave honestly, and consequently, the adversary must pay a higher bribe to incentivize them.

\begin{figure}[t]
    \centering
    \includegraphics[height=2.5in]{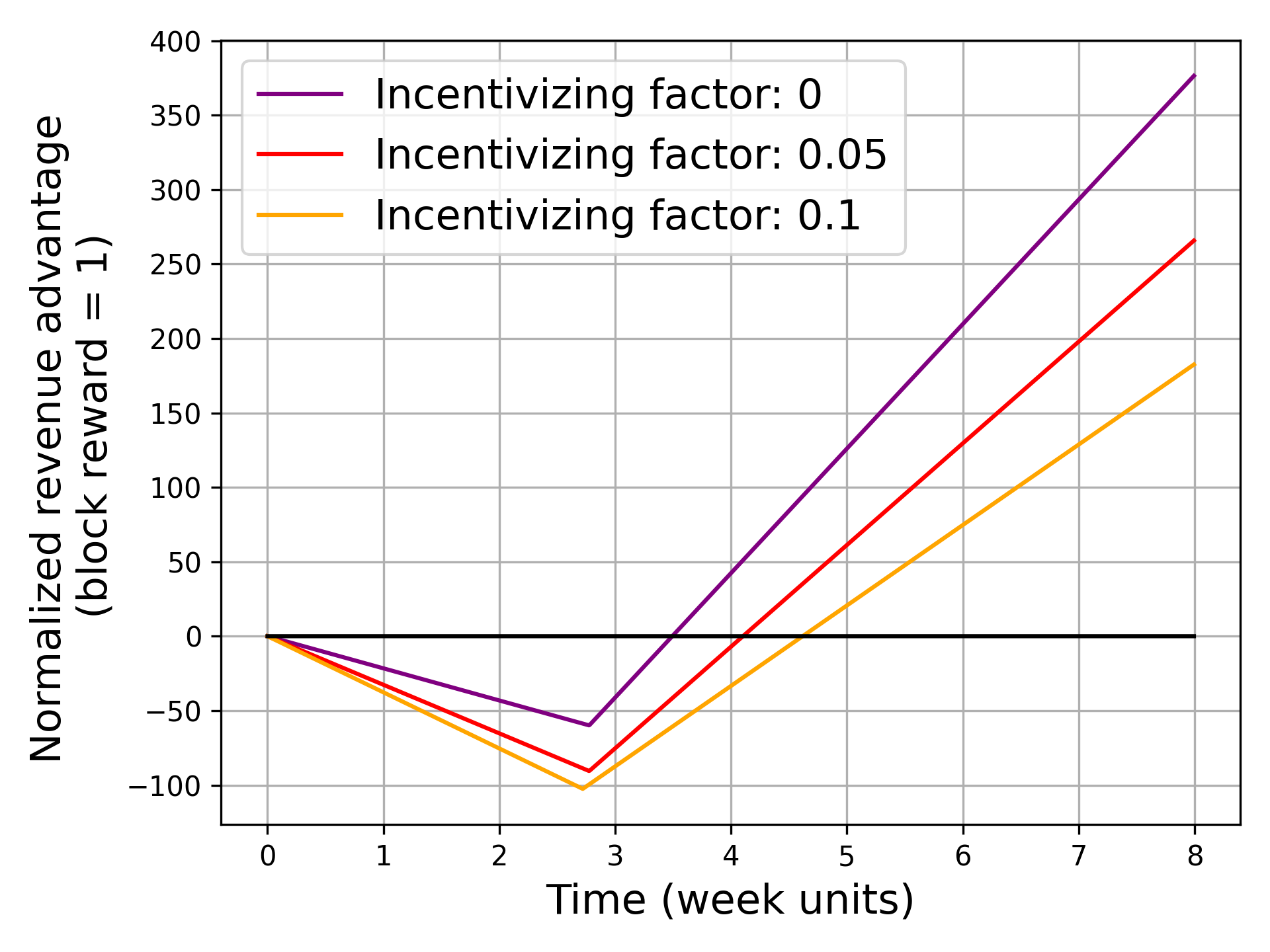}
    \caption{The normalized revenue advantage of selfish mining over time (adversarial mining power share: $0.29033$).}
    \label{fig:profit_lag_selfish}
\end{figure}

\section{Bribery Attack: Supplementary}
\subsection{Proof of Theorem~\ref{theorem_bribery attack_long_term}} \label{appendix:bribery_theorem_proof}
\begin{proof}
    Assume $L = 2016$ denotes the epoch length, $R$ denotes the block reward, and $\lambda$ denotes the block mining rate (the number of blocks mined per unit of time) when all the mining pools follow the honest strategy. We first calculate the time-averaged profit of the adversarial pool under the bribery attack. Since the average number of target blocks per epoch for which a rival block is mined is equal to $k$, we can conclude that the average number of non-adversarial blocks that get orphaned per epoch is also $k$. Given that, on average, the same number of non-adversarial blocks are orphaned in each epoch, the mining difficulty is well-adjusted, indicating that the mining rate in all epochs is the same as $\lambda$. Therefore, the \texttt{duration} of each epoch is equal to $L/\lambda$.
    In the absence of an attack, out of $L$ canonical blocks per epoch, $\alpha L$ blocks are adversarial. Under the bribery attack, however, $k$ blocks are orphaned, where all of them are non-adversarial. The $k$ non-adversarial orphaned blocks need to be replaced with $k$ new blocks, of which $\alpha k$ blocks are adversarial. This implies that under the bribery attack, the number of adversarial blocks per epoch increases to $\alpha(L+k)$ blocks.
    Therefore, the adversarial mining pool receives a total \texttt{Rev} of $\alpha_\mathcal{A}(L + k)R$ per epoch. The cost that the adversary incurs in each epoch includes the mining cost and the bribery cost. Here, we ignore the mining cost in our calculation as it is the same as when mining honestly. The total bribery \texttt{Cost} in each epoch is $k \texttt{br} R$ since the adversary needs to pay a normalized bribe of $\texttt{br}$ for each of the $k$ successful attempts of the bribery attack. Therefore, the time-averaged profit of the bribery attack can be obtained as follows:
    \begin{equation}
        \texttt{profit}_\texttt{br} = \frac{\texttt{Rev}-\texttt{Cost}}{\texttt{duration}} = \frac{\lambda\big(\alpha_\mathcal{A}(L+k)R - k\texttt{br}R\big)}{L} = \lambda \alpha_\mathcal{A}R + \frac{\lambda k R}{L} (\alpha_\mathcal{A}-\texttt{br}) \enspace.
    \end{equation}
    The time-averaged profit of honest mining is equal to $\texttt{profit}_\mathcal{H}=\lambda \alpha_\mathcal{A}R$. The bribery attack can strongly dominate the honest strategy if the following inequalities hold: 
    \begin{equation}
        \texttt{profit}_\texttt{br} > \texttt{profit}_\mathcal{H} \Longleftrightarrow \lambda \alpha_\mathcal{A}R + \frac{\lambda k R}{L} (\alpha_\mathcal{A}-\texttt{br}) > \lambda \alpha_\mathcal{A}R \Longleftrightarrow \alpha_\mathcal{A} > \texttt{br} \enspace .
    \end{equation}
\null\hfill\qedsymbol\end{proof}
\subsection{Bribery Attack Under the Assumption of Unknown Miners}\label{appendix:bribery_unknown}
In Section~\ref{section: bribe_attack_known_miners}, we analyzed the bribery attack under the assumption that the miner of a target block is known. However, a mining pool may choose not to reveal its identity in the coinbase transaction. This makes it difficult for the adversarial and remaining petty-compliant mining pools to detect the miner of the target block. In this section, we analyze the bribery attack in the setting where the miner of the target block is unknown.
In the following, we only discuss the smart contract-based attack.

\subsubsection{Bribery Attack Using Smart Contracts}
\paragraph{Attack Description} An adversarial mining pool with mining share $\alpha_\mathcal{A}$ can conduct the smart contract-based bribery attack if there is at least one petty-compliant mining pool $p_j \in \mathcal{P}$ whose residual centralization factor is less than the mining share of the adversarial mining pool. Otherwise, the attack is not feasible.

Assume that block $B_1$ represents the head of the canonical chain, and its corresponding miner is not known to the adversarial and petty-compliant mining pools. Block $B_1$ is considered a target block of the bribery attack if it is a non-adversarial block. Let $\beta_j$ denote the residual centralization factor with respect to the mining pool $p_j$, and $\texttt{min}(\beta_j)$ denote the minimum residual centralization factor among all the mining pools. The adversarial mining pool deposits two normalized bribes, $\texttt{br}_1$ and $\texttt{br}_2 = \epsilon$, in the smart contract, where $\texttt{br}_1$ can take any value in the range $[\texttt{min}(\beta_j) + \epsilon, \alpha_\mathcal{A}-2\epsilon)$. The other details are similar to the smart contract-based attack in the known miner setting.

\begin{lemma} \label{lemma:unknown_semi-rational_incentivize}
    Assume block $B_1$ which denotes the head of the canonical chain is mined by an unknown mining pool. Assume further that block $B_1$ is a target of the smart contract-based bribery attack that offers two normalized bribes $\texttt{br}_1$ and $\texttt{br}_2 = \epsilon$. In an environment with the incentivizing factor $\epsilon$, a petty-compliant mining pool $p_j$ which is not the the miner of block $B_1$ is incentivized to follow the bribery attack strategy if $\texttt{br}_1 \ge \beta_j + \epsilon$, where $\beta_j$ denotes the residual centralization factor w.r.t. the mining pool $p_j$.
\end{lemma}
\begin{proof}
    Let $\mathcal{P}$ denote the set of all mining pools (both semi-rational and adversarial). If petty-compliant mining pool $p_j \in \mathcal{P}$ mines on top of the target block, its expected return for the next block is equal to $r_1 = \alpha_j R$. If it tries to mine a rival block, its expected return is equal to $r_2 = \alpha_j p_\texttt{success} (1+\texttt{br}_1)R + \alpha_j (1-p_\texttt{success}) \texttt{br}_1 R$, where $p_\texttt{success}$ represents the probability that the mining pool $p_j$ wins the fork race. Let $\mathcal{P}_{\overline{j}}$ denote the set of all the mining pools excluding $p_j$. The target block $B_i$ can be mined by any of the mining pools in $\mathcal{P}_{\overline{j}}$. If block $B_1$ is mined by mining pool $p_i$, $p_j$ will win the fork race with probability $1-\alpha_i$ as all the remaining mining pools are incentivized to mine on top of the rival block. Therefore, the total probability that the mining pool $p_j$ wins the fork race can be obtained as follows: 
    \begin{equation}
        p_\texttt{success} = \frac{\sum_{\mathcal{P}_{j}} {\alpha_i(1-\alpha_i)}}{1-\alpha_j} = 1 - \beta_j,
    \end{equation}
    where $\beta_j$ is the residual centralization factor w.r.t. the mining pool $p_j$. The mining pool $p_j$ is incentivized to follow the adversarial bribing strategy if $r_2 \ge r_1 + \epsilon \alpha_jR$. We have: 
   \begin{equation}
        \begin{split}
            & r_2\ge r_1 + \epsilon \alpha_j R \Longleftrightarrow \alpha_j (1-\beta_j)(1+\texttt{br}_1) R+ \alpha_j \beta_j \texttt{br}_1 R \ge \alpha_j R + \epsilon \alpha_j R \\
            & \Longleftrightarrow \texttt{br}_1 \ge \beta_j + \epsilon \enspace.
        \end{split}
    \end{equation}
\null\hfill\qedsymbol\end{proof}

\begin{lemma}
    In an environment with an unknown miner setting and an incentivizing factor $\epsilon$, an adversarial mining pool with mining share $\alpha_\mathcal{A}$ is incentivized to conduct a smart contract-based bribery attack if there exists at least one petty-compliant mining pool such as $p_j$ with residual centralization factor $\beta_j$ for which the following inequality holds:
    \begin{equation}
         \alpha_\mathcal{A} > \beta_j + 2\epsilon \enspace.
    \end{equation}
\end{lemma}

\begin{proof}
    According to Theorem~\ref{theorem_bribery attack_long_term}, the bribery attack can be profitable as long as the bribe that the adversarial mining pool pays out for orphaning a non-adversarial block is less than its mining share. Therefore, since $ \alpha_\mathcal{A} >  \beta_j + 2\epsilon$, the adversarial mining pool can set the total bribe for orphaning a single non-adversarial block to $\texttt{br}=\beta_j + 2\epsilon$ and still profits from the bribery attack. According to Lemma~\ref{lemma:unknown_semi-rational_incentivize}, with a bribe of value $\texttt{br}$, the petty-compliant mining pool $p_j$ is incentivized to follow the adversarial bribing strategy. Since there exists at least one petty-compliant mining pool that follows the adversarial strategy, the attack is profitable. 
\null\hfill\qedsymbol\end{proof}

\subsection{Bribery Attack Using Whale Transactions} \label{appendix:whale_based_known_miners}
In the bribery attack using whale transactions, the adversarial mining pool sets its bribe as the transaction fee of a transaction that is only valid if included in the chain mined on top of the rival block. This transaction becomes invalid if included in the chain mined on top of the target block. We first explain the attack.

\paragraph{Attack Description} 
Let $\alpha_\mathcal{A}$ represent the mining share of the adversarial mining pool $p_\mathcal{A}$. Assume block $B_1$ denotes the head of the canonical chain and is mined by mining pool $p_i$ with mining share $\alpha_i$. From the perspective of the adversarial mining pool, block $B_1$ can be considered a valid target block for the bribery attack if the following condition holds: 
\begin{itemize}
    \item $B_1$ is a non-adversarial block.
    \item $\alpha_\mathcal{A} > \frac{\alpha_i}{1-\alpha_i} + \epsilon (\frac{1}{1-\alpha_i}+1)$.
\end{itemize}
If the conditions above are not satisfied, the attack is not applicable, and the adversarial mining pool waits for the next block. If block $B_1$ is a valid target block, the adversarial mining pool proceeds with the bribery attack.

We assume an adversarial transaction $tx_\mathcal{A}$ is already included in block $B_1$. The adversarial mining pool then creates two bribes, $\texttt{br}_1$ and $\texttt{br}2$, and sets them as the transaction fees for transactions $tx_1$ and $tx_2$, respectively. These transactions are only valid if $tx_\mathcal{A}$ is not included beforehand in the canonical chain. The adversarial mining pool publishes transactions $tx_1$ and $tx_2$ once block $B_1$ is published. $\texttt{br}_1$ serves as a bribe for the miner of the rival block. $\texttt{br}_2$ serves as a bribe for the miner of the supporting block, the block that is mined on top of the rival block in the case of a fork race between the rival and target blocks. Once transactions $tx_1$ and $tx_2$ are published in the network, a petty-compliant mining pool should choose one of the following two strategies: it can either continue mining on top of block $B_1$ without including transaction $tx_1$, or it can try to mine a rival block that includes transaction $tx_1$\footnote{The miner of the rival block can include both transactions $tx_1$ and $tx_2$ in its block. However, it would be a wise decision to leave $tx_2$ as a bribe for the supporting block.}. If the next block $B_2$ is mined on top of block $B_1$, the bribery attack targeting $B_1$ fails, and the adversarial mining pool repeats the same attack on block $B_2$. If the next block $B_2$ is a rival block for $B_1$, the attack is successful.

During this attack, the adversarial mining pool commits to the following strategy: if no rival block is published, it mines on top of the target block (the head of the canonical chain). Once a rival block is published, it switches to mining on top of the rival block.

\begin{lemma} \label{lemma:bribing_attack_whale_trransaction}
    Assume block $B_1$ is mined by mining pool $p_i$ with mining share $\alpha_i$. Assume further that block $B_1$ is a target of the whale transaction-based bribery attack that offers two normalized bribes $\texttt{br}_1 = \frac{\alpha_i + \epsilon}{1-\alpha_i}$ and $\texttt{br}_2 = \epsilon$. In an environment with the incentivizing factor $\epsilon$, all the petty-compliant mining pools except $p_i$ are incentivized to mine a rival block for $B_1$.
\end{lemma}
\begin{proof}
    Let $\mathcal{P}_{\overline{i}}$ denote the set of all mining pools excluding $p_i$. Under the bribery attack, the petty-compliant mining pools need to choose between mining on top of block $B_1$ or mining on top of its parent block to mine a rival block.

    Let $B_\texttt{next}$ denote the next block that is mined during the attack. We first determine the best strategy that a non-adversarial mining pool $p_j \in \mathcal{P}_{\overline{i}}$ should take once block $B_\texttt{next}$ is published. If block $B_\texttt{next}$ is mined on top of $B_1$, the attack fails and all mining pools, including $p_j$, will continue mining on top of $B_1$. However, if $B_\texttt{next}$ is mined on top of $B_0$, a fork race will occur between $B_1$ and $B_\texttt{next}$. During this fork race, since there is a bribe of $\texttt{br}_2 = \epsilon$ on top of block $B_\texttt{next}$, all the mining pools in $\mathcal{P}_{\overline{i}}$, including $p_j$, will mine on top of block $B_1$, and only $p_i$ will mine on top of block $B_1$.

    Knowing the best strategy for pool $p_j$ after the publishing of block $B_\texttt{next}$, we can obtain the optimal strategy for $p_j$ to follow before $B_\texttt{next}$ is published. To this end, we compare the average expected return under the following two strategies for the next block:
    \begin{itemize}
        \item Strategy $\pi_1$: $p_j$ mines on top of block $B_1$ and does not include any of the transactions $tx_1$ or $tx_2$ in the block, as these transactions are not valid in a chain mined on top of block $B_1$.
        \item Strategy $\pi_2$: $p_j$ mines on top of block $B_0$ and includes only $tx_1$ with a transaction fee of $\texttt{br}_1R$ in the block. In this strategy, the mining pool $p_j$ refrains from including $tx_2$ in its block, leaving it as a bribe on top of its potential block in the upcoming fork race.
    \end{itemize}

    Under strategy $\pi_1$, the expected return of mining pool $p_j$ for the next block is equal to $r_1 = \alpha_j R$.
    Under strategy $\pi_2$, the expected return of mining pool $p_j$ for the next block can be obtained as follows:
    \begin{equation}
        r_2 = \alpha_j \cdot R_\texttt{rival} \cdot p_\texttt{success} \enspace ,
    \end{equation}
    where $R_\texttt{rival}$ is the total block reward of the rival block, and $p_\texttt{success}$ is the probability that the rival block wins the fork race against the target block $B_1$. The total block reward $R_\texttt{rival}$ consists of both the block reward and the transaction fee of $tx_1$, implying that $R_\texttt{rival} = (1+\texttt{br}_1)R$. The probability that the rival block wins the fork race is $1-\alpha_i$, as the only mining pool mining on top of the target block is $p_i$. All the remaining petty-compliant mining pools choose to mine on top of the rival block due to the bribe $\texttt{br}_2 = \epsilon$ available on top of it. Additionally, the adversarial mining pool mines on top of the rival block since it has already committed to it. Therefore, we obtain $r_2 = \alpha_j (1-\alpha_i)(1+\texttt{br}_1)R \enspace$. In an environment with incentivizing factor $\epsilon$, mining pool $p_j$ chooses the strategy $\pi_2$ over strategy $\pi_1$ if $r_2\ge r_1 + \epsilon \alpha_j R$. We have:
    \begin{equation}
        r_2\ge r_1 + \epsilon \alpha_j R \Longleftrightarrow \alpha_j (1-\alpha_i)R(1+\texttt{br}_1) \ge \alpha_j R + \epsilon \alpha_j R \Longleftrightarrow 
        \texttt{br}_1 \ge \frac{\alpha_i + \epsilon}{1-\alpha_i} \enspace.
    \end{equation}
\null\hfill\qedsymbol\end{proof}

\begin{theorem} \label{theorem_bribery_whale}
    In an environment with an incentivizing factor $\epsilon$, an adversarial mining pool with mining share $\alpha_\mathcal{A}$ is incentivized to conduct a whale transaction-based bribery attack on any target block $B$ if block $B$ is mined by a mining pool with mining share $\alpha_i$, where $\alpha_\mathcal{A} > \frac{\alpha_i}{1-\alpha_i} + \epsilon \left(\frac{1}{1-\alpha_i} + 1\right)$.
\end{theorem}
\begin{proof}
    According to Lemma~\ref{lemma:bribing_attack_whale_trransaction}, in an environment with an incentivizing factor $\epsilon$, any block $B$ mined by a mining pool with mining share $\alpha_i$ is susceptible to a whale transaction-based bribery attack if normalized bribes are set as $\texttt{br}_1 = \frac{\alpha_i + \epsilon}{1-\alpha_i}$ and $\texttt{br}_2 = \epsilon$. Therefore, to conduct a bribery attack on block $B$, the adversarial mining pool should pay the normalized bribe of $\texttt{br} = \texttt{br}_1 + \texttt{br}_2 = \frac{\alpha_i}{1-\alpha_i} + \epsilon \left(\frac{1}{1-\alpha_i}+1\right)$. According to Theorem~\ref{theorem_bribery attack_long_term}, the bribery attack is profitable for the adversarial mining pool if $\alpha_\mathcal{A} > \texttt{br}$. This completes the proof.
\null\hfill\qedsymbol\end{proof}
According to Theorem~\ref{theorem_bribery_whale}, assuming $\epsilon = 0$, an adversarial mining pool with a mining share $\alpha_\mathcal{A}$ can perform a bribery attack on blocks mined by pools whose mining power is less than $\frac{\alpha_\mathcal{A}}{1+\alpha_\mathcal{A}}$. This implies that even mining pools with mediocre or low amounts of mining power can exploit weaker mining pools through the bribery attack.

\subsection{Proof of Theorem~\ref{theorem_bribery_smart}} \label{appendix:prrof_theorem_bribery_smart}
To prove Theorem~\ref{theorem_bribery_smart}, we first present the following lemma.
\begin{lemma}\label{lemma_bribery_smart_complinat}
    Assume block $B_1$ is mined by mining pool $p_i$ with mining share $\alpha_i$. Assume further that block $B_1$ is a target of the smart contract-based bribery attack that offers two normalized bribes $\texttt{br}_1 = \alpha_i + \epsilon$ and $\texttt{br}_2 = \epsilon$. In an environment with the incentivizing factor $\epsilon$, all the petty-compliant mining pools except $p_i$ are incentivized to follow the bribery attack strategy.
\end{lemma}
\begin{proof}
    Let $\mathcal{P}_{\overline{i}}$ denote the set of all mining pools excluding $p_i$. Under the bribery attack, the petty-compliant mining pools need to choose between mining on top of block $B_1$ or mining on top of its parent block to mine a rival block.

    Let $B_\texttt{next}$ denote the next block that is mined during the attack. We first determine the best strategy that a non-adversarial mining pool $p_j \in \mathcal{P}_{\overline{i}}$ should take once block $B_\texttt{next}$ is published. If block $B_\texttt{next}$ is mined on top of $B_1$, the attack fails and all mining pools, including $p_j$, will continue mining on top of $B_\texttt{next}$. However, if $B_\texttt{next}$ is mined on top of $B_0$, a fork race will occur between $B_1$ and $B_\texttt{next}$. During this fork race, since there is a bribe of $\texttt{br}_2 = \epsilon$ on top of block $B_\texttt{next}$, all the mining pools in $\mathcal{P}_{\overline{i}}$, including $p_j$, will mine on top of block $B_1$, and only $p_i$ will mine on top of block $B_1$. Note that even for the adversarial pool, mining on top of $B_{\texttt{next}}$ weakly dominates mining on top of $B_1$. If the adversarial mining pool manages to mine the next block on top of $B_{\texttt{next}}$, the bribe $\texttt{br}_2 = \epsilon$ will be cleared by the adversary, and there is no need to pay that.

    Knowing the best strategy for mining pool $p_j \in \mathcal{P}_{\overline{i}}$ after the publishing of block $B_\texttt{next}$ is to mine on top of that, we can obtain the optimal strategy for $p_j$ to follow before $B_\texttt{next}$ is mined. To this end, we compare the average expected return under the following two strategies for the next block:
    \begin{itemize}
        \item Strategy $\pi_1$: $p_j$ mines on top of block $B_1$.
        \item Strategy $\pi_2$: $p_j$ mines on top of block $B_0$.
    \end{itemize}

    Under strategy $\pi_1$, the expected return of mining pool $p_j$ for the next block is equal to $r_1 = \alpha_j R$.
    Under strategy $\pi_2$, the expected return of mining pool $p_j$ for the next block can be obtained as follows:
    \begin{equation}
        r_2 = \alpha_j R_\texttt{success} p_\texttt{success} + \alpha_j R_\texttt{failure} \cdot p_\texttt{failure}\enspace .
    \end{equation}
    The expected return is composed of two parts: the return when the mining pool $p_j$ wins the fork race against target block $B_1$, and the return when the mining pool $p_j$ loses the fork race. According to the proof of Lemma~\ref{lemma:bribing_attack_whale_trransaction}, the expected return under winning the fork race is equal to $ R_\texttt{success}=\alpha_j (1-\alpha_i)(1+\texttt{br}_1)R$. However, if the mining pool $p_j$ loses the fork race, despite losing the block reward, it can collect the bribe deposited in the smart contract as it has mined a valid rival block for block $B_1$. Therefore, $R_\texttt{failure} = \texttt{br}_1 R$. As a result, the total expected return can be obtained as follows:
    \begin{equation}
        r_2 = \alpha_j (1-\alpha_i)(1+\texttt{br}_1) R+ \alpha_j \alpha_i \texttt{br}_1 R \enspace .
    \end{equation}
    In an environment with incentivizing factor $\epsilon$, mining pool $p_j$ follows the bribery attack strategy over the honest strategy if $r_2\ge r_1 + \epsilon \alpha_j R$. We have:
    \begin{equation}
        \begin{split}
            & r_2\ge r_1 + \epsilon \alpha_j R \Longleftrightarrow \alpha_j (1-\alpha_i)(1+\texttt{br}_1) R+ \alpha_j \alpha_i \texttt{br}_1 R \ge \alpha_j R + \epsilon \alpha_j R \\
            &\Longleftrightarrow \texttt{br}_1 \ge \alpha_i + \epsilon \enspace.
        \end{split}
    \end{equation}
\null\hfill\qedsymbol\end{proof}

\begin{proof}[Proof of Theorem~\ref{theorem_bribery_smart}]
    According to Lemma~\ref{lemma_bribery_smart_complinat}, in an environment with an incentivizing factor $\epsilon$, any block $B$ mined by a mining pool with mining share $\alpha_i$ is susceptible to a smart contract-based bribery attack if normalized bribes are set as $\texttt{br}_1 = \alpha_i + \epsilon$ and $\texttt{br}_2 = \epsilon$. Therefore, to conduct a bribery attack on block $B$, the adversarial mining pool should pay the normalized bribe of $\texttt{br} = \texttt{br}_1 + \texttt{br}_2 = \alpha_i + 2\epsilon$. According to Theorem~\ref{theorem_bribery attack_long_term}, the bribery attack is profitable for the adversarial mining pool if $\alpha_\mathcal{A} > \texttt{br}$. This completes the proof.
\null\hfill\qedsymbol\end{proof}
\subsection{Markov Chain Analysis of Bribery Attack}\label{appendix:markov_chain_bribery}
To analyze the bribery attack, we divide the mining pools into three groups: the adversarial mining pool, the target mining pools, and the non-target mining pools. The target mining pools are those
whose blocks are considered targets for the bribery attack. The non-target mining pools are those
whose blocks are not considered targets for the bribery attack.
We denote by $\alpha_\mathcal{A}$ the mining share of the adversarial pool. Let $\mathcal{P}_\texttt{target}$ denote the set of $N$ target mining pools. For each target mining pool $p_i \in \mathcal{P}_\texttt{target}$, we denote its mining share by $b_i$. We denote by $b$ the total mining share of all the target mining pools in $\mathcal{P}_\texttt{target}$. 
We define value $\beta$ as follows:
\begin{equation}
    \beta = \sum_{p_i \in \mathcal{P}_\texttt{target}} {\frac{b_i}{1-b_i}}
\end{equation}

\noindent\textbf{States:} State $S_0$ denotes the state in which the tip of the chain is not a target block and all the pools are mining on the tip of the chain. State $S_1^i$ for $i \in [N]$\footnote{We define $[N]$ to be $\{1, 2,\ldots, N\}$.} denotes a state in which the tip of the chain is a target block mined by the mining pool $p_i$ and the adversary has placed 2 normalized bribes $\texttt{br}_1 = b_i + \epsilon$ and $\texttt{br}_2 = \epsilon$ for orphaning the chain tip. According to Theorem~\ref{theorem_bribery_smart}, these bribes incentivizes all the petty-compliant mining pools except for the target pool $p_i$ to mine a rival block. State $S_2^i$ for $i \in [N]$ denotes a state in which a rival block is mined for the target block of mining pool $p_i$. 
Let $P_0$, $P_1^i$, and $P_2^i$, for $i \in [N]$, denote the probability of being at states $S_0$, $S_1^i$, and $S_2^i$, respectively.

We assume the smart contract of the bribery attack is designed such that $\texttt{br}_1$ is payable only upon mining a rival block, and $\texttt{br}_2$ is payable only upon mining a non-adversarial block supporting the rival block. In all other cases, the adversarial mining pool is able to clear the bribes.

\noindent\textbf{State transitions:} 
\begin{itemize}
    \item Let $S_0$ be the current state. If the adversarial mining pool mines a new tip block, it immediately publishes the block, and all the mining pools accept the block. The system remains in the same state.
    \item Let $S_0$ be the current state. If a non-target mining pool mines and publishes a new tip block, all the mining pools accept the block. The system remains in the same state.
    \item Let $S_0$ be the current state. If a target mining pool $p_i$ mines a new tip block $B_1^\texttt{target}$, this block is considered a target for the bribery attack. The adversarial mining pool places two normalized bribes $\texttt{br}_1 = b_i + \epsilon$ and $\texttt{br}_2 = \epsilon$ for orphaning $B_1^\texttt{target}$. The system transitions to state $S_1^i$. At the new state $S_1^i$, mining pool $p_i$ and the adversarial mining pool mine on top of block $B_1^\texttt{target}$, while the other mining pools mine on top of the parent of block $B_1^\texttt{target}$ to mine a rival block.
    \item Let $S_1^i$ be the current state. If the adversarial mining pool mines the next block, which is a block that extends $B_1^\texttt{target}$, it immediately publishes the block and clears the previous bribes. In this case, all mining pools adopt the chain ending at the adversarial block, and the system transitions to state $S_0$.
    \item Let $S_1^i$ be the current state. If any petty-compliant pool except $p_i$ mines the next block, which is a rival block to $B_1^\texttt{target}$, it immediately publishes the block and collects the bribe $\texttt{br}_1 = b_i + \epsilon$. The system transitions to state $S_2^i$. At this new state $S_2^i$, all the mining pools except $p_i$ mine on the rival block.
    \item Let $S_1^i$ be the current state. If mining pool $p_i$ mines the next block $B_2^\texttt{target}$, which extends $B_1^\texttt{target}$, all the mining pools accept the chain ending at the previous target block $B_1^\texttt{target}$, and the adversary clears the previous bribes. The adversarial mining pool places new bribes for orphaning the new target block $B_2^\texttt{target}$. The system remains in the same state.
    \item Let $S_2^i$ be the current state. If the adversarial mining pool mines the next block, which is a block that extends the rival block, the adversary immediately publishes the block and clears the second bribe $\texttt{br}_2 = \epsilon$. All the mining pools adopt the rival fork ending at the adversarial block\footnote{According to Lemma~\ref{lemma_0.43}, if the mining share of the target mining pool $p_i$ is less than $0.4302$, it adopts the longer rival fork}, and the system transitions to state $S_0$.
    \item Let $S_2^i$ be the current state. If a non-target mining pool mines the next block, which is a block that extends the rival block, it publishes the block and collects the bribe $\texttt{br}_2 = \epsilon$. All the mining pools adopt the rival fork ending at the non-target block, and the system transitions to state $S_0$.
    \item Let $S_2^i$ be the current state. If a target mining pool $p_j$, where $p_j$ differs from $p_i$, mines the next block $B_2^\texttt{target}$, which is a block that extends the rival block, it publishes the block and collects the bribe $\texttt{br}_2 = \epsilon$. All the mining pools adopt the rival fork ending at the rival block. The adversarial mining pool places new bribes for orphaning the new target block $B_2^\texttt{target}$, and the system transitions to state $S_1^j$.
    \item Let $S_2^i$ be the current state. If the target mining pool $p_i$ mines the next block $B_2^\texttt{target}$, which is a block that extends the target block $B_1^\texttt{target}$, it publishes the block, and the adversarial mining pool clears the bribe $\texttt{br}_2 = \epsilon$. All the mining pools adopt the target fork ending at the first target block $B_1^\texttt{target}$. The adversarial mining pool places new bribes for orphaning the new target block $B_2^\texttt{target}$, and the system transitions to state $S_1^i$.
\end{itemize}

The state transition probabilities, the expected profit of the adversarial pool for each state transition, and the number of blocks added for each state transition are described in Table~\ref{tab:markov_bribery}.
\begin{table}[t]
    \centering
    \caption{State transitions under the bribery attack.}
    \renewcommand{\arraystretch}{1.5} 
    \resizebox{\textwidth}{!}{
    \begin{tabular}{|c|c|c|c|c|}
        \toprule
        \textbf{Initial State} & \textbf{Next State} & \textbf{Transition Probability} & \textbf{Total Blocks Added} & \textbf{Adversary Expected Profit} \\
        \midrule
        \multirow{2}{*}{$S_0$} & $S_0$ & $\alpha_\mathcal{A}$ & $1$ & $1$ \\ 
        \cline{2-5} 
        & $S_0$ & $1-b-\alpha_\mathcal{A}$ & $1$ & $0$ \\
        \cline{2-5} 
        & $S_1^i$ & $b_i$ & $0$ & $0$ \\ 
        \midrule
        \multirow{3}{*}{$S_1^i$} & $S_0$ & $\alpha_\mathcal{A}$ & $2$ & $1$ \\ 
        \cline{2-5} 
        & $S_1^i$ & $b_i$ & $1$ & $0$ \\ 
        \cline{2-5} 
        & $S_2^i$ & $1-\alpha_\mathcal{A}-b_i$ & $0$ & $-b_i-\epsilon$ \\ 
        \midrule
        \multirow{2}{*}{$S_2^i$} & $S_0$ & $\alpha_\mathcal{A}$ & $2$ & $1$ \\
        \cline{2-5} 
        & $S_0$ & $1-b-\alpha_\mathcal{A}$ & $2$ & $-\epsilon$ \\
        \cline{2-5} 
        & $S_1^i$ & $b_i$ & $1$ & $0$ \\ 
        \cline{2-5} 
        & $S_1^j \enspace (j\ne i)$ & $b_j$ & $1$ & $-\epsilon$ \\
        \bottomrule
    \end{tabular}
    }
    \label{tab:markov_bribery}
\end{table}

Based on the transition probabilities mentioned in Table~\ref{tab:markov_bribery}, we can obtain the state probabilities as follows:
\begin{equation}
    \begin{split}
        & P_0 = \frac{1-b+\alpha_\mathbf{A}\beta}{1+\beta}\\
        & P_1^i = \frac{b_i}{(1-b_i)(1+\beta)} \enspace, \\
        & P_2^i = \frac{b_i(1-\alpha_\mathbf{A}-b_i)}{(1-b_i)(1+\beta)} \enspace .
    \end{split}
\end{equation}
According to Table~\ref{tab:markov_bribery}, the reward share $R_\mathcal{A}$ of the adversarial pool under the bribery attack can be obtained as follows:
\begin{align}
    R_\mathcal{A} = \frac{\alpha_\mathcal{A} - \sum_{i=1}^{N}{\Big(P_1^i (1-\alpha_\mathcal{A}-b_i)(b_i+\epsilon) + P_2^i (1-\alpha_\mathcal{A}-b_i)\epsilon\Big)}}{P_0 (1-b) + \sum_{i=1}^{N}{\Big(P_1^i (b_i + 2\alpha_\mathcal{A}) + P_2^i (2-b)\Big)}} \enspace .
\end{align}

\subsection{Markov Chain Analysis of Undercutting Attack}\label{appendix:markov_chain_undercut}
To analyze the undercutting attack, we divide the mining pools into three groups: the adversarial mining pool, the target mining pools, and the non-target mining pools. The target mining pools are those
whose blocks are considered targets for the undercutting attack. The non-target mining pools are those
whose blocks are not considered targets for the undercutting attack. We use the same notations as in Appendix~\ref{appendix:markov_chain_bribery}. Whenever a target block is mined, the adversarial mining pool attempts to undercut the target block. If the adversary successfully mines a rival block, it places a bribe of $\epsilon$ on top of the rival block to incentivize all mining pools, except the miner of the target block, to mine on top of the rival block.


\noindent\textbf{States:} State $S_0$ denotes the state in which the tip of the chain is not a target block and all the pools are mining on the tip of the chain. State $S_1^i$ for $i \in [N]$ denotes a state in which the tip of the chain is a target block mined by the mining pool $p_i$, and the adversarial mining pool is attempting to undercut the chain tip. State $S_2^i$ for $i \in [N]$ denotes a state in which the adversary has mined a rival block for the target block of mining pool $p_i$ and placed a bribe of $\texttt{br} = \epsilon$ on top of the rival block.
Let $P_0$, $P_1^i$, and $P_2^i$, for $i \in [N]$, denote the probability of being at states $S_0$, $S_1^i$, and $S_2^i$, respectively.

\noindent\textbf{State Transitions:}
\begin{itemize}
    \item Let $S_0$ be the current state. If the adversarial mining pool mines a new tip block, it immediately publishes the block, and all mining pools accept it. The system remains in the same state.
    \item Let $S_0$ be the current state. If a non-target mining pool mines and publishes a new tip block, all mining pools accept it. The system remains in the same state.
    \item Let $S_0$ be the current state. If a target mining pool $p_i$ mines a new tip block $B_1^\texttt{target}$, this block becomes the target for the undercutting attack. The system transitions to state $S_1^i$. At $S_1^i$, all non-adversarial mining pools mine on top of $B_1^\texttt{target}$, while the adversarial mining pool mines on the parent block of $B_1^\texttt{target}$ to create a rival block.
    \item Let $S_1^i$ be the current state. If the adversarial mining pool mines a rival block for $B_1^\texttt{target}$, it immediately publishes the block and places a bribe $\texttt{br} = \epsilon$ on top of it. The system transitions to state $S_2^i$. At $S_2^i$, all mining pools except $p_i$ mine on top of the rival block.
    \item Let $S_1^i$ be the current state. If a non-target mining pool mines a block extending $B_1^\texttt{target}$, it publishes the block, and all mining pools adopt the chain ending at this non-target block. The system transitions to state $S_0$.
    \item Let $S_1^i$ be the current state. If a target mining pool $p_j$, where $p_j \neq p_i$, mines a block $B_2^\texttt{target}$ extending $B_1^\texttt{target}$, it publishes the block. The adversarial mining pool then adopts $B_1^\texttt{target}$ and attempts to undercut $B_2^\texttt{target}$. The system transitions to state $S_1^j$.
    \item Let $S_1^i$ be the current state. If the target mining pool $p_i$ mines a block $B_2^\texttt{target}$ extending $B_1^\texttt{target}$, it publishes the block. The adversarial mining pool continues to target $B_2^\texttt{target}$ while attempting to undercut it. The system remains in state $S_1^i$.
    \item Let $S_2^i$ be the current state. If the adversarial mining pool mines a block extending the rival block, it immediately publishes the block and clears the bribe $\texttt{br} = \epsilon$. All mining pools adopt the rival fork, and the system transitions to state $S_0$.
    \item Let $S_2^i$ be the current state. If a non-target mining pool mines a block extending the rival block, it publishes the block and collects the bribe $\texttt{br} = \epsilon$. All mining pools adopt the rival fork, and the system transitions to state $S_0$.
    \item Let $S_2^i$ be the current state. If a target mining pool $p_j$, where $p_j \neq p_i$, mines a block $B_2^\texttt{target}$ extending the rival block, it publishes the block and collects the bribe $\texttt{br} = \epsilon$. The adversarial mining pool adopts $B_1^\texttt{target}$ and attempts to undercut $B_2^\texttt{target}$. The system transitions to state $S_1^j$.
    \item Let $S_2^i$ be the current state. If the target mining pool $p_i$ mines a block $B_2^\texttt{target}$ extending $B_1^\texttt{target}$, it publishes the block. The adversarial mining pool clears the bribe $\texttt{br} = \epsilon$, adopts $B_1^\texttt{target}$, and attempts to undercut $B_2^\texttt{target}$. The system transitions to state $S_1^i$.
\end{itemize}

The state transition probabilities, the expected profit of the adversarial pool for each state transition, and the number of blocks added for each state transition are described in Table~\ref{tab:markov_undercut}.
\begin{table}[t]
    \centering
    \caption{State transitions under the undercutting attack.}
    \renewcommand{\arraystretch}{1.5} 
    \resizebox{\textwidth}{!}{
    \begin{tabular}{|c|c|c|c|c|}
        \toprule
        \textbf{Initial State} & \textbf{Next State} & \textbf{Transition Probability} & \textbf{Total Blocks Added} & \textbf{Adversary Expected Profit} \\
        \midrule
        \multirow{2}{*}{$S_0$} & $S_0$ & $\alpha_\mathcal{A}$ & $1$ & $1$ \\ 
        \cline{2-5} 
        & $S_0$ & $1-b-\alpha_\mathcal{A}$ & $1$ & $0$ \\
        \cline{2-5} 
        & $S_1^i$ & $b_i$ & $0$ & $0$ \\ 
        \midrule
        \multirow{3}{*}{$S_1^i$} & $S_2^i$ & $\alpha_\mathcal{A}$ & $0$ & $0$ \\ 
        \cline{2-5} 
        & $S_0$ & $1-b-\alpha_\mathcal{A}$ & $2$ & $0$ \\ 
        \cline{2-5} 
        & $S_1^j (j \in [N])$ & $b_j$ & $1$ & $0$ \\ 
        \midrule
        \multirow{2}{*}{$S_2^i$} & $S_0$ & $\alpha_\mathcal{A}$ & $2$ & $2$ \\
        \cline{2-5} 
        & $S_0$ & $1-b-\alpha_\mathcal{A}$ & $2$ & $1-\epsilon$ \\
        \cline{2-5} 
        & $S_1^i$ & $b_i$ & $1$ & $0$ \\ 
        \cline{2-5} 
        & $S_1^j \enspace (j\ne i)$ & $b_j$ & $1$ & $1-\epsilon$ \\
        \bottomrule
    \end{tabular}
    }
    \label{tab:markov_undercut}
\end{table}
Based on the transition probabilities mentioned in Table~\ref{tab:markov_undercut}, we can obtain the state probabilities as follows:
\begin{equation}
    \begin{split}
        & P_0 = 1 -b(1+\alpha_\mathcal{A}) \\
        & P_1^i = b_i \enspace, \\
        & P_2 = \alpha_\mathcal{A} b_i \enspace .
    \end{split}
\end{equation}
According to Table~\ref{tab:markov_undercut}, the reward share $R_\mathcal{A}$ of the adversarial pool under the undercutting attack can be obtained as follows:
\begin{align}
    R_\mathcal{A} = \frac{P_0 \alpha_\mathcal{A} + \sum_{i=1}^{N} {\Big(P_2^i(2\alpha_\mathcal{A}+(1-\alpha_\mathcal{A}-b_i)(1-\epsilon))\Big)}}{P_0 (1-b) + \sum_{i=1}^{N} {\Big(P_1^i (2-2\alpha_\mathcal{A}-b) + P_2^i (2-b)\Big)}} \enspace .
\end{align}
\subsection{Naive Solutions to Mitigate the Bribery Attack}\label{appendix:naive_solution}
The introduction of the bribery attack raises the question of how a petty-compliant mining pool can defend against this attack. As discussed in the previous section, conducting a bribery attack requires a higher budget in a setting with unknown miners. This suggests that if a petty-compliant mining pool does not reveal its identity in the block, the probability of being targeted by a bribery attack decreases. This strategy is particularly useful for smaller mining pools, which are more susceptible to bribery attacks. However, not disclosing its identity in the coinbase transaction does not guarantee complete concealment for the mining pool. For instance, petty-compliant mining pools may conduct network analyses to identify the miner of a targeted block. Moreover, if the remaining petty-compliant pools, especially the largest ones, continue to reveal their identity in their blocks, it becomes easier to detect whether a block is mined by one of the largest pools or by a smaller pool prone to the bribery attack.

Another potential solution for petty-compliant mining pools is to engage in commitment games. To protect their blocks, they can submit commitments offering incentives to mining pools that support their block in the event of a fork race. However, this approach ties Bitcoin's security to commitments that typically occur outside the Bitcoin framework. Furthermore, the presence of such commitments may encourage other petty-compliant mining pools to initiate fork races to access the funds deposited in those commitments. For instance, other petty-compliant mining pools may intentionally disrupt block propagation in the network to increase the likelihood of a fork race occurring.

\subsection{Smart Contract-Based vs. Whale Transaction-Based Bribery Attacks}
One of the advantages of the smart contract-based bribery attack is that there is no need for the adversarial mining pool to commit to its strategy to incentivize the petty-compliant mining pools to follow the attack strategy.

In the whale transaction-based attack, the mining pool $p_j$ must trust the assumption that the adversarial mining pool will mine on top of the rival block rather than the target block in the case of a fork race. This assumption is necessary as the dominant strategy for the adversarial mining pool is to mine on top of the target block. This is because once a rival block is mined and a fork race is created, the attack is considered successful from the perspective of the adversarial mining pool regardless of the fork race outcome, as one of the blocks will be orphaned. In this case, mining on top of the target block is the dominant strategy for the adversarial mining pool because it can invalidate the bribe transaction in the rival block and avoid paying the bribe.
This illustrates that if the adversarial mining pool does not commit to the strategy of mining on top of the rival block in the whale transaction-based attack, the total mining share that mines on top of the targeted block becomes equal to $\alpha_i + \alpha_\mathcal{A}$, resulting in unprofitability of the attack from the perspective of the petty-compliant mining pool $p_j$.

However, in the smart contract-based attack, the adversarial mining pool is not incentivized to mine on top of the target block in the case of a fork race. This is because the bribe $\texttt{br}_1$ will be paid to the miner of the rival block regardless of the fork race result. Therefore, in contrast to the whale transaction-based attack, there is no hope for the adversarial mining pool to reclaim the bribe $\texttt{br}_1$ by mining on top of the target block. Also, regarding the bribe $\texttt{br}_2$, the adversarial mining pool has no preference between the rival block and the target block to mine on top of. If the adversarial mining pool mines on top of the target block, no petty-compliant mining pool can withdraw the bribe $\texttt{br}_2$ from the smart contract. If it mines on top of the rival block, the adversarial mining pool will be eligible to withdraw $\texttt{br}_2$ from the smart contract. Therefore, in both scenarios, the adversarial mining pool can reclaim $\texttt{br}_2$. This indicates that in the smart contract-based attack, mining on top of the rival block dominates mining on top of the target block for the adversarial mining pool. 

To make mining on top of the rival block strongly dominate mining on top of the target block, the adversarial mining pool can increase the bribe for the rival block from $\texttt{br}_1 = \alpha_i + \epsilon$ to $\texttt{br}_1 = \alpha_i + 2\epsilon$, where the petty-compliant miner of the rival block can use the extra bribe of $\epsilon$ to incentivize the adversarial mining pool to mine on top of the rival block. To be more precise, the adversarial mining pool deposits $\texttt{br}_1 = \alpha_i + 2\epsilon$ and $\texttt{br}_2 = \epsilon$ in the smart contract, where both bribes $\texttt{br}_1$ and $\texttt{br}_2$ can be withdrawn by the adversarial mining pool after the expiration of a time lock. During this time lock, if a petty-compliant mining pool mines a valid rival block, it withdraws $\alpha_i + \epsilon$ of bribe $\texttt{br}_1$, and the rest of bribe $\texttt{br}_1$ gets deposited as $\texttt{br}_3 = \epsilon$ in the smart contract without any time lock. The mining pool that mines on top of the rival block is eligible to withdraw both $\texttt{br}_2$ and $\texttt{br}_3$, resulting in a total bribe of $2\epsilon$. In this case, if the adversarial mining pool mines on top of the rival block, it receives $2\epsilon$. But, if it mines on top of the target block, it would only be able to withdraw $\texttt{br}_2 = \epsilon$ after the time lock, indicating that mining on top of the rival block strongly dominates mining on top of the target block.


\subsection{The Impact of DAM on the Profitability of Selfish Mining and Bribery Attacks}\label{appendix:mitigation_DAM_selfish_bribery}
The profitability of an adversarial mining pool that engages in selfish mining or bribery attacks is mainly due to the destruction of a portion of the network's mining power. This mining power destruction leads to the reduced total effective mining power of the network. To balance the network's throughput, the difficulty adjustment mechanism lowers the mining difficulty to align with the reduced active mining power. As a result of this reduction in mining difficulty, the block generation rate of the adversarial poll increases, thereby enhancing its profitability.

Each Bitcoin epoch is defined as the duration in which a total of $L=2016$ blocks are added to the canonical chain. At the end of each epoch, the difficulty adjustment mechanism sets the mining difficulty for the upcoming epoch based on the duration of the past epoch. If a bribing or selfish mining attack occurs during an epoch, some blocks become orphaned, leading to an increase in the epoch's duration and a decrease in the canonical block generation rate. As a result, the difficulty adjustment mechanism at the end of the epoch reduces the mining difficulty for the next epoch to restore the canonical block generation rate. Selfish mining and bribery attacks exploit a bottleneck in the difficulty adjustment mechanism, which fails to distinguish between a portion of the mining power being offline or being wasted due to an attack. The current Bitcoin difficulty adjustment mechanism balances the mining difficulty based on the effective mining power, defined as the mining power that extends the canonical chain and does not contribute to any other blocks outside the canonical chain. The advantage of this approach is that it maintains a nearly constant block generation rate and transaction throughput.

However, it is possible to design a difficulty adjustment mechanism that aligns the difficulty based on active mining power, which is defined as all the mining power used to mine both canonical and non-canonical blocks. A difficulty adjustment mechanism that precisely estimates the active mining power can disincentivize adversarial mining pools from conducting selfish or bribery attacks. This is because, no matter how many blocks the adversarial mining pool manages to orphan, the difficulty adjustment mechanism would account for the total active mining power and would not reduce the mining difficulty. This lack of difficulty reduction deters petty-compliant mining pools from engaging in adversarial behaviors such as selfish mining or bribery attacks.

Although the second approach to designing difficulty adjustment mechanisms can be beneficial to a network of petty-compliant miners, it cannot completely mitigate the mining power destruction attacks. First, this type of DAM is ineffective against a Byzantine mining pool that disregards profitability and acts maliciously only to disrupt blockchain progress. For example, a Byzantine mining pool can exploit the active mining power-based DAM to reduce the transaction throughput. Additionally, there is the critical question of how to design a DAM that accurately estimates active mining power. Papers~\cite{sapirshtein2017optimal} and~\cite{sarenche2024time} propose solutions for calculating active mining power by incorporating the counts of both canonical and orphan blocks into the difficulty adjustment mechanism. In these solutions, the difficulty for the next epoch is calculated based on the total number of orphaned and canonical blocks in the previous epoch. Orphan blocks serve as valid proof of mining power loss. However, it is not always possible to provide the DAM with such proof. For instance, if a portion of the mining power is directed toward a platform other than Bitcoin, valid proof of mining power loss cannot be provided unless the Bitcoin blockchain decides to trust external platforms.
\section{Mining Power Distraction Attack} \label{sec:distraction}
In this section, we present an attack referred to as the \emph{mining power distraction attack}, which is capable of destroying a portion of the network's mining power without producing any evidence of mining power destruction.

In the distraction attack, the adversarial mining pool aims to distract a portion of mining power from mining on top of the canonical chain by incentivizing them to mine on another platform. Since the distracted mining power contributes to a work that is only valid outside the Bitcoin framework, one cannot prove the validity of work to the difficulty adjustment mechanism of Bitcoin. 
Therefore, even a DAM that balances the difficulty with the active mining power cannot disincentivize the payoff-maximizing mining pools from performing the mining power distraction attack.

The general idea behind the attack is that once the adversarial mining pool mines a block under Bitcoin difficulty, denoted by \(D_1\), it does not publish the block to the other mining pools. Instead, it conveys the following message to them: "I have already mined the next block of the canonical chain. If you want to know the contents of this block, you must submit a valid proof of work under difficulty \(D_2\), where \(D_2 < D_1\)." We will later discuss the details of how to convey this message and why mining pools should trust it. For now, let us assume that the petty-compliant mining pools accept the existence of a hidden block that extends the canonical chain. This situation puts them in a dilemma: on one hand, if they continue mining on top of the public chain and successfully mine the next block, it will enter a fork race with the adversarial block. On the other hand, if they wish to access the adversarial block, they need to waste a portion of their mining power to solve a puzzle with difficulty \(D_2\). We will demonstrate that by selecting an appropriate value for \(D_2\), the adversarial mining pool can incentivize the petty-compliant mining pools to choose the latter option, thereby wasting some of their mining power on a puzzle defined outside the Bitcoin framework.
Our distraction attack can be considered an extension of the blockchain denial-of-service (BDoS) attack introduced in~\cite{mirkin2020bdos}. The description of the BDoS attack and its differences from our distraction attack are presented in~\ref{Appendix:BDoS}. The BDoS attack primarily affects mining pools whose revenue is only marginally higher than their mining costs and is ineffective against miners with a higher profit-to-loss ratio. In contrast, our distraction attack can disincentivize mining on Bitcoin, even under the assumption that mining costs are negligible.

\noindent\textbf{Attack Description:} Let $D_1$ denote the Bitcoin mining difficulty. Once the adversarial mining pool mines a block, it deploys a smart contract on an external blockchain platform. This smart contract defines a mining puzzle that is similar to the Bitcoin puzzle, with the difference that the puzzle difficulty is $D_2$, which is less than $D_1$. We refer to the solution to this easier puzzle as \textbf{mini-PoW}. The smart contract is designed so that if anyone submits a valid mini-PoW, the adversarial mining pool should publish its hidden Bitcoin block; otherwise, it will lose a significant deposit. The adversarial mining pool deposits three normalized bribes of values $\texttt{br}_1$, $\texttt{br}_2$, and $\texttt{br}_3$ in the smart contract. If a user submits a mini-PoW and the adversarial mining pool does not reveal the hidden block within a short time (indicating dishonesty of the claim of having a Bitcoin block), the user is eligible to withdraw $\texttt{br}_1$. If a user submits a mini-PoW and the adversarial mining pool immediately reveals its hidden Bitcoin block (indicating honesty), the user is eligible to withdraw $\texttt{br}_2$. If a user mines a supporting block on top of the adversarial block during the fork race, the user is eligible to withdraw $\texttt{br}_3$.

The smart contract stores two sets of parameters. The first set relates to the Bitcoin chain and includes: the hash of the head of the canonical chain from the perspective of petty-compliant mining pools, denoted by \(H_1\); the Bitcoin mining difficulty \(D_1\); and a commitment to the hash of the adversarial hidden block, denoted as \(\texttt{Hash}(H_\mathcal{A})\). The second set of parameters is used to construct a mini-PoW puzzle and includes: a random value serving as a parent block hash \(H_2\); a random value serving as a block Merkle root, denoted as \(\texttt{MR}_2\); and a mining difficulty \(D_2\). The values \(H_2\), \(\texttt{MR}_2\), and \(D_2\) are selected randomly to prevent mining pools from using Bitcoin block PoW solutions as a valid mini-PoW solution.
The smart contract satisfies the following conditions:
\begin{itemize}
    \item If a user submits $\texttt{nonce}_2$ that satisfies the inequality $\texttt{Hash}\big( H_2, \texttt{MR}_2, \texttt{nonce}_2\big) \leq \frac{1}{D_2}$ (successfully mines a mini-PoW), a time lock is activated. If, before the expiration of the time lock, the smart contract deployer submits $\texttt{nonce}_1$ and a Merkle root $\texttt{MR}_1$ that satisfies the inequality $\texttt{Hash}\big( H_1, \texttt{MR}_1, \texttt{nonce}_1\big) \leq \frac{1}{D_1}$ (reveals the hidden block\footnote{In Appendix~\ref{appendix_reveal_transaction}, we discuss a smart contract that also forces revealing transactions.}), the user can withdraw $\texttt{br}_2$; otherwise, it can withdraw a larger bribe of $\texttt{br}_1$. 
    \item If a user submits a proof of mining a valid supporting block on top of the adversarial block during the fork race, they can withdraw $\texttt{br}_3$. The user should submit $H_3$, $\texttt{nonce}_3$, and $\texttt{MR}_3$ that satisfy the equality $\texttt{Hash}(H_3) = \texttt{Hash}(H_\mathcal{A})$ and the inequality $\texttt{Hash}\big( H_3, \texttt{MR}_3, \texttt{nonce}_3\big) \leq \frac{1}{D_1}$.
\end{itemize}
\begin{lemma} \label{lemma_lying}
    Assume the adversarial mining pool has no hidden block to reveal upon the submission of a valid mini-PoW. In an environment with incentivizing factor $\epsilon$, the petty-compliant mining pools are incentivized to mine for a mini-PoW rather than a Bitcoin block if $\texttt{br}_1 \ge \frac{D_2}{D_1} (1+\epsilon)$.
\end{lemma}
\begin{proof}
    Let $\lambda_1$ denote the Bitcoin mining rate. In this case, the mini-PoW mining rate is equal to $\lambda_2=\frac{D_1}{D_2}$. The expected return of mining a Bitcoin block for mining pool $p_j$ is equal to $r_1 = \alpha_j \lambda_1 R$. The expected return of mining a mini-PoW for mining pool $p_j$ is equal to $r_2 = \alpha_j \lambda_2 \texttt{br}_1 R$. The mining pool $p_j$ is incentivized to mine a mini-PoW if $r_2 \ge r_1 + \epsilon \alpha_j \lambda_1 R$. We have:
    \begin{equation}
        r_2 \ge r_1 + \epsilon \alpha_j \lambda_1 R \Longleftrightarrow \frac{D_1}{D_2} \texttt{br}_1 \ge 1 + \epsilon \Longleftrightarrow 
        \texttt{br}_1 \ge \frac{D_2}{D_1} (1+\epsilon) \enspace .
    \end{equation}
\null\hfill\qedsymbol\end{proof}

\begin{lemma} \label{lemma_distraction_wasting_power}
    Let \(D_1\) denote the Bitcoin mining difficulty. Consider an adversarial mining pool with a mining share \(\alpha_\mathcal{A}\), which offers a normalized bribe \(\texttt{br}\) for mining each mini-PoW of difficulty \(D_2\). Assume \(k\) represents the average number of valid mini-PoW solutions mined per epoch during the distraction attack. If no adversarial block is orphaned, the time-averaged profit of the adversarial mining pool under the distraction attack exceeds its profit under the honest strategy for any value of \(k\), as long as \(\texttt{br} < \frac{D_2}{D_1} \alpha_\mathcal{A}\).
\end{lemma}
\begin{proof}
    The proof is similar to the proof of Theorem~\ref{theorem_bribery attack_long_term} presented in~\ref{appendix:bribery_theorem_proof}. Since \(k\) mini-PoW solutions are mined per epoch, an expected \(\frac{D_2}{D_1}k\) non-adversarial blocks are wasted per epoch, which must be replaced by new blocks. Of these new \(\frac{D_2}{D_1}k\) blocks, \(\alpha_\mathcal{A}\frac{D_2}{D_1}k\) are adversarial.
    Therefore, the time-averaged profit of the distraction attack can be obtained as follows:
    \begin{equation}
        \begin{split}
            & \texttt{profit}_\texttt{Distraction} = \frac{\texttt{Rev}-\texttt{Cost}}{\texttt{duration}} = \frac{\lambda\big(\alpha_\mathcal{A}(L+\frac{D_2}{D_1}k)R - k\texttt{br}R\big)}{L} =  \\
            & \lambda \alpha_\mathcal{A}R + \frac{\lambda k R}{L} (\frac{D_2}{D_1}\alpha_\mathcal{A}-\texttt{br}) \enspace.
        \end{split}
    \end{equation}
    The time-averaged profit of honest mining is equal to $\texttt{profit}_\mathcal{H}=\lambda \alpha_\mathcal{A}R$. The distraction attack can strongly dominate the honest strategy if the following inequalities hold: 
    \begin{equation}
        \begin{split}
            &\texttt{profit}_\texttt{Distraction} > \texttt{profit}_\mathcal{H} \Longleftrightarrow \lambda \alpha_\mathcal{A}R + \frac{\lambda k R}{L} (\frac{D_2}{D_1}\alpha_\mathcal{A}-\texttt{br}) > \lambda \alpha_\mathcal{A}R \\
            &\Longleftrightarrow \frac{D_2}{D_1}\alpha_\mathcal{A} > \texttt{br} \enspace.
        \end{split}
    \end{equation}
\null\hfill\qedsymbol\end{proof}
According to Lemma~\ref{lemma_lying}, if the adversary has no hidden Bitcoin block to reveal, it must set the bribe for each mini-PoW to \(\texttt{br}_1 = \frac{D_2}{D_1} (1+\epsilon)\) to incentivize miners to comply with the distraction attack. However, Lemma~\ref{lemma_distraction_wasting_power} shows that for the distraction attack to be profitable for the adversary, the maximum bribe the adversary can offer per mini-PoW is \(\frac{D_2}{D_1} \alpha_\mathcal{A}\), which is less than \(\texttt{br}_1\). Therefore, as expected, if the adversary has no hidden block, the distraction attack is not profitable for him. For the distraction attack to be profitable, upon submission of a valid mini-PoW, the adversary must be able to reveal a valid Bitcoin block. In this case, rather than offering a large bribe of \(\texttt{br}_1 = \frac{D_2}{D_1} (1+\epsilon)\), each mini-PoW receives a reward \(\texttt{br}_2\), which, according to Lemma~\ref{lemma_distraction_wasting_power}, should be set to \(\texttt{br}_2 < \frac{D_2}{D_1} \alpha_\mathcal{A}\).

To demonstrate the feasibility of the distraction attack, we need to show that even if the adversary has a hidden Bitcoin block, petty-compliant pools are still incentivized to mine mini-PoWs with a reward \(\texttt{br}_2 < \frac{D_2}{D_1} \alpha_\mathcal{A}\) instead of Bitcoin mining. To do this, we must compare the expected returns of Bitcoin mining and mini-PoW mining for a petty-compliant mining pool \(p_i\). 

\subsection{Expected Return Calculations}\label{appendix: expected_return_distraction}
Let \(r_\texttt{mini-PoW}\) and \(r_\texttt{Bitcoin}\) denote the expected returns of petty-compliant pool \(p_i\) for mini-PoW mining and Bitcoin mining, respectively, under the assumption that the adversary has a hidden block to reveal.
In an $\epsilon$-semi-rational environment, the petty-compliant mining pool $p_i$ is incentivized to mine a mini-PoW if the following inequality holds:
$r_\texttt{mini-PoW} \ge r_\texttt{Bitcoin} + \epsilon \alpha_i \lambda R$.
We define the normalized expected return difference as 
\begin{equation}
    \Delta = \frac{r_\texttt{mini-PoW}-r_\texttt{Bitcoin}}{\alpha_i \lambda R} \enspace .
\end{equation}
If $\Delta \ge \epsilon$, mining a mini-PoW is the dominant strategy for mining pool $p_i$. We also define the difficulty ratio \(d\) as the ratio of Bitcoin difficulty to mini-PoW difficulty, i.e., \(d = \frac{D_1}{D_2}\), where by construction $d\ge1$.
We consider $4$ different mining pools: the adversarial mining pool with mining share $\alpha_\mathcal{A}$, the petty-compliant mining pool $p_i$ with mining power $\alpha_i$ that wants to decide which puzzle is more profitable, the set of compliant mining pools with total mining share of $\alpha_C$ that always mines a mini-PoW (if available), and the set of non-compliant mining pools with total mining share of $\alpha_{NC}$ that always mines a Bitcoin block. Note that, under this attack, there is a possibility that some mining pools may decide to stop mining. We do not consider this scenario, as their decision to stop mining would signify the attack success.

There are $3$ different states that can take place: state $s_0$ where all the mining pools mine on top of the Bitcoin longest chain, state $s_1$ where the adversarial mining pool has a hidden Bitcoin block and the hash distraction contract is published, and state $s_2$ where there is a fork race between the adversarial block and a semi-rational block. According to the attack description, there is a bribe of $\texttt{br}_3 = \epsilon$ on top of the adversarial fork in state $s_2$. The state transitions are depicted in Figures~\ref{fig:first_scenario} and~\ref{fig:second_scenario}. For the sake of calculation simplicity, we assume if the adversarial mining pool mines a block in state $s_1$ (mines $2$ consecutive blocks), it will publish its first hidden block and repeat the attack with its second block. The petty-compliant pool $p_i$ needs to decide on its action in state $s_1$: mine a Bitcoin PoW or mine a mini-PoW.
\begin{figure*}[b]
    \centering
    \begin{subfigure}[!hb]{0.49\textwidth}
        \centering
        \includegraphics[height=1.3in]{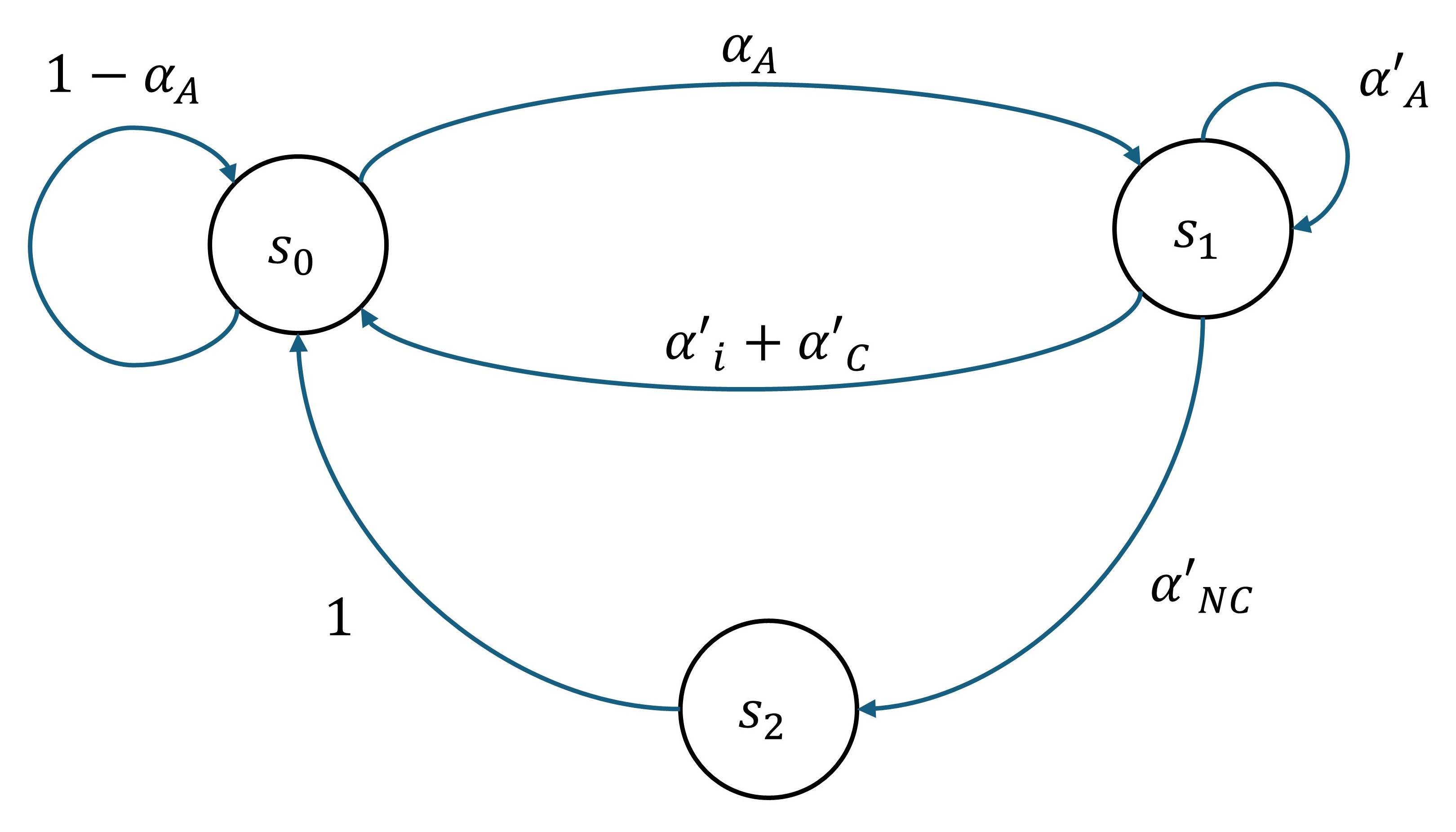}
        \caption{Mining a mini-PoW}
        \label{fig:first_scenario}
    \end{subfigure}%
    ~ 
    \begin{subfigure}[!hb]{0.5\textwidth}
        \centering
        \includegraphics[height=1.3in]{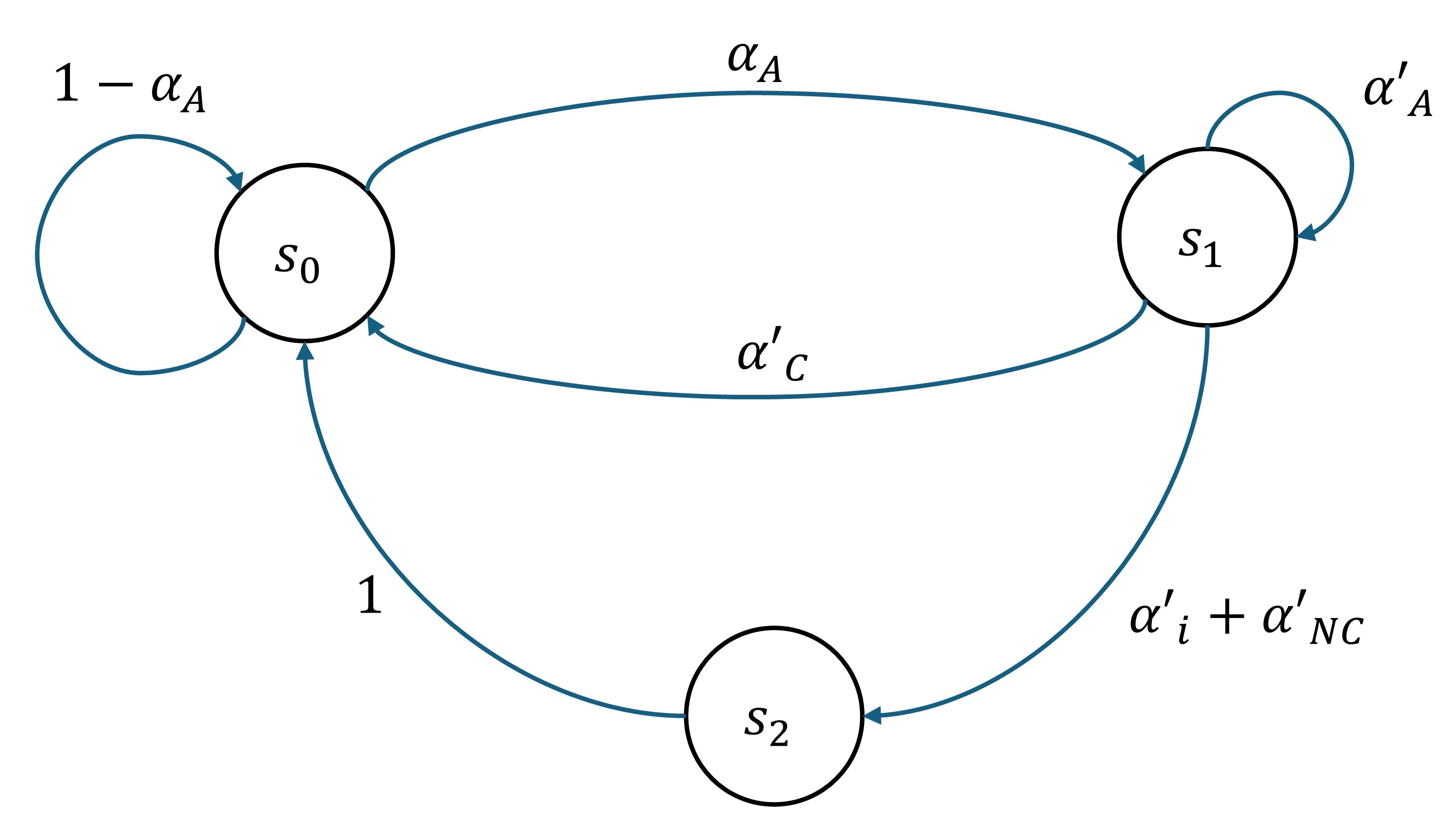}
        \caption{Mining a Bitcoin block}
        \label{fig:second_scenario}
    \end{subfigure}
\end{figure*}

In the first scenario, we consider the case that the petty-compliant pool $p_i$ mines a mini-PoW. The puzzle mining rate at states $s_0$ and $s_2$ are the same as $\lambda$. However, the puzzle mining rate at state $s_1$ is equal to $\lambda_1 = \big(d(\alpha_C+\alpha_i) + \alpha_{NC}+\alpha_\mathcal{A}\big)\lambda$. The transition probabilities depicted in Figure~\ref{fig:first_scenario} can be obtained as follows:
\begin{equation}
    \begin{split}
        & \alpha'_i = \frac{d\alpha_i}{d(\alpha_C+\alpha_i) + \alpha_{NC}+\alpha_\mathcal{A}} \enspace,\enspace  \alpha'_C = \frac{d\alpha_C}{d(\alpha_C+\alpha_i) + \alpha_{NC}+\alpha_\mathcal{A}} \enspace, \\
        & \alpha'_\mathcal{A} = \frac{\alpha_\mathcal{A}}{d(\alpha_C+\alpha_i) + \alpha_{NC}+\alpha_\mathcal{A}} \enspace,\enspace  \alpha'_{NC} = \frac{\alpha_{NC}}{d(\alpha_C+\alpha_i) + \alpha_{NC}+\alpha_\mathcal{A}} \enspace.
    \end{split}
\end{equation}
The state probabilities can be calculated as follows: 
\begin{equation}
    \begin{split}
        & p_0 = \frac{1-\alpha'_\mathcal{A}}{1-\alpha'_\mathcal{A}+\alpha_\mathcal{A}(1+ \alpha'_{NC})} \enspace, \enspace
          p_1 = \frac{\alpha_\mathcal{A}}{1-\alpha'_\mathcal{A}+\alpha_\mathcal{A}(1+ \alpha'_{NC})} \enspace, \\
        & p_2 = \frac{\alpha_\mathcal{A} \alpha'_{NC} }{1-\alpha'_\mathcal{A}+\alpha_\mathcal{A}(1+ \alpha'_{NC})} \enspace. 
    \end{split}
\end{equation}
The expected return of mining a min-PoW for pool $p_i$ is equal to
\begin{equation}
    r_1 = \alpha_i p_0 \lambda R + \alpha_i p_2 (1+\texttt{br}_3) \lambda R + \alpha'_i p_1 \texttt{br}_2 \lambda_1 R \enspace .
\end{equation}

In the second scenario, we consider the case that the petty-compliant pool $p_i$ mines a Bitcoin block. The puzzle mining rate at states $s_0$ and $s_2$ are the same as $\lambda$. However, the puzzle mining rate at state $s_1$ is equal to $\lambda_1 = \big(d\alpha_C+\alpha_i + \alpha_{NC}+\alpha_\mathcal{A}\big)\lambda$. The transition probabilities depicted in Figure~\ref{fig:second_scenario} can be obtained as follows:
\begin{equation}
    \begin{split}
        & \alpha'_i = \frac{\alpha_i}{d\alpha_C+\alpha_i + \alpha_{NC}+\alpha_\mathcal{A}} \enspace,\enspace  \alpha'_C = \frac{d\alpha_C}{d\alpha_C+\alpha_i + \alpha_{NC}+\alpha_\mathcal{A}} \enspace, \\
        & \alpha'_\mathcal{A} = \frac{\alpha_\mathcal{A}}{d\alpha_C+\alpha_i + \alpha_{NC}+\alpha_\mathcal{A}} \enspace,\enspace  \alpha'_{NC} = \frac{\alpha_{NC}}{d\alpha_C+\alpha_i + \alpha_{NC}+\alpha_\mathcal{A}} \enspace.
    \end{split}
\end{equation}
The state probabilities can be calculated as follows: 
\begin{equation}
    \begin{split}
        & p_0 = \frac{1-\alpha'_\mathcal{A}}{1-\alpha'_\mathcal{A}+\alpha_\mathcal{A}(1+ \alpha'_{NC}+\alpha'_i)} \enspace, \enspace
          p_1 = \frac{\alpha_\mathcal{A}}{1-\alpha'_\mathcal{A}+\alpha_\mathcal{A}(1+ \alpha'_{NC}+\alpha'_i)} \enspace, \\
        & p_2 = \frac{\alpha_\mathcal{A} (\alpha'_{NC}+\alpha'_i) }{1-\alpha'_\mathcal{A}+\alpha_\mathcal{A}(1+ \alpha'_{NC}+\alpha'_i)} \enspace. 
    \end{split}
\end{equation}
The expected return of mining a Bitcoin block for pool $p_i$ is equal to
\begin{equation}
    r_2 = \alpha_i p_0 \lambda R + \alpha_i p_2 (\frac{\alpha'_{NC}}{\alpha'_{NC}+\alpha'_i})(1+\texttt{br}_3)\lambda R + 2 \alpha_i p_2 (\frac{\alpha'_i}{\alpha'_{NC}+\alpha'_i})\lambda R \enspace .
\end{equation}

\subsection{Quantitative Analysis of the Distraction Attack}
For an adversarial mining power share of $\alpha_\mathcal{A}=0.4$ and the mini-PoW mining reward of $\texttt{br}_2 = 0.1 \times \alpha_\mathcal{A}$, Figure~\ref{fig:delta_distraction} represents the normalized expected return difference $\Delta$ of a petty-compliant pool based on its mining share, for different values of the difficulty ratio $d$. As can be seen, for an incentivizing factor of $\epsilon=0.02$, all petty-compliant mining pools prefer mini-PoW mining over Bitcoin mining if the difficulty ratio is set to $d \geq 5$. Note that, as the mini-PoW mining reward is set to $\frac{1}{10}$ of the adversarial mining share, according to Lemma~\ref{lemma_distraction_wasting_power}, the maximum difficulty ratio that can result in a profitable distraction attack for the adversary is $d=10$. For the same mini-PoW reward, the smaller the difficulty ratio, the more profitable the distraction attack would be for the adversary. However, setting the difficulty ratio to a low value limits the portion of petty-compliant miners who follow mini-PoW mining.

Figure~\ref{fig:percentage_increase_distraction} depicts the percentage increase in the adversarial reward share as a function of the mining share of the largest non-adversarial mining pool in the network, for a mini-PoW mining reward of $\texttt{br}_2 = 0.1 \times \alpha_\mathcal{A}$ and $\epsilon=0$. To generate Figure~\ref{fig:percentage_increase_distraction}, we calculated the minimum difficulty ratio that makes mini-PoW mining the dominant strategy for all petty-compliant mining pools. Using this difficulty ratio, we then determined the corresponding adversarial reward share.
\begin{figure*}[t] 
    \centering
    \begin{subfigure}[t]{0.49\textwidth} 
        \centering
        \includegraphics[width=\linewidth]{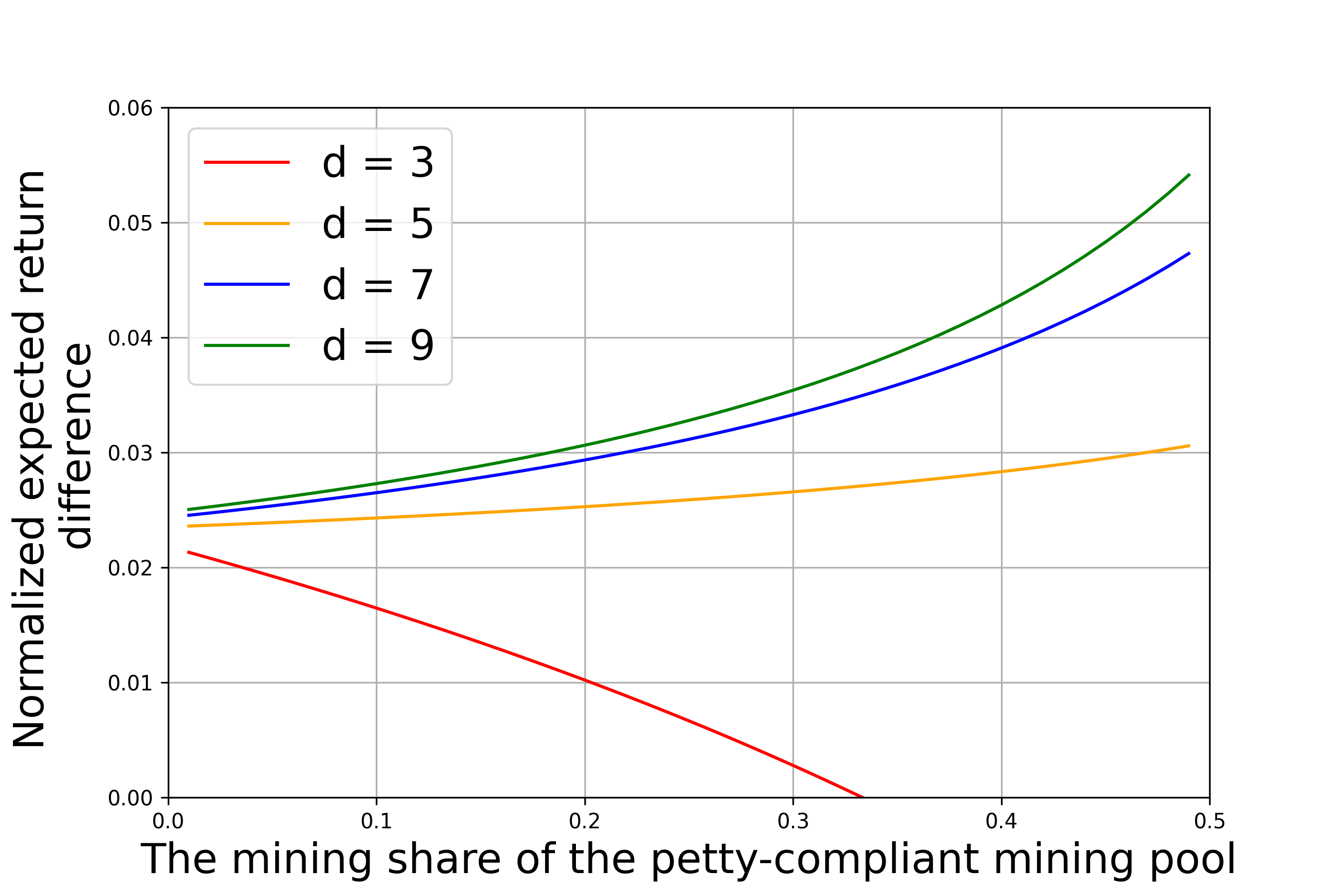}
        \caption{Mini-PoW mining vs Bitcoin mining.}
        \label{fig:delta_distraction}
    \end{subfigure}%
    \hfill 
    \begin{subfigure}[t]{0.49\textwidth}
        \centering
        \includegraphics[width=\linewidth]{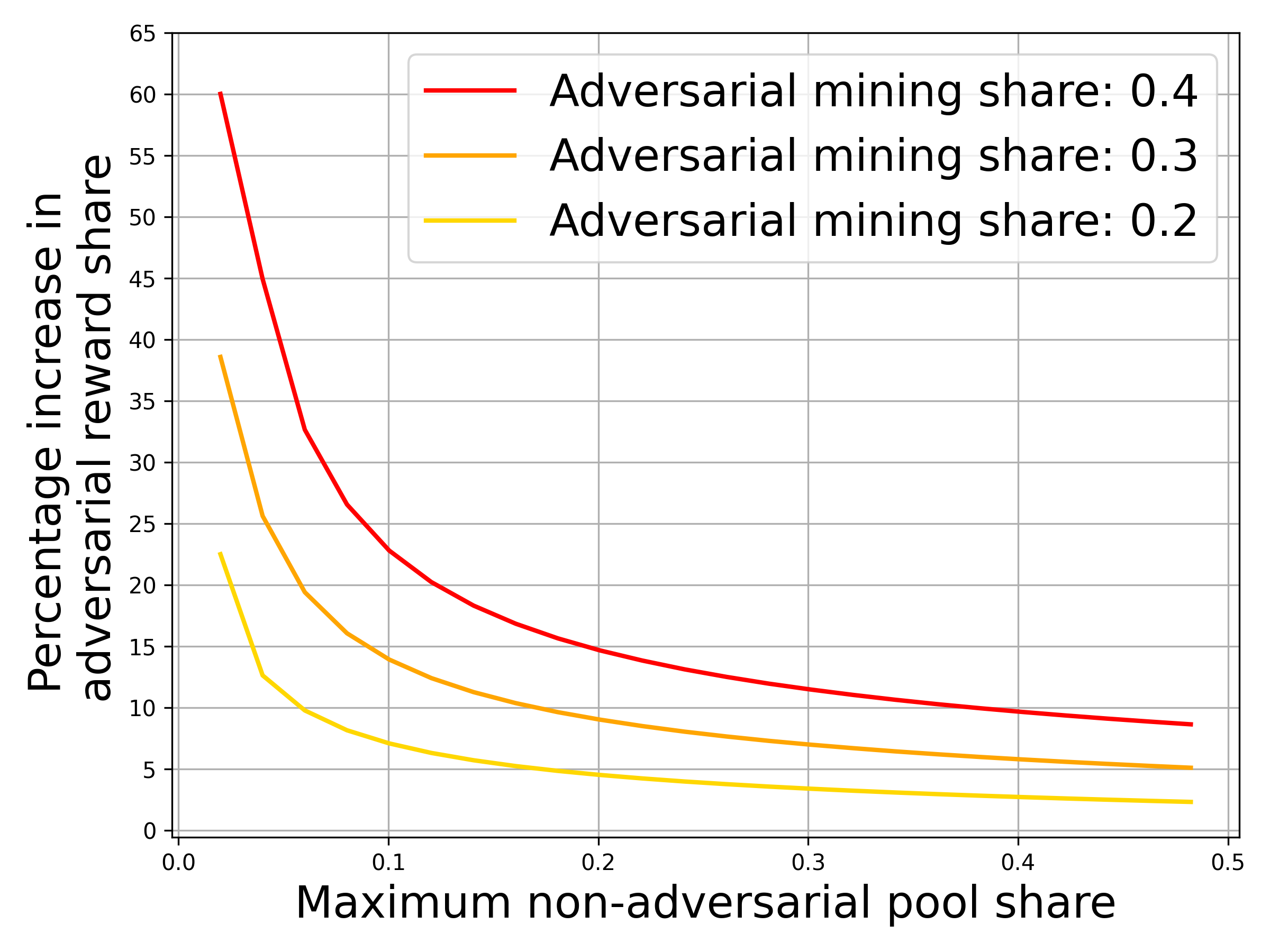} 
        \caption{Percentage increase in reward share.}
        \label{fig:percentage_increase_distraction}
    \end{subfigure}%
    \vspace{-10 pt}
    \caption{Quantitative analysis of the distraction attack.}
\end{figure*}
In Appendix~\ref{appnedix:extended_distraction attack}, we discuss how the distraction attack can be further extended to be applicable using non-adversarial blocks.

\subsection{BDoS attack Description}\label{Appendix:BDoS}
Let $B_0$ denote the tip of the canonical chain. In the BDoS attack, when the adversary mines a new block $B_1$ on top of $B_0$, it only publishes the header of block $B_1$. As a result, other miners are unable to mine on top of $B_1$ since they cannot confirm the correctness of the transactions included in it. This situation places miners in a dilemma: whether to stop mining or continue mining on top of block $B_0$. Mining on top of $B_0$ is risky; if they successfully mine another block on top of it, that block will enter a fork race with block $B_1$, potentially leading to it being orphaned. Under a BDoS attack, a victim pool may decide to stop mining, depending on the mining costs incurred. If the mining reward is significantly higher than the mining cost, the miner will likely continue mining on top of block $B_0$. The BDoS attack primarily affects mining pools whose mining revenue is only slightly above their mining costs, making it ineffective against miners with a higher profit-to-loss ratio. In contrast, the introduced distraction attack can disincentivize mining on Bitcoin, even when miners assume that there are no mining costs. In our distraction attack, once the adversarial mining pool announces mining of the next block, other mining pools face three possibilities: 1) continue mining on top of the Bitcoin chain, 2) completely stop mining, or 3) engage in a mini-PoW mining. The first two possibilities are those covered in the BDoS attack. The third possibility, which is new to the distraction attack, expands the effect of the attack to a broader range of mining pools rather than only those suffering from high mining costs.

Another point regarding the BDoS attack is that the adversarial mining pool always faces a fork race, which can result in the orphaning of its block. However, in our distraction attack, by publishing the block after solving a mini-PoW puzzle, the adversary can avoid the fork race and save its block while still conducting a successful mining power destruction attack.

\subsection{Revealing Transactions}\label{appendix_reveal_transaction}
To ensure the adversary reveals block transactions, the smart contract must provide a URL to all transactions (excluding the coinbase transaction), making them visible to all pools. The adversarial pool must reveal the coinbase transaction, which includes an extra nonce, upon submitting a valid mini-PoW. The contract should also include parameters to verify that the Merkle root of the transactions and the coinbase transaction matches the stored Merkle root. This smart contract allows petty-compliant mining pools to verify the correctness of block transactions before starting mini-PoW mining.



\subsection{Distraction Attack Using the Non-Adversarial Blocks}\label{appnedix:extended_distraction attack}
In Section~\ref{sec:distraction}, we analyzed the mining power distraction attack conducted by an adversarial mining pool using its own blocks. 
In the extended attack, the adversarial mining pool incentivizes a petty-compliant mining pool to only publish its block to the adversarial mining pool while hiding it from the other petty-compliant mining pools. The block gets revealed to other petty-compliant mining pools only upon receiving a valid mini-PoW.

\paragraph{Attack Description} The adversarial mining pool requests a petty-compliant mining pool to only share its block $B$ with the adversarial mining pool. Once the adversarial pool verifies the validity of block $B$, it deploys the same contract as mentioned in Section~\ref{sec:distraction} for block $B$, where block $B$ is hidden from all petty-compliant mining pools except for its miner. In addition to the $3$ bribes discussed in Section~\ref{sec:distraction}, the adversarial mining pool deposits another normalized bribe $\texttt{br}_4 = \epsilon$ in the smart contract. The miner of block $B$ is eligible to withdraw $\texttt{br}_4$ if a user submits a valid mini-PoW.

Note that in this attack, there is no need for the miner of block $B$ and the adversarial mining pool to trust each other as the attack is risk-free for both of them. If the petty-compliant mining pool shares its block only with the adversarial mining pool and observes that the adversarial pool does not deploy the contract, it can promptly publish the block to all other mining pools. Similarly, if the adversarial mining pool deploys the contract for the block but then the block gets published by its miner, the adversarial mining pool does not lose anything as the bribe is payable only upon the submission of a mini-PoW.


\end{document}